\documentclass[11pt]{article}

\usepackage[dvipsnames]{xcolor}
\definecolor{Gred}{RGB}{219, 50, 54}
\definecolor{ToCgreen}{RGB}{0, 128, 0}

\usepackage[margin=1in]{geometry}
\usepackage[T1]{fontenc}

\usepackage{newtxtext,newtxmath}

\usepackage{babel}
\usepackage{iftex}
\ifPDFTeX
  \usepackage[spacing=true,kerning=true,babel=true,tracking=true]{microtype}
\else
  \usepackage[tracking=false,expansion=false]{microtype}
\fi

\usepackage{xspace}
\usepackage{makecell}
\usepackage{dsfont}

\usepackage{multirow}
\usepackage{amsmath,amsthm}
\usepackage{bm}
\usepackage{bbm}
\usepackage{textgreek}
\usepackage{mathtools}
\usepackage{enumitem}
\usepackage[numbers,comma,sort&compress]{natbib}
\usepackage{authblk}
\usepackage{graphicx}
\usepackage[font=small]{caption}
\usepackage[labelformat=simple]{subcaption}

\usepackage{float}
\usepackage[linesnumbered,ruled,vlined]{algorithm2e}
\SetKwInput{KwInput}{Input}
\SetKwInput{KwOutput}{Output}
\SetKwInOut{Promise}{Promise}
\SetKwInput{Goal}{Goal}
\SetKwProg{Fn}{function}{}{}
\SetKwFor{RepTimes}{repeat}{times}{}
\SetKwFunction{Ver}{Verify}
\SetKwFunction{Prep}{PrepareState}
\usepackage{algorithmic}
\renewcommand{\algorithmicrequire}{\textbf{Input: }}  
\renewcommand{\algorithmicensure}{\textbf{Output: }}
\usepackage{physics}
\usepackage{footnote}
\usepackage{xcolor}
\usepackage{bbm}
\usepackage{braket}
\usepackage[colorlinks]{hyperref}
\usepackage{cleveref}
\hypersetup{
      colorlinks=true,
  citecolor=ToCgreen,
  linkcolor=Sepia,
  filecolor=Gred,
  urlcolor=Gred
  }
\usepackage{placeins}
\usepackage{booktabs}

\newtheorem{theorem}{Theorem}
\newtheorem{definition}{Definition}
\newtheorem{lemma}{Lemma}

\newtheorem{proposition}{Proposition}

\newtheorem{corollary}{Corollary}
\newtheorem{remark}{Remark}

\newcommand{\algo}[1]{\hyperref[algo:#1]{Algorithm~\ref*{algo:#1}}}
\newcommand{\gl}{g.l.\xspace}
\newcommand{\params}{{\boldsymbol{\alpha}}}

\renewcommand{\arraystretch}{1.25}

\newcommand{\Z}{\mathbb{Z}}
\newcommand{\R}{\mathbb{R}}

\usepackage{textgreek}

\DeclareMathOperator{\poly}{poly}

\DeclareMathOperator{\diag}{diag}

\renewcommand{\tr}{\text{tr}}
\renewcommand{\Tr}{\text{tr}}

\newcommand{\0}{\mathbf{0}}
\renewcommand{\emptyset}{\varnothing}
\def\Tr{\tr}\def\:{\hbox{\bf:}}

\renewcommand{\tr}{\mathrm{tr}}
\renewcommand{\Tr}{\mathrm{tr}}

\renewcommand{\emptyset}{\varnothing}
\def\Tr{\operatorname{tr}}\def\:{\hbox{\bf:}}

\renewcommand{\epsilon}{\varepsilon}

\usepackage[normalem]{ulem}

\renewcommand{\epsilon}{\varepsilon}

\allowdisplaybreaks
\usepackage{bm}
\renewcommand{\vec}[1]{\boldsymbol{\mathbf{#1}}}

\begin{document}

\title{Out-of-control Hamiltonian Learning}
\author{Weiyuan Gong
\thanks{Harvard University. Email: \href{mailto:wgong@g.harvard.edu}{wgong@g.harvard.edu}.}
\qquad
Muzhou Ma
\thanks{Caltech. Email: \href{mailto:mma2@caltech.edu}{mma2@caltech.edu}.}
\qquad
Sitan Chen\textsuperscript{\#}
\thanks{Harvard University. Email: \href{mailto:sitan@seas.harvard.edu}{sitan@seas.harvard.edu}.}
\qquad
Jordan Cotler\textsuperscript{\#}
\thanks{Harvard University. Email: \href{mailto:jcotler@fas.harvard.edu}{jcotler@fas.harvard.edu}.}
\qquad
Hsin-Yuan Huang\textsuperscript{\#}
\thanks{Oratomic. Email: \href{mailto:hhuang@oratomic.com}{hhuang@oratomic.com}.}
}
\date{\today}
\maketitle
\begingroup
\renewcommand{\thefootnote}{}
\footnotetext{\textsuperscript{\#}Senior coauthors.}
\endgroup

\pagenumbering{gobble}

\begin{abstract}
Learning the Hamiltonian of a many-body system from its dynamics is a central task in quantum science, yet the algorithms with the strongest provable guarantees assume some level of quantum control---fast, arbitrary single-qubit gates interleaved with time evolution, and measurements in arbitrary bases---that is beyond the capabilities of near-term analog quantum simulators. 
Motivated by analog atom- and ion-based platforms, we study Hamiltonian learning under minimal access models.
\begin{itemize}[leftmargin=*]
\item \textbf{Uniform state preparation and measurements}: We first consider the setting where in every experiment, one can rotate each qubit to the \emph{same} state,  perform short-time evolution, and measure every qubit in the \emph{same} basis. Surprisingly, we show that for generic 2-local Hamiltonians on any interaction graph, all of the parameters can be reconstructed from such experiments.
\item \textbf{Computational basis state preparation and measurements}: We then consider a similarly constrained setting, but where state preparation and measurement are restricted to the computational basis. For nearest-neighbor Hamiltonians with only Pauli $X/Z$ interactions, a class which captures contemporary Rydberg atom platforms, we show that over 1D and 2D rectangular lattices, all of the parameters can be reconstructed from such experiments up to unavoidable gauges.
\end{itemize}
Our protocols introduce new techniques for solving structured polynomial systems over an \emph{extensive} number of parameters. Taken together, our results suggest that one can learn a great deal from the dynamics of quantum many-body systems even under the most stringent experimental constraints.
\end{abstract}

\clearpage
 
\setcounter{tocdepth}{2}
{\footnotesize \tableofcontents}

\newpage

\pagenumbering{arabic}

\section{Introduction}\label{sec:intro}

In this work we revisit the classic question of learning a quantum many-body Hamiltonian $H$ from its dynamics $e^{-iHt}$. This basic task underlies a wide range of applications across condensed-matter physics, quantum simulation, metrology, and the verification and calibration of quantum devices~\cite{eisert2020quantum,carrasco2021theoretical,georgescu2014quantum,giovannetti2006quantum,giovannetti2011advances}. Yet despite the explosion of recent progress towards ever more powerful and general algorithms with provable guarantees for this problem, essentially all existing protocols take for granted the ability to prepare and measure each qubit in a basis of one's own choosing.\footnote{A complementary line of work studies Hamiltonian identification when only one or a few sites of a spin chain are accessible~\cite{burgarth2009coupling,difranco2009hamiltonian,zhang2014quantum}. There the restriction is on \emph{which} qubits can be addressed; here every qubit is observed, but no qubit's basis can be chosen individually.} This is an assumption underlying not just the most sophisticated Heisenberg-limited protocols, which interleave time evolution with carefully designed control operations~\cite{HuangTongFangSu2023learning,bakshi2024structure,dutkiewicz2024advantage,hu2025ansatz,sinha2025improved}, but even much simpler protocols operating at the standard quantum limit, which merely perform short time evolution in between preparing and measuring every site in some choice of Pauli basis~\cite{li2020hamiltonian,bairey2019learning,stilck2024efficient}.

Yet individually addressable control on all qubits is unavailable, or at best costly, on many of the platforms where characterizing a many-body Hamiltonian is most pressing, namely programmable \emph{analog} quantum simulators, whose atom arrays now reach hundreds to thousands of qubits~\cite{ebadi2021quantum,scholl2021quantum,manetsch2025tweezer}.
In Rydberg atom arrays, a single laser sets a common Rabi frequency and detuning for every atom, and individually addressable gates are very limited~\cite{labuhn2016tunable,bernien2017probing,browaeys2020many,ebadi2021quantum,scholl2021quantum,semeghini2021probing,wurtz2023aquila,kornjavca2024large}. Analog trapped-ion simulators are similar: the native operations are global spin rotations and spin-spin interactions mediated by global laser beams~\cite{kim2010quantum,britton2012engineered,islam2013emergence,jurcevic2014quasiparticle,richerme2014nonlocal,bohnet2016quantum,zhang2017observation,monroe2021programmable,shapira2025programmable}. Learning the Hamiltonian of such devices at scale is by now an active experimental goal.

The experimentalist is, in this sense, \emph{out of control}: the fast, individually-addressed gates assumed by Hamiltonian learning protocols are often not part of the native toolbox. This motivates the question:

\begin{center}
    \emph{What can be learned about Hamiltonians using only the operations that these devices natively support?}
\end{center}

\noindent We formalize this with two different access models (see \Cref{sec:access_model} for formal definitions):
\begin{itemize}
    \item \textbf{Uniform control ($\hat a$-to-$\hat b$ model).} In every experiment, one prepares $\ket{0^n}$, applies the \emph{same} single-qubit rotation $U^{\otimes n}$ to each qubit, applies time evolution, applies a second uniform rotation $V^{\otimes n}$, possibly different from the first, and measures every qubit in the computational basis. Equivalently, every qubit is prepared along a common Bloch direction $\hat a$ and measured along a common direction $\hat b$.

    \item \textbf{$Z$-basis access ($Z$-to-$Z'$ model).} In every experiment, one prepares a computational basis state $\ket z$, applies time evolution, and measures in the computational basis. In this work we also study the even weaker \emph{all-zero} ($0$-to-$Z'$) model, in which $z$ must be $0^n$.
\end{itemize}
\noindent Neither model subsumes the other: uniform control allows an arbitrary common basis but no site resolution, whereas $Z$-basis access allows site-resolved bit flips at preparation. In both models no control is applied during the evolution, the evolution times are short and decrease only polylogarithmically in the target precision, and the data collected from the experiments are bitstrings.

\emph{A priori}, one might expect that under both models much of the Hamiltonian is simply inaccessible. In the $Z$-basis model, conjugating $H$ by any diagonal unitary, or replacing $H$ with $-H^*$, leaves every transition probability unchanged. The uniform-control model has no such exact symmetry, but it comes with a different challenge. Each experiment is specified by just two directions $\hat a,\hat b$, a four-parameter family, from which one hopes to recover every parameter of a 2-local Hamiltonian on an interaction graph $G=(V,E)$, that is, $3|V|+9|E|$ real numbers. Worse, the standard approach of linear response directly fails: the first-order (in $t$) data determine $H$ only up to a linear subspace of dimension at least $|E|$, which is insufficient to recover $3|V|+9|E|$ real numbers in the unknown Hamiltonian.

Our main results show that this pessimism is unwarranted. Under uniform control, for generic 2-local Hamiltonians on \emph{any} known interaction graph, including graphs of unbounded degree with no geometric locality, \emph{all} $3|V|+9|E|$ parameters can be recovered with no gauge ambiguity, by exploiting second- and third-order response in short-time evolutions. Under $Z$-basis access, for Hamiltonians with only $X/Z$ interactions on chains and rectangular lattices, the symmetries stated in the previous paragraph are the \emph{only} obstruction: all parameters can be learned modulo a discrete group of $2^{n+1}$ sign flips. Both protocols use a number of experiments and total evolution time scaling as $\widetilde O(\poly(n)/\epsilon^2)$, whose $1/\epsilon^2$ dependence we show is unavoidable in the absence of interleaved control, and for translation-invariant families the cost is independent of $n$ altogether. At a technical level, our protocols reduce the learning problem to highly structured systems of polynomial equations in the unknown coefficients. We prove that these systems identify the parameters modulo the unavoidable symmetries, and that they are well conditioned for generic (smoothed) instances, by exhibiting explicit algebraic certificates. We expect these tools to be useful for other learning problems under restricted access. Below, we describe these results in greater detail.

\section{Our results}\label{sec:results}

In this work we study learnability of 2-local Hamiltonians in the two restricted access models described above. Following prior work on Hamiltonian learning under experimental restrictions~\cite{chen2025quantum}, we consider \emph{generic} Hamiltonians in order to avoid pathological problem instances: the coefficients of the Hamiltonian are arbitrary up to a small Gaussian perturbation of scale $\varsigma$, and the unknown Hamiltonian we would like to learn is the perturbed Hamiltonian; we review this smoothed analysis model in \Cref{sec:smooth_analysis}. Our results are summarized in \Cref{fig:settings_results} and \Cref{tab:results}.

\begin{table}[!htbp]
\centering
\footnotesize
\begin{tabular}{|c|c|c|c|c|}
\hline
\textbf{Access model} & \textbf{Hamiltonian model} & \textbf{Gauges} & \makecell{ \\\textbf{Evolution time}\\} & \textbf{Classical post-proc.}\\
\hline
\makecell{\\Uniform control ($\hat a$-to-$\hat b$)\\ \Cref{thm:uniform}, \ref{thm:graph_ham_uniform_local}} & General graphs & \emph{No} gauge & $\widetilde O(\poly(n)\epsilon^{-2})$ & $\widetilde O(\poly(n)\epsilon^{-2})$\\
\hline
\makecell{$0$-to-$Z'$\\
\Cref{thm:RI_TI_IXZ}} & Isotropic TI NN IXZ & \makecell{\\$2$ signs\\($Z^{\otimes n}$, time reversal)} & $\widetilde O(\epsilon^{-2})$ & $\widetilde O(\epsilon^{-2})$ \\
\hline
\makecell{$Z$-to-$Z'$\\ \Cref{thm:TI_NN_z_to_zp}} & TI NN IXZ, 1D/2D/3D & \makecell{\\$2$ signs\\($Z^{\otimes n}$, time reversal)} & $\widetilde O(\epsilon^{-2})$ & $\widetilde O(\epsilon^{-2})$ \\
\hline
\makecell{$Z$-to-$Z'$\\\Cref{thm:TI_3_local}} & 1D TI $3$-\gl IXZ & \makecell{\\$2$ signs\\($Z^{\otimes n}$, time reversal)} & $\widetilde O(\epsilon^{-2})$ & $\widetilde O(\epsilon^{-2})$ \\
\hline
\makecell{$Z$-to-$Z'$\\ \Cref{thm:local_NN_z_to_zp}} & 1D/2D NN IXZ & \makecell{\\$2^{n+1}$ signs\\($Z^{S}$, time reversal)} & $\widetilde O(\poly(n)\epsilon^{-2})$ & $\widetilde O(\poly(n)\epsilon^{-2})$\\
\specialrule{1.5pt}{1pt}{1pt}
\makecell{Lower bound \\ \Cref{thm:sql_2local}} & 1D NN IXZ & --- & $\Omega(\epsilon^{-2})$ & ---\\
\hline
\end{tabular}
\caption{Summary of our results. Upper bounds hold in smoothed analysis framework and suppress polynomial factors in $R,1/\varsigma,1/\delta$, where $R$ bounds magnitude of Hamiltonian coefficients and $\varsigma$ is smoothing scale (see \Cref{sec:smooth_analysis}).}
\label{tab:results}
\end{table}

\begin{figure}[t!]
\centering
\includegraphics[width=0.99\textwidth]{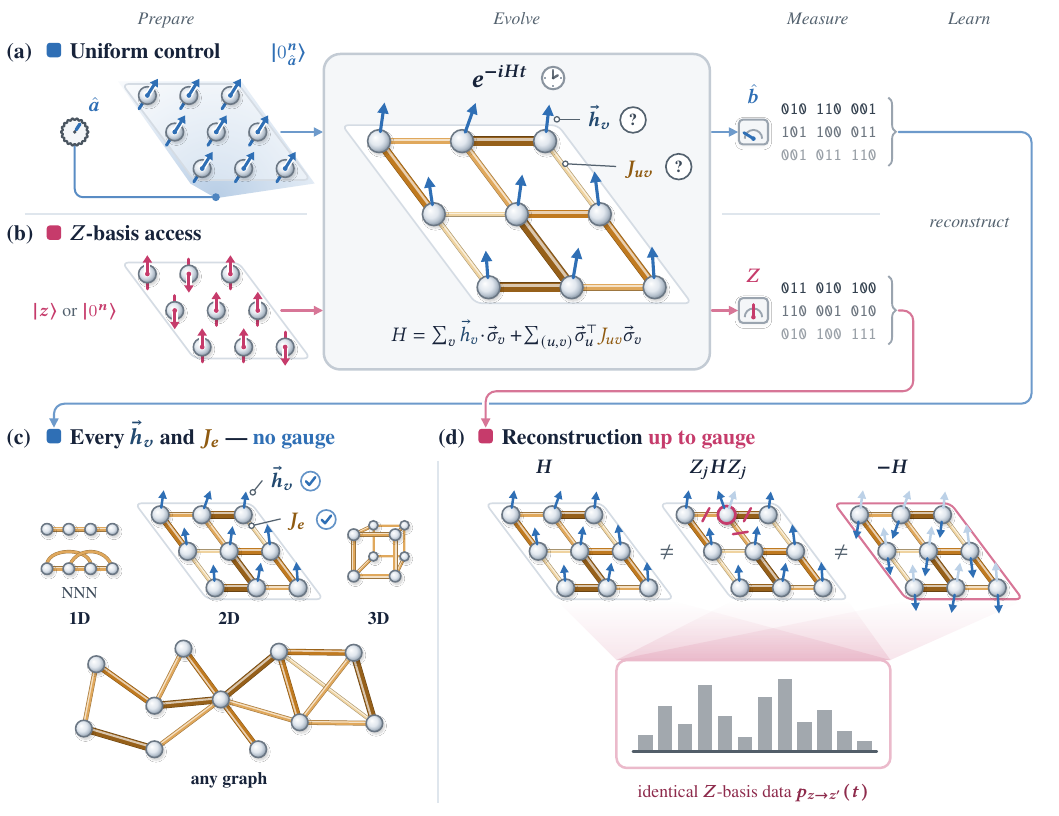}
\caption{
\textbf{(a)~Uniform control ($\hat a$-to-$\hat b$ model):} all qubits are prepared
along one common direction $\hat a$, $|0_{\hat a}^{n}\rangle=(U|0\rangle)^{\otimes n}$, and measured along one common direction $\hat b$.
\textbf{(b)~$Z$-basis access ($Z$-to-$Z'$ model; $0$-to-$Z'$ for $|0^n\rangle$ input):} the qubits start in a computational basis state $|z\rangle$  and are measured in the computational basis.
\textbf{(c)}~From the data of (a), all parameters can be learned, for \emph{any} interaction graph.
\textbf{(d)}~From the data of (b), a generic NN IXZ Hamiltonian on a 1D chain or square lattice can be learned up to unavoidable Pauli-$Z$ conjugation
and time-reversal gauges.}
\label{fig:settings_results}
\end{figure}

\subsection{Uniform control}

In the uniform control model, we will consider NN Hamiltonians on \emph{general} interaction graphs. Given a simple graph $G = (V,E)$ on $n = |V|$ qubits, a general NN Hamiltonian on $G$ can be written as
\begin{equation*}
    H = \sum_{v\in V} \vec h_v \cdot \vec\sigma_v + \sum_{ \{u,v\} \in E } \vec\sigma_u^{\top} J_{uv} \vec\sigma_v\,,
\end{equation*}
where $\vec\sigma_v = (X_v,Y_v,Z_v)^\top$ is the collection of single-site Pauli operators acting on qubit $v$, the vector $\vec h_v \in \mathbb{R}^3$ denotes the single-site fields at $v$, and the matrix $J_{uv}\in \mathbb{R}^{3\times 3}$ encodes the pairwise interactions between qubits $u,v$, with the convention $J_{vu}=J_{uv}^{\top}$, so that $H$ has $3|V|+9|E|$ real parameters in total. Our results apply to \emph{every} graph $G$, which we assume is known to the learner. No geometric locality is assumed.

Our main result in this section is that for generic instances, uniform control suffices to reconstruct all of the parameters of such a Hamiltonian.

\begin{theorem}[Informal, see Theorem~\ref{thm:graph_ham_uniform_local}]\label{thm:uniform}
    Let $H$ be a NN Hamiltonian on any simple graph $G = (V,E)$ which is known to the learner, with coefficients perturbed as in the smoothed analysis model with smoothing parameter $\varsigma$ (Definition~\ref{def:smoothed_hamiltonians}). There is a learning protocol that performs only experiments under the uniform control model and outputs estimates $(\{\widehat{\vec h}_v \}, \{\hat{J}_{uv}\})$, such that with probability at least $1 - \delta$,
    \begin{equation*}
        \max_{v\in V} \| \widehat{\vec h}_v - \vec h_v \|_2 \le \epsilon\,, \qquad \max_{\{u,v\} \in E} \|\hat{J}_{uv} - J_{uv}\|_F \le \epsilon\,.
    \end{equation*}
    The sample complexity, total evolution time, and classical post-processing time scale as
\begin{equation*}
    \widetilde{O}(\epsilon^{-2}\mathrm{poly}(n, R, 1/\delta, 1/\varsigma))\,.
\end{equation*}
The individual queried evolution times are of order $\widetilde O(1/((\Delta+1)R))$, where $\Delta$ is the maximum degree of $G$, and the protocol uses only marginals of the output bitstrings on at most four sites.
\end{theorem}

Several features of \Cref{thm:uniform} deserve emphasis. First, the result pins down every parameter, even though first-order data alone leave many parameters undetermined. The missing information is recovered from the connected second- and third-order response of at most four sites, so the protocol only ever uses short-time evolutions, $O(1)$-body marginals of the output bitstrings, and a fixed finite set of direction pairs $(\hat a,\hat b)$. Second, the result holds for every interaction graph and makes no use of geometric locality: the maximum degree $\Delta$ enters the cost only polynomially, and the evolution times scale as $1/((\Delta+1)R)$. Third, the classical post-processing is a sequence of $O(n+|E|)$ constant-size linear system solves. Smoothed analysis enters only to guarantee that these systems are well conditioned.

\subsection{$Z$-basis access}

Next, we turn to the setting where state preparation and measurement are done exclusively in the computational basis. Unfortunately, in this setting it is impossible, even generically, to recover all of the parameters of the Hamiltonian because of the following simple gauges:

\begin{proposition}[Unobservable gauges, informal, see \Cref{thm:unobservable_gauges}]
\label{prop:results_gauge}
Under $Z$-basis access, conjugation
\begin{equation*}
H\longmapsto U_ZHU_Z^\dagger
\end{equation*}
by any unitary $U_Z$ diagonal in the computational basis is unobservable. The transformation
\begin{equation*}
H\longmapsto-H^*
\end{equation*}
is also unobservable, where $H^*$ denotes entrywise complex conjugation. We refer to the latter as the time-reversal gauge.
\end{proposition}

\noindent Nevertheless, in this work we show that for a natural family of Hamiltonians and for generic couplings, these simple gauges are the \emph{only} symmetries that $Z$-basis access cannot resolve, and all other information about the Hamiltonian parameters can be learned modulo these symmetries. 

Specifically, we focus on \emph{IXZ} Hamiltonians, that is, Hamiltonians whose terms only consist of $X$ and $Z$ Paulis. This class includes the native Hamiltonians of the platforms motivating this work. The Rydberg Hamiltonian consists of a transverse field (the Rabi drive $\tfrac{\Omega}{2}\sum_i X_i$), a longitudinal field (the detuning term, proportional to $\sum_i n_i$ with $n_i=(I-Z_i)/2$), and $ZZ$ couplings (the van der Waals interaction $\sum_{i<j}V_{ij}\,n_in_j$)~\cite{browaeys2020many,wurtz2023aquila}, and the transverse-field Ising models realized in trapped-ion simulators have the same form~\cite{monroe2021programmable}. Both have power-law interaction tails; we treat their nearest-neighbor truncation and discuss the extension in \Cref{sec:outlook}. The IXZ class is considerably broader, allowing arbitrary spatially varying $X$, $Z$, $XX$, $XZ$, $ZX$, and $ZZ$ terms. For IXZ Hamiltonians $H$ is real, so the time-reversal gauge is the global sign flip $H\mapsto-H$. Among diagonal conjugations, those that map the family of generic NN IXZ Hamiltonians to itself are the Pauli-$Z$ conjugations $H\mapsto Z^SHZ^S$, which flip the sign of every term containing an odd number of $X$ factors on sites in $S$.

Our main result in this section is that for generic instances, $Z$-basis access suffices to reconstruct the parameters of NN IXZ Hamiltonians on chains and rectangular lattices, up to the above gauges.

\begin{theorem}[Informal, see \Cref{thm:local_NN_z_to_zp}]
\label{thm:zbasis}
    Consider a general NN IXZ Hamiltonian on an open or periodic 1D chain or an open or periodic 2D rectangular lattice. There is an algorithm using $Z$-basis access that, with probability at least $1-\delta$ in the smoothed-analysis framework, reconstructs all Hamiltonian coefficients to accuracy $\epsilon$ up to the Pauli-$Z$ conjugation and time-reversal gauges, reading out only constant-size regions of the output bitstrings. The number of experiments, total evolution time, and classical post-processing time are polynomial in $n$, $R$, $1/\varsigma$, and $1/\delta$, and scale quadratically in $1/\varepsilon$.
\end{theorem}

\paragraph{Extension 1: fewer gauges with translation-invariant Hamiltonians.}
The toolbox we develop in this setting is flexible, and we apply it to obtain learnability results in several other settings. First, we consider \emph{translation-invariant} (TI) NN IXZ Hamiltonians on a 1D chain, 2D square lattice, or 3D cubic lattice. Here, the set of parameters consists only of site-independent $(h_Z, h_X, J^\alpha_{ZZ}, J^\alpha_{XZ}, J^\alpha_{ZX}, J^\alpha_{XX})$, where $\alpha$ ranges over all standard basis vectors corresponding to directions in the lattice. Translation invariance also collapses the gauge group: among the Pauli-$Z$ conjugations, only $Z^{\otimes n}$ preserves translation invariance, so exactly two signs remain, a common sign on $(h_X,\{J^\alpha_{XZ},J^\alpha_{ZX}\}_\alpha)$ and a common sign on $(h_Z,\{J^\alpha_{ZZ},J^\alpha_{XX}\}_\alpha)$. We reconstruct the parameters up to these two signs using a constant number of constant-size experiments, with cost $\widetilde O(\epsilon^{-2}\poly(R,1/\varsigma,1/\delta))$ \emph{independent of $n$}, see \Cref{thm:TI_NN_z_to_zp} for details. We also show an analogous result for translation-invariant $3$-\gl IXZ Hamiltonians on a 1D chain, which may include next-nearest-neighbor and three-body terms on consecutive sites, see \Cref{thm:TI_3_local} for details.

\paragraph{Extension 2: all-zero input.} Finally, we consider the most stringent model, in which every experiment starts from $\ket{0^n}$, evolves, and measures in the computational basis, so that the only experimental knob is the evolution time. For \emph{isotropic} TI NN IXZ Hamiltonians on open 1D, 2D, and 3D lattices, whose couplings $(h_Z,h_X,J_{ZZ},J_{XX},J_{XZ}=J_{ZX})$ are the same in every lattice direction, we show that all five parameters are determined by second- and fourth-order data up to the same two signs as above (see Table~\ref{tab:1_2and3D_TI_2-local_gauge-configs}), again with cost independent of $n$; see Appendix~\ref{sec:RI_TI_IXZ} and \Cref{thm:RI_TI_IXZ} for details. The protocol makes use of sites at the boundary of the lattice, where translation symmetry is broken, to obtain enough independent equations.

\paragraph{No Heisenberg scaling without interleaved control.} Our protocols have total evolution time scaling as $1/\epsilon^2$, the standard quantum limit, whereas protocols that interleave the evolution with individually addressable qubit control achieve the Heisenberg limit $1/\epsilon$. This gap is inherent to the access models. Prior work has established that in the general setting of Hamiltonian learning from dynamics, any learning protocol whose experiments do not interleave time evolutions with multiple rounds of control cannot achieve Heisenberg scaling~\cite{dutkiewicz2024advantage}. In~\Cref{thm:sql_2local}, we observe that this no-go result carries over to an explicit one-parameter family of 1D NN IXZ Hamiltonians, and that it holds for every protocol in which each experiment is one uninterrupted evolution, even with arbitrary entangled inputs, ancillas, arbitrary measurements, and classical adaptivity. Hence the $O(1/\varepsilon^2)$ scaling of all our upper bounds is unavoidable under the access models that we consider.

\section{Overview of techniques}

The main statistic that we will estimate in all of the experiments that we consider is the \emph{transition probability}
\begin{equation*}
  p^{\hat{a},\hat{b}}_{z\to z'}(t) \triangleq |\bra{z_{\hat{a}}} e^{iHt} \ket{z'_{\hat{b}}}|^2\,,
\end{equation*}
where $\hat{a}, \hat{b}$ are Bloch vectors corresponding to the global rotations applied to the input state and prior to measurement respectively, and $z,z'\in\{0,1\}^n$ are computational basis strings.
Here $z_{\hat a}$ denotes the rotation of the computational basis state $\ket{z}$ into the basis specified by $\hat{a}$.

In the $Z$-basis access setting, we will omit the superscripts $\hat{a}, \hat{b}$ as everything is already in the computational basis.

\subsection{Uniform control}\label{sec:overview_uniform}

In the uniform control setting, we can only consider transition probabilities where $z = 0^n$. It will be convenient to work in the Fourier basis, regarding $p^{\hat{a},\hat{b}}_{0\to z'}(t)$ as a Boolean function of $z'$. Given $S\subseteq[n]$, its $S$-th Fourier coefficient is given by
\begin{equation}\label{eq:def_moment}
m^{\hat{a},\hat{b}}_S(t) \triangleq \sum_{z'\in\{0,1\}^n} \chi_S(z') p^{\hat{a},\hat{b}}_{0^n\to z'}(t) = \bra{0^n_{\hat{a}}} e^{iHt} \bigotimes_{v\in S} (\hat{b}\cdot \vec\sigma_v) e^{-iHt} \ket{0^n_{\hat{a}}}\,,
\end{equation}
where $\chi_S: \{0,1\}^n\to\{\pm 1\}$ denotes the $S$-th Fourier character.
Because $\bigotimes_{v\in S} (\hat{b}\cdot \vec\sigma_v)$ is bounded in operator norm, this observable can be directly estimated to error $\upsilon$ using $O(1/\upsilon^2)$ experiments.

As is standard among many standard quantum limited protocols for Hamiltonian learning~\cite{stilck2024efficient,zubida2021optimal,gu2024practical,haah2024learning}, we aim to extract information about the couplings from the Taylor expansion of such observables around $t = 0$, a trick commonly referred to as ``derivative estimation.''

We will expand the observables to \emph{third order}, obtaining
\begin{align*}
m_S^{\hat a,\hat b}(t)
=m_S^{\hat a,\hat b}(0)+tL_S^{\hat a,\hat b}
+t^2Q_S^{\hat a,\hat b}+t^3\mathcal{W}_S^{\hat a,\hat b}+O(t^4),
\end{align*}
for terms $L_S$, $Q_S$, and $\mathcal{W}_S$ which can be expressed in terms of certain nested commutators involving $H$. By polynomial interpolation of short-time data, we can extract estimates for each of these terms and obtain polynomial constraints on $H$. 

This is where our setting differs from conventional Hamiltonian learning. With site-resolved preparation and measurement, suitably chosen first-order observables can directly reveal the Hamiltonian coefficients. Those observables require the ability to prepare individual qubits in different Pauli eigenstates and measure them in different Pauli bases. Under uniform control, this information is unavailable, and in fact the first-order coefficients $L_S^{\hat a,\hat b}$ do not contain enough information to determine all of the Hamiltonian parameters.

We therefore use the higher-order coefficients $Q_S$ and $\mathcal{W}_S$ in addition to the first-order data $L_S$. A priori, these higher-order constraints are nonlinear and involve an extensive number of Hamiltonian parameters. The key step is to combine them so that the resulting equations become local and, after substituting information recovered at earlier stages, linear in the remaining unknown parameters.

\paragraph{Learning the pairwise couplings in stages.} We first decompose the coupling matrices $J_{uv}$ into a \emph{symmetric traceless part}, a \emph{skew-symmetric part}, and a multiple of the identity. Given a vector $\vec d = (d_1,d_2,d_3)\in\mathbb{R}^3$, define the skew-symmetric matrix $[\vec d]_\times$ by
\begin{equation*}
  [\vec d]_\times = \begin{pmatrix}
    0 & -d_3 & d_2 \\
    d_3 & 0 & -d_1 \\
    -d_2 & d_1 & 0
  \end{pmatrix}
\end{equation*}
so that for any $\vec y\in\mathbb{R}^3$, $\vec d\times \vec y = [\vec d]_\times \vec y$.
Then we can write
\begin{equation*}
  J_{uv} = T_{uv} + [\vec d_{uv}]_\times + \lambda_{uv} I\,,
\end{equation*}
where $T_{uv}$ is symmetric and traceless and $3\lambda_{uv} = \Tr(J_{uv})$. Recalling our convention that $J_{uv} = J_{vu}^\top$, we see that $T_{vu} = T_{uv}$, $\vec d_{vu} = -\vec d_{uv}$, and $\lambda_{vu} = \lambda_{uv}$.

Each component of $J_{uv}$ acts on the input state $\ket{0^n_{\hat{a}}}$ in a particular way. To see this, we say that an operator $O$ can \emph{flip} a set of qubits $S$ if $O\ket{0^n_{\hat{a}}}$ has nonzero amplitude on strings in the $\hat{a}$-basis which are $\ket{1_{\hat{a}}}^{\otimes |S|}$ in the qubits indexed by $S$. For example, a local field $\vec h_v$ can flip qubit $v$. On the other hand, the different components of ${\vec\sigma}^\top_u J_{uv} \vec\sigma_v$ above can flip anywhere between zero and two qubits. For the $\lambda_{uv}I$ component, note that 
\begin{equation*}
  \vec\sigma_u \cdot \vec\sigma_v = 2\mathrm{SWAP}_{uv} - I\,,
\end{equation*}
which simply fixes $\ket{0_{\hat{a}}}_u \otimes \ket{0_{\hat{a}}}_v$ and flips neither qubit. For the skew-symmetric component, note that for any single-qubit state $\ket{\psi}$,
\begin{equation*}
  \vec\sigma_u^\top [\vec d_{uv}]_\times \vec\sigma_v \ket{\psi}^{\otimes n} = -i \vec d_{uv}\cdot (\vec\sigma_u - \vec\sigma_v) \ket{\psi}^{\otimes n}\,.
\end{equation*}
In other words, the skew-symmetric part acts on $\ket{0^n_{\hat{a}}}$ as a sum of 1-local terms, so it can only flip one qubit at a time. Finally, the symmetric traceless component can possibly flip two qubits at once. 

This distinction between the three components will be crucial later, as it will dictate the order at which they become visible: roughly speaking, the less that a component can flip, the higher the order at which it first becomes detectable.
Motivated by this, our protocol will reconstruct these three pieces in rounds. In the first round, it will use first-order information to recover $T_{uv}$ and the local fields $\vec h_v$. In the second round, it will use second-order information to recover $\vec d_{uv}$. In the third round, it will use third-order information to recover $\lambda_{uv}$. \Cref{fig:overview_uniform} summarizes this sequence of steps. 

Throughout, we will only use input and output rotations $\hat{a}$ and $\hat{b}$ which are orthogonal to each other. Let 
\begin{equation*}
  \hat{k} \triangleq \hat{b}\times \hat{a}
\end{equation*}
denote one of the two vectors on the Bloch sphere orthogonal to $\hat{a},\hat{b}$.

\begin{figure}[t!]
\centering
\includegraphics[width=0.99\textwidth]{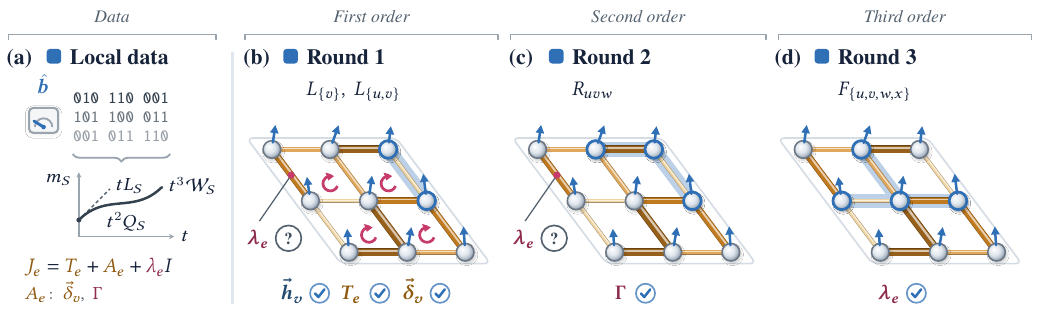}
\caption{\textbf{Three-round reconstruction under uniform control.}
 \textbf{(a)~Local data:} Taylor coefficients $L_S$, $Q_S$ and $\mathcal{W}_S$ (orders one to three) of the Fourier coefficients $m^{\hat a,\hat b}_S(t)$ on sets $S$ of at most four sites, with $J_e=T_e+A_e+\lambda_eI$ for $A_e=[\vec d_e]_\times$.
\textbf{(b)--(d):} checks denote what is learned in that round, crimson and question marks denote the remaining ambiguities.
\textbf{(b)~Round~1:} first-order data identifies $T_e$, $\vec h_v$, as well as $A_e$ up to a global curl gauge.
\textbf{(c)~Round~2:} Second-order terms $R_{uvw}$ of three-site cumulants determine the skew-symmetric components $A_e$ for edges on induced paths. For the remaining edges, we use observables adapted to the local graph structure (see \Cref{fig:graph_solve_curl}).
\textbf{(d)~Round~3:} Third-order terms $F_{\{u,v,w,x\}}$ of connected four-site cumulants determine the edge traces $\lambda_e$ for edges with an endpoint of degree at least three. Second-order single-site responses recover the remaining traces (see \Cref{fig:graph_trace}).}
\label{fig:overview_uniform}
\end{figure}

\paragraph{Round 1: recovering the symmetric part at first order.} 
We can compute the first-order terms to get
\begin{equation*}
  L_S^{\hat a,\hat b}=i\bigl\langle\bigl[H,\bigotimes_{v\in S}(\hat b\cdot\vec\sigma_v)\bigr]\bigr\rangle\,.
\end{equation*}
We will only consider $S$ given by a single edge or a single vertex. In these cases, we have for any $v\in V$ and any edge $\{u,v\}\in E$ that
\begin{equation*}
  L^{\hat{a},\hat{b}}_{\{u,v\}} = -4\hat{k}^\top T_{uv} \hat{b}\,.
\end{equation*}
\begin{equation*}
  L^{\hat{a},\hat{b}}_{\{v\}} = -2\Bigl\langle\hat{k}, \vec h_v + \sum_{w\sim v} J_{vw} \hat{a}\Bigr\rangle = -2\langle\hat{k}, \vec h_v\rangle - 2\hat{k}^\top \Bigl(\sum_{w\sim v} T_{vw}\Bigr) \hat{a} - \Bigl\langle \hat{b}, 2\sum_{w\sim v} \vec d_{vw}\Bigr\rangle\,.
\end{equation*}
By varying the orthonormal Bloch vectors $\hat{a}, \hat{b}, \hat{k}$, we can easily read off from the former equation the symmetric component $T_{uv}$ and then, with this in hand, we can also read off the local fields $\vec h_v$ from the latter equation.

In addition, we can read off $\sum_{w\sim v} \vec d_{vw}$. Because $\vec d_{vw} = -\vec d_{wv}$, this quantity has the illuminating interpretation as the \emph{(discrete) divergence} of the vector field $\vec d: E\to \mathbb{R}^3$ at $v$ (see e.g., the lecture notes of~\cite{mahoney2016lecture}), denoted $(\nabla \cdot \vec d)_v$. A vector field $\vec\Gamma: E\to\mathbb{R}^3$ for which $(\nabla \cdot \vec\Gamma)_v = 0$ for all $v\in V$ is called a \emph{circulation}, and we conclude that using the above information, we can only learn the skew-symmetric part of the couplings up to a global circulation.

To resolve this remaining ambiguity, as well as the trace terms $(\lambda_{uv})$, we consider the higher-order terms $Q_S$ and $\mathcal{W}_S$.

\paragraph{Round 2: recovering the skew-symmetric part at second order.} \emph{A priori}, the second order term $Q_S$ is cumbersome to deal with because it involves products of terms corresponding to pairs of connected edges, thus yielding quadratic constraints on the unknown parameters that are difficult to solve directly. It turns out however that one can use the information learned from the previous round to extract from these quadratic constraints a linear system that can then resolve the aforementioned circulation symmetry.

The general strategy is to carefully design a polynomial in the terms $L_S, Q_S$ which, upon substituting the information obtained in the first round, simplifies appropriately. There are two desiderata for this simplification: (1) the system should become ``localized'' in the sense that it only depends on unknown parameters in a small neighborhood of the graph, and (2) the system should ideally be linear so that the remaining skew-symmetric terms can be solved efficiently.

We achieve this by looking at a series expansion not of the Fourier coefficients $m^{\hat{a},\hat{b}}_S(t)$, but instead of the following cumulant quantity:
\begin{align*}
C_{uvw}^{\hat a,\hat b}(t)
&\triangleq
m_{\{u,v,w\}}^{\hat a,\hat b}(t)
-m_{\{u,v\}}^{\hat a,\hat b}(t)m_{\{w\}}^{\hat a,\hat b}(t)\nonumber\\
&\quad-m_{\{u,w\}}^{\hat a,\hat b}(t)m_{\{v\}}^{\hat a,\hat b}(t)
-m_{\{v,w\}}^{\hat a,\hat b}(t)m_{\{u\}}^{\hat a,\hat b}(t)\nonumber\\
&\quad+2m_{\{u\}}^{\hat a,\hat b}(t)m_{\{v\}}^{\hat a,\hat b}(t)m_{\{w\}}^{\hat a,\hat b}(t)\,.
\end{align*}
Whereas $m^{\hat{a},\hat{b}}_{uvw}(t)$ can be regarded as the joint moment of the three signs in $\{\pm 1\}$ obtained by measuring qubits $u, v, w$ at the end of the experiment each along Bloch vector $\hat{b}$, $C_{uvw}^{\hat a,\hat{b}}(t)$ is the joint \emph{cumulant} of those signs. In particular, it vanishes whenever there is at least one sign among the three which is statistically independent of the other two. Of course, unless the graph has multiple connected components, such independence will not hold exactly, but only weakly, and by inspecting the Taylor expansion of this cumulant in $t$, we can isolate dominant sources of dependence among these signs. Writing the Hamiltonian as a sum of local terms $H=\sum_j g_jH_j$, the $k$th Taylor coefficient of the cumulant is a polynomial $p_k(g)$ in the coefficients $g_j$. A monomial $g_{j_1}\cdots g_{j_k}$ can contribute only if the corresponding Hamiltonian terms collectively connect the three sites $u,v,w$. In particular, if there is a graph cut separating one of the three sites from the other two and none of the terms $H_{j_1},\ldots,H_{j_k}$ crosses that cut, then the monomial does not appear in $p_k$. In this way, we see that $p_0 = p_1 = 0$, and $p_2$ only consists of monomials $g_j g_k$ where terms $j$ and $k$ correspond to distinct edges connecting $\{u,v,w\}$. This is the key mechanism allowing us to ``localize'' to a constraint involving only terms in a small neighborhood.

In particular, in the special case where $u,v,w$ belong in a path but not a triangle, we conclude that the leading quadratic term in the Taylor expansion of $C_{uvw}^{\hat{a},\hat{b}}(t)$ must be a bilinear function of $J_{uv}$ and $J_{vw}$. Denote this quadratic term by
\begin{equation*}
R_{uvw}^{\hat a,\hat b}\triangleq[t^2]\,C_{uvw}^{\hat a,\hat b}(t)\,.
\end{equation*}
Crucially, this bilinear form vanishes when either of the arguments is a multiple of the identity or both are skew-symmetric. The reason, roughly speaking, is that the only contributions to the bilinear form must come from flips of all three qubits $u,v,w$, which would not be possible under either of these two cases. We conclude that $R_{uvw}^{\hat{a},\hat{b}}$ is \emph{linear} in the remaining unknowns $\vec d_{uv}$ and $\vec d_{vw}$, and by varying $\hat{a},\hat{b}$, we can argue that there is enough information to reconstruct both.

This already suffices to recover the skew-symmetric component of the couplings for all edges connecting vertices which do not belong in an isolated clique. We handle this corner case separately, deferring the details to Appendix~\ref{sec:isolated_triangle}.

\paragraph{Round 3: recovering the isotropic part at third order.} It remains to recover the $\lambda_{uv}$ terms, which only become fully visible in a localized fashion at third order.

The overall structure of this round closely parallels that of Round 2, just one degree higher, so we only briefly describe the ideas that go into this round and defer the details to Appendix~\ref{sec:general_third}. Given a set of \emph{four} sites $S$, let $C_S^{\hat a, \hat b}(t)$ denote the joint cumulant of the four signs measured at those sites, in analogy with $C_{uvw}^{\hat{a}, \hat b}(t)$ from Round 2, and let $F_S^{\hat a,\hat b}\triangleq [t^3]\, C_S^{\hat a, \hat b}(t)$. The same graph cut argument from before proves that $F_S$ only has contributions coming from sets of three edges that connect all four sites in $S$. This rules out contributions from local fields $\vec h_v$ as well as from edges not connecting vertices of $S$. Some careful accounting shows that after substituting known terms that were computed in previous rounds, the resulting constraint is linear in the unknown $\lambda_{uv}$'s that connect sites in $S$. By varying $\hat{a},\hat{b}$, we can then solve for all $\lambda_{uv}$'s corresponding to edges containing an endpoint with degree at least $3$. This leaves some corner cases involving edges whose endpoints both have degree at most $2$, which we show how to handle in Appendix~\ref{sec:general_third}.

\paragraph{Propagation of errors from shot noise.} In the above discussion, we assumed that all of the low-order terms in the Taylor expansion of the various cumulant quantities considered were known exactly. In reality, these cumulants are estimated from finitely many experiments, and their Taylor coefficients are further approximated through polynomial interpolation. It still remains to argue that the linear systems that we formed in the preceding discussion are sufficiently well-conditioned that error in estimating these observables does not translate into large error in estimating the unknown parameters. Fortunately, the relevant systems are all of constant size and do not involve long chains of substitutions, and their coefficients are polynomials in the true parameters that were estimated in previous rounds of the algorithm. A standard argument via anti-concentration of Gaussian polynomials thus ensures that these systems have bounded condition number with high probability over the smoothing, proving that the aforementioned error propagation is benign.

\subsection{$Z$-basis access}
\label{sec:overview_zbasis}

Next, we turn to our second access model in which global rotations are not allowed, but the input can be prepared in any computational basis state. As discussed previously, this gives rise to some unavoidable gauges, namely Pauli-$Z$ conjugation and time reversal, which cannot be resolved through such experiments, but we will show that it is nevertheless possible to generically recover all of the parameters modulo these gauges. In this overview, we will focus on our results for general (non-translation-invariant) IXZ NN Hamiltonians.

As in the previous part, we will consider the transition probabilities $p_{z\to z'}(t) = |\bra{z'}e^{-iHt}\ket{z}|^2$, specifically Taylor expansions of them and of their Fourier transforms. Furthermore, unlike in the uniform control setting, here the information about the Hamiltonian is entirely invisible at first order, because Hamiltonians with only $X$ and $Z$ interactions are real symmetric matrices in the computational basis and thus the transition probabilities are even functions of $t$. In going to second order and above, however, we encounter systems of \emph{nonlinear} polynomial equations that cannot be reduced to linear systems as in the previous part.

\paragraph{Learning the diagonal and off-diagonal terms in stages.} In place of the decomposition of the pairwise couplings into symmetric and skew-symmetric parts, here we decompose the entire Hamiltonian into a \emph{diagonal} part and an \emph{off-diagonal} part:
\begin{equation*}
    H = D + K\,,
\end{equation*}
where $D$ consists of all $Z$-only terms, and $K$ consists of all terms with at least one $X$. Similar to the previous part, in this decomposition the different terms can flip different numbers of qubits (none for $D$, and at least one for $K$), so they naturally appear at different orders in the Taylor expansion of the transition probabilities. Indeed, Taylor expanding to fourth order, we have for any $z\neq z'$ that
\begin{equation}
    p_{z\to z'}(t) = t^2 \bra{z'}H\ket{z}^2 + t^4\Bigl(\frac{1}{4} \bra{z'}H^2\ket{z}^2 - \frac{1}{3}\bra{z'}H\ket{z} \cdot\bra{z'}H^3\ket{z}\Bigr) + O(t^6)\,.\label{eq:taylor_p_overview}
\end{equation}
In particular, at second order $\bra{z'}H\ket{z} = \bra{z'}K\ket{z}$ does not depend on $D$ at all, and it is thus necessary to look to the fourth order term to get information about $D$.

\paragraph{Extracting local polynomial systems from the transition probabilities.} As the polynomial constraints given by these second and fourth order terms in the Taylor expansion of $p_{z\to z'}(t)$ are over an extensive number of variables, we need a way to \emph{localize} them to a small neighborhood of variables as in the previous part. To do this, we consider low-degree terms in the Fourier decomposition of the Taylor coefficients. As an illustrative example, consider the case of a nearest-neighbor IXZ Hamiltonian on a 1D chain, and take $z'$ to be given by flipping the $i$-th bit of $z$. The only terms in the Hamiltonian that flip only the $i$-th bit are $h^i_X X_i, J^i_{XZ} X_i Z_{i+1}, J^{i-1}_{ZX} Z_{i-1}X_i$, yielding
\begin{equation*}
    \bra{z'}H\ket{z} = h^i_X + J^i_{XZ} (-1)^{z_{i+1}} + J^{i-1}_{ZX} (-1)^{z_{i-1}}\,.
\end{equation*}
Recalling that the second-order term in Eq.~\eqref{eq:taylor_p_overview} is $\bra{z'}H\ket{z}^2$, we see that the Fourier coefficients of this term allow us to read off the pairwise products $h^i_XJ^i_{XZ}$, $h^i_XJ^{i-1}_{ZX}$, and $J^i_{XZ}J^{i-1}_{ZX}$. Generically, these allow us to recover the triple $(h^i_X, J^i_{XZ}, J^{i-1}_{ZX})$ up to one common sign; this sign ambiguity comes from the gauge symmetry of conjugation by $Z_i$. In addition, at second order we can also look at $z'$ given by flipping \emph{two} adjacent bits of $z$. A similar calculation as above allows us to recover $|J^i_{XX}|$. 

We can resolve these sign ambiguities by using the fact that at fourth order, there are two ways to flip non-adjacent qubits $i-1$ and $i+1$, namely either by flipping them individually, or by flipping the pairs $\{i-1,i\}$ and $\{i,i+1\}$. This can be used to read off the sign of $J^{i-1}_{XX}J^i_{XX}$, for instance, and related arguments can be used to globally resolve the signs of all terms in $K$ up to the unavoidable gauges (Appendix~\ref{sec:local_z_common}).

It remains to recover the diagonal terms of the Hamiltonian, namely $h^i_Z$ and $J^i_{ZZ}$. From the fourth order term in Eq.~\eqref{eq:taylor_p_overview}, we can subtract off the analogous quantity where $H$ is replaced by its off-diagonal component, which is at this stage of the protocol completely known upon fixing a representative under the unavoidable gauges. If $z$ and $z'$ are given by neighboring strings, this subtraction removes contributions from transitions that flip bits outside of the neighborhood over which they disagree, allowing us to obtain a properly localized polynomial system. Specifically, for $E_z \triangleq \bra{z}D\ket{z}$, we get
\begin{align*}
[t^4]\bigl(p_{z\to z'}-p^K_{z\to z'}\bigr)
&=-\frac{1}{12}\bra{z'}K\ket{z}^2\,(E_z-E_{z'})^2\nonumber\\
&\quad+\frac{1}{6}\bra{z'}K\ket{z}\sum_{s}\bra{z'}K\ket{s}\bra{s}K\ket{z}\,(E_z+E_{z'}-2E_s)\,.
\end{align*}
On a chain, we combine the Fourier coefficients of single-site and adjacent-pair equations on each four-site window.
The resulting constant-size quadratic system determines the two middle-site fields, the three internal $ZZ$ couplings, and two endpoint combinations, each consisting of the endpoint field plus its outside $ZZ$ coupling.
As the window moves along the chain, every nonendpoint field appears as a middle-site field, and every edge appears as an internal edge.
At an open endpoint, the outside coupling is absent, so the endpoint combination directly gives the field.
This is how the overlapping windows can be ``stitched'' together to recover all individual coefficients.

Every window uses the same globally fixed representative of $K$, and its diagonal solution is generically unique in that representative, as certified below.
Thus there are no additional local sign choices to align.
If several windows estimate the same coefficient, we retain any one estimate, and numerical errors do not accumulate through successive windows.
On a square lattice, the corresponding systems recover each center field and its incident $ZZ$ couplings, with separate equations for straight sides and a final linear solve for the corners.
The detailed reconstruction and coverage arguments are in Appendices~\ref{sec:proof_1D_local_Z_basis} and~\ref{sec:proof_2D_local_Z_basis}.

\begin{figure}[t!]
\centering
\includegraphics[width=0.99\textwidth]{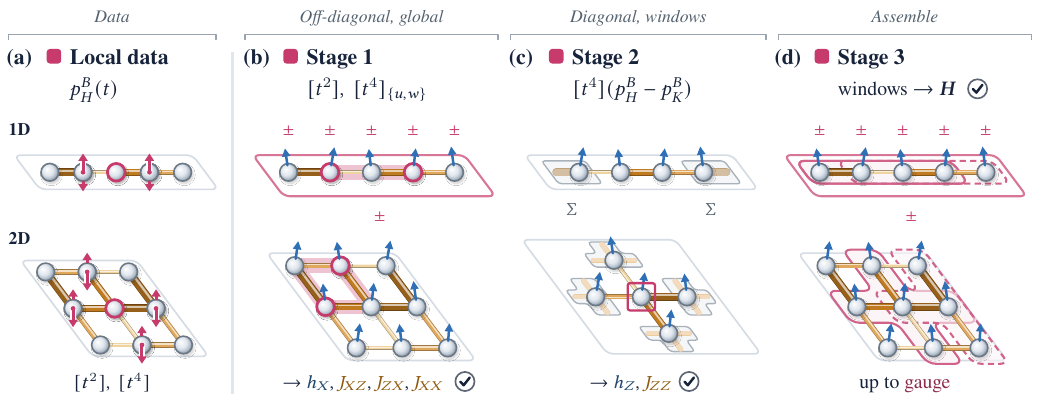}
\caption{\textbf{Local reconstruction of a generic NN IXZ Hamiltonian from $Z$-to-$Z'$ data.}
\textbf{(a)}~For transitions that flip the marked site (ring), vary the neighboring input bits (double arrows) and extract the second- and fourth-order coefficients of the local transition probabilities $p_{z\to z'}(t)$.
\textbf{(b)~Stage~1:} Second-order data recover each transverse field and its incident mixed couplings up to a common site sign ($\pm$). Fourth-order tests on nonadjacent pairs resolve the chain-endpoint ambiguities and determine relative $J_{XX}$ signs along the shaded paths, leaving one common sign.
\textbf{(c)~Stage~2:} Subtracting the fourth-order contribution from evolution under the recovered off-diagonal Hamiltonian $K$ gives local equations for $h_Z$ and $J_{ZZ}$. These are solved on windows, with separate equations for lattice sides and corners. The boundary sums $\Sigma$ combine fields with exterior $ZZ$ couplings.
\textbf{(d)}~Overlapping windows recover every field and coupling. Once a common representative of $K$ is fixed, each local diagonal system has a unique solution. The remaining gauges are site-dependent Pauli-$Z$ conjugation, with signs $s_j=\pm1$, and time reversal, with sign $\tau=\pm1$, as in \Cref{fig:settings_results}(d). In particular, $J_{XX}^j\mapsto\tau s_js_{j+1}J_{XX}^j$.}
\label{fig:overview_zbasis}
\end{figure}

\paragraph{Solving nonlinear systems.}

Unlike in the uniform control setting, the local polynomial systems we need to solve are not all linear. Furthermore, because there are unavoidable gauges, the systems at earlier stages are solved by fixing some globally consistent gauge representative before proceeding with later systems. Not only does this mean that the error propagation arguments from the previous setting no longer apply, but we are even missing a crucial step in justifying why these nonlinear systems provide enough constraints to identify the parameters up to gauge. Similar considerations arose in the recent work~\cite{chen2025quantum}, albeit in a different learning context, and in this work we use a simple algebraic argument to establish identifiability which can be made quantitative using smoothed analysis.

Roughly speaking, the challenge is as follows. We are given an overconstrained, albeit constant-sized, system of polynomial equations of the form $P(x) = P(x^*)$, where $x$ is a vector of indeterminates corresponding to a candidate estimate for the parameters, $P$ is a known vector of polynomials which maps a set of parameters to a set of experimental observables, and $P(x^*)$ is the value of those polynomials under the ground truth parameters $x^*$. We would like to certify that, modulo unavoidable gauges, any (approximate) solution $x$ to this system must be (approximately) equal to $x^*$, under a generic choice of $x^*$ sampled according to the smoothed analysis model. We show that for any such $P$, it suffices to exhibit a simple certificate depending only on the coefficients of the polynomials in $P$ to prove that this behavior holds. Roughly, it suffices to exhibit a collection of products of these polynomials with monomials in $x$ such that the matrix whose rows consist of the indicator vector of the constant monomial and the coefficient vectors of these product polynomials is nonsingular. Although the systems we consider are of constant size, it will be computationally cumbersome to directly exhibit such certificates, and instead we will leverage the structure of the polynomials that arise in our applications to directly solve for some subset of the unknowns in the first stage before invoking generic polynomial system solving for the remaining unknowns in the second stage. It turns out that the aforementioned certification procedure can still be adapted to work in this two-stage setting. Additionally, all of these arguments can be made quantitative, as the certificates ultimately demonstrate that certain polynomials are non-vanishing, and under smoothed analysis, polynomial anti-concentration of Gaussians ensures that these polynomials are not only non-vanishing but bounded away from zero. We defer the details to Appendix~\ref{sec:certificate}.

\section{Related works}\label{sec:related_works}

\paragraph{Hamiltonian learning from real-time dynamics.}
Hamiltonian learning from real-time dynamics has a long history, including methods that infer couplings from short-time derivatives, measured time traces, or correlation matrices~\cite{da2011practical,wiebe2014hamiltonian,bairey2019learning,stilck2024efficient,haah2024learning}. These protocols typically achieve standard-quantum-limit scaling $T=\widetilde O(1/\epsilon^2)$ in the total evolution time.

A more recent line of work achieves Heisenberg scaling $T=\widetilde O(1/\epsilon)$. Huang, Tong, Fang, and Su~\cite{HuangTongFangSu2023learning} first showed this for many-body Hamiltonians with known local interaction structure, and Bakshi, Liu, Moitra, and Tang~\cite{bakshi2024structure} extended it to Hamiltonians whose interaction structure is unknown (see also~\cite{ma2024learningkbodyhamiltonianscompressed}); more recent work, in the so-called \emph{ansatz-free} setting, learns sparse Hamiltonians without any assumptions on their interaction structure~\cite{hu2025ansatz,zhao2025learning,sinha2025improved}. These protocols rely on substantial quantum control, interleaving the evolution with fast, individually addressable control operations in the spirit of dynamical decoupling, or with ancillas and more general coherent operations. Dutkiewicz, O'Brien, and Schuster~\cite{dutkiewicz2024advantage} showed that such control is in fact necessary: without it, Heisenberg scaling is impossible in general. The standard quantum limit is therefore the best one can hope for without control, and Zhou and Gong~\cite{zhou2026optimalansatzfree} recently showed that it can be attained even in the ansatz-free setting, by a control-free and ancilla-free protocol that uses only Pauli product-state preparation and measurement.

In contrast to all of these works, including the control-free ones, which still assume the ability to prepare and measure each qubit in a basis of one's choosing, our access models allow no site-resolved basis changes at all: only a uniform rotation before and after the evolution, or only computational-basis preparation and measurement.

\paragraph{Hamiltonian learning on experimental platforms.}
Our access models are motivated primarily by programmable analog quantum simulators. In Rydberg atom arrays, which realize programmable Ising- and $XY$-type models with tunable geometry, the coherent drive is naturally global and the standard readout is site-resolved fluorescence in the computational basis~\cite{browaeys2020many,wurtz2023aquila}. Trapped-ion simulators similarly realize long-range Ising and $XY$ models with global spin rotations and state-dependent fluorescence~\cite{monroe2021programmable}. Most relevant to our work, Guo et al.~\cite{guo2025hamiltonian} recently learned the all-to-all Ising Hamiltonian of a $300$-ion crystal using only global manipulations and single-qubit-resolved state detection, a setting closely related to our uniform control model.

\paragraph{Learning under restricted access.}
A growing body of recent work asks what can be learned with limited experimental resources. Two recent works of Tong and collaborators show that the Heisenberg limit does not require fast control pulses: it can be reached using only static single-qubit fields during the evolution~\cite{shin2026emergent}, or with control applied only at long time intervals~\cite{shin2026eigenphase}. Relatedly, Flammia et al.~\cite{flammia2026autonomous} certify Hamiltonians using only single-qubit operations. Closer to our access models, Mao et al.~\cite{mao2026certifying} showed that uniform measurements, i.e., a common rotation of all qubits followed by computational-basis readout, suffice to certify certain graph states. Most relevant to our methods is the recent work by a subset of the current authors~\cite{chen2025quantum}, which studies the question of learning a many-body Hamiltonian from time evolution and a single local probe. In particular, it also sets up a system of polynomial constraints to identify the Hamiltonian parameters up to unavoidable gauge within the smoothed analysis model, and it similarly uses algebraic geometric methods to certify identifiability.

\section{Outlook}\label{sec:outlook}

We have developed a method to reconstruct a generic 2-local Hamiltonian $H$ on any known simple graph, using uniform rotations before and after evolution under $H$. With computational-basis preparation and measurement, we reconstruct the nearest-neighbor IXZ families studied here up to Pauli-$Z$ conjugation and time reversal. Both the uniform control and computational basis models use one time evolution of the Hamiltonian per experiment, and \Cref{thm:sql_2local} rules out Heisenberg scaling for a 1D nearest-neighbor IXZ family in the absence of interleaved control. There are a number of interesting questions suggested by our work.

\paragraph{Beyond IXZ Hamiltonians.}
Our $Z$-basis results assume an IXZ Hamiltonian. For general nearest-neighbor Hamiltonians with Pauli $Y$ terms, conjugation by diagonal unitaries and time reversal $H \mapsto -H^*$ are both invisible to solely $Z$-basis measurements. Are these transformations the only generic ambiguities within a fixed local family?  Although independent on-site $Z$ rotations do not generally remove all $Y$ terms simultaneously, it would be interesting to study which Hamiltonians admit an IXZ representative under a diagonal conjugation that preserves locality.

\paragraph{Learning from all-zero input.}
The isotropic translation-invariant nearest-neighbor results only require an all-zero input. Can we reconstruct the spatially varying nearest-neighbor IXZ families studied here from all-zero input alone, up to Pauli-$Z$ conjugation and time reversal, with polynomial resource bounds under Gaussian smoothing?

\paragraph{The role of adaptivity.}
In our reconstruction algorithm, the Hamiltonian coefficients recovered at one stage are substituted into the systems solved at later stages. The input states and measurement observables can be fixed in advance, but we can choose estimation precision adaptively (see Appendix~\ref{sec:certificate_refinement}). Can further adaptation of the input, allowed measurement basis, or evolution time based on previous outcomes reduce the overall sample complexity?  Could it reduce the number of experimental settings, or the amplification of measurement errors?

\paragraph{Heisenberg scaling with interleaved global control.}
The lower bound in \Cref{thm:sql_2local} assumes that each experiment consists of a single uninterrupted time evolution under $H$. But what changes if we allow uniform control during the time evolution? More precisely, suppose that state preparation remains a uniform product state and measurement is still performed in a common basis, but between evolution segments we may apply global pulses $U_j^{\otimes n}$ with $U_j\in\mathrm{SU}(2)$. Perhaps this additional control can be used to reconstruct the generic nearest-neighbor Hamiltonians of \Cref{thm:uniform} with total evolution time $\widetilde O(\epsilon^{-1}\poly(n,R,\varsigma^{-1},\delta^{-1}))$, polynomial classical post-processing time, and only polynomially many pulses.

\paragraph{Toward experimental implementation.}
A main motivation for this work is to develop Hamiltonian learning protocols that remain experimentally viable under severe control restrictions. Even so, arbitrary $Z$-basis preparation requires site-resolved bit flips, which are available on some tweezer and ion-trap platforms but not in a strictly global-control setting. In Appendix~\ref{sec:TI_IXZ}, we additionally provide an analysis of state-preparation and measurement (SPAM) noise which covers Theorems~\ref{thm:TI_NN_z_to_zp}, \ref{thm:TI_3_local}, and~\ref{thm:local_NN_z_to_zp}, assuming known local depolarizing noise and arbitrary computational-basis preparation. However, it does not address more general preparation and measurement errors. An important open question is which calibration assumptions are sufficient to learn the Hamiltonian and SPAM parameters jointly~\cite{stilck2024efficient,hangleiter2024robustly}.

Another experimental complication is that Rydberg and trapped-ion interactions generally have power-law tails~\cite{browaeys2020many,monroe2021programmable}. Our uniform-control results already apply to generic two-body Hamiltonians on arbitrary known interaction graphs, but the $Z$-basis reconstruction results rely on finite-range lattice interactions. Extending these reconstructions to sufficiently rapidly decaying long-range interactions remains an open problem.

Finally, realistic dynamics may include noise during the evolution itself. For time-homogeneous Markovian dynamics, the unknown generator is then a Lindbladian rather than a Hamiltonian. It would be valuable to know if our same basic approach can be used to identify a local Lindbladian using only computational-basis access or uniform control~\cite{ivashkov2026ansatz,arad2026near}.

\section*{Acknowledgements}
We thank Scott Aaronson, Richard Allen, Soonwon Choi, Fangjun Hu, Hong-Ye Hu, Liang Mao, John Preskill, and Yuxin Wang for helpful and stimulating discussions.

\paragraph{AI disclosure.}
The authors developed the original problem formulation and overall mathematical approach. Frontier models were used to assist at several points in implementing this approach. At an earlier stage, the authors devised on their own a protocol for learning under uniform control that held for a large family of graphs including most standard lattices (hypercubic, triangular, kagome, etc.), but GPT-6 subsequently simplified parts of the argument, yielding the result for general graphs. GPT-6 also suggested the simple algebraic certification approach which is used to resolve the diagonal terms in the IXZ setting. Fable 5.1 was used to produce the illustrations, and GPT-6 and Opus 5.5 assisted with literature search and preliminary drafting of the exposition. All proofs, arguments, and statements included in the paper were subsequently checked, digested, and refined by the human authors.

\bibliographystyle{unsrt}
\bibliography{ExpHamLearn}
\clearpage
\appendix
\section{Basic notations}\label{sec:notations}

Throughout, we write $\|B\|$, $\|B\|_1$, and $\norm{B}_F$ for the operator norm, the trace norm, and the Frobenius norm of a matrix $B$, and $\|v\|_p$ for the $\ell^p$ norm of a vector $v$.
We use $\tilde{\mathcal O}$, $\tilde{\Theta}$, and $\tilde{\Omega}$ to suppress polylogarithmic factors. 
We write $\mathfrak D$ for spatial dimension, and abbreviate $\{1,\dots,n\}$ as $[n]$.

For $z,z'\in\{0,1\}^n$, we write $(A)_{zz'}\coloneqq\braket{z|A|z'}$.
For operators $A,B$ and $r\in\mathbb Z_{\ge0}$, we write $\text{ad}_A^r(B)$ for the $r$-th nested commutator, defined recursively by
\begin{align*}
\text{ad}_A^0(B)=B,\qquad
\text{ad}_A^{r+1}(B)=[A,\text{ad}_A^r(B)].
\end{align*}

For $\vec x=(x_1,x_2,x_3)^\top\in\mathbb R^3$, the cross-product matrix $[\vec x]_\times$ is defined by $[\vec x]_\times\vec y=\vec x\times\vec y$ for every $\vec y\in\mathbb R^3$; explicitly,
\begin{equation}\label{eq:cross_matrix}
[\vec x]_\times=
\begin{pmatrix}
0&-x_3&x_2\\
x_3&0&-x_1\\
-x_2&x_1&0
\end{pmatrix}.
\end{equation}

\section{Hamiltonian structures and access models}\label{sec:access_model}
In this section, we specify the Hamiltonian family and experimental access separately. We fix the Hamiltonian to be traceless, as the trace of a Hamiltonian is an unobservable quantity. Therefore, we can write the Hamiltonian in the Pauli basis as:
\begin{equation*}
H=\sum_{s_x,s_y,s_z}\alpha_{s_x,s_y,s_z}X^{s_x}Y^{s_y}Z^{s_z},\quad \alpha_{s_x,s_y,s_z}\in\mathbb{R},
\end{equation*}
where $s_x,s_y,s_z\subseteq[n]$ are pairwise disjoint and their union is not empty. We write $X^{s_x}\coloneqq\prod_{j\in s_x}X_j$, $Y^{s_y}\coloneqq\prod_{j\in s_y}Y_j$, and $Z^{s_z}\coloneqq\prod_{j\in s_z}Z_j$, with $X_j,Y_j,Z_j$ the Pauli matrices acting on qubit $j$ and the identity on all other qubits. We define reconstruction directly in terms of these Pauli coefficients. 
The smoothing model and the coefficient normalization used in its analysis are introduced separately in Appendix~\ref{sec:smooth_analysis}.

\begin{definition}[Hamiltonian reconstruction]\label{def:Reconstruct_Hamiltonian}
We reconstruct $H$ to accuracy $\epsilon$ in the coefficient $\ell^\infty$-norm if the estimates $\hat\alpha_{s_x,s_y,s_z}$ satisfy $|\hat\alpha_{s_x,s_y,s_z}-\tilde\alpha_{s_x,s_y,s_z}|\le\epsilon$ for every Pauli term, where all $\tilde\alpha_{s_x,s_y,s_z}$ are coefficients of one Hamiltonian $\widetilde H$ in the specified family that is related to $H$ by an unavoidable gauge of the access model.
These gauges are described in Appendix~\ref{sec:gauges_general}.
\end{definition}

\noindent Now, we formally define geometric locality, translation invariance, isotropy, and IXZ Hamiltonians that are considered in this work:
\begin{definition}[Geometric locality]\label{def:k-local_Hamiltonian_XYZ}
Let $G=(V,E)$ be a graph with one qubit at each vertex, and let $d_G$ denote its shortest-path distance. We call $H$ geometrically $k$-local ($k$-\gl) if the support $S$ of every Pauli term with nonzero coefficient satisfies
\[
\operatorname{diam}_G(S)\coloneqq\max_{u,v\in S}d_G(u,v)\le k-1.
\]
In other words, any two qubits in the same interaction are at most $k-1$ graph steps apart.
\end{definition}
\noindent Notice that for our results for nearest-neighbor (NN) Hamiltonians, we additionally restrict $H$ to single-site terms and two-site terms on edges. This excludes three-body terms on triangles even though their support has diameter one. We choose an orientation $\vec E$ storing each edge of a simple graph $G=(V,E)$ once. With $\vec\sigma_v=(X_v,Y_v,Z_v)^\top$, the general NN Hamiltonian is
\begin{equation}\label{eq:graph-ham-general}
H=\sum_v\vec h_v^\top\vec\sigma_v+
\sum_{(u\to v)\in\vec E}\vec\sigma_u^\top J_{uv}\vec\sigma_v,
\qquad \vec h_v\in\R^3,\quad J_{uv}\in\R^{3\times3}.
\end{equation}
There are in total $3|V|+9|E|$ parameters and reversing an edge orientation uses the same coupling matrix transposed, i.e. $J_{vu}=J_{uv}^\top$.

\begin{definition}[Translation invariance and isotropy]\label{def:TI-Hamiltonian}\label{def:isotropic}
A Hamiltonian on a finite open lattice $\Lambda\subset\Z^{\mathfrak D}$ is translation-invariant (TI) if it is the restriction of an infinite-lattice interaction whose coefficients are unchanged by translations. Only terms wholly supported in $\Lambda$ are retained. It is isotropic if the coefficients are also invariant under lattice rotations and reflections. 
\end{definition}
\begin{definition}[IXZ Hamiltonian]\label{def:IXZ_Hamiltonian}
An IXZ Hamiltonian contains no Pauli-$Y$ terms:
\begin{equation*}
H=\sum_{s_x,s_z}\alpha_{s_x,s_z}X^{s_x}Z^{s_z},
\qquad s_x\cap s_z=\varnothing,\quad \alpha_{s_x,s_z}\in\R.
\end{equation*}
It is therefore a real symmetric matrix in the computational basis.
\end{definition}

In the following part of this section, we will first define the access models considered in this work (Appendix~\ref{sec:SU2_basis_access_models}), and then we will discuss the unobservable gauges for different access models (Appendix~\ref{sec:gauges_general}); in the end, we will present the lower bound for control-free access models that excludes the possibility of achieving Heisenberg scaling (Appendix~\ref{sec:no_heisenberg_scaling}).

\subsection{Access models}\label{sec:SU2_basis_access_models}
Each experiment consists of state preparation, one uninterrupted evolution $U(t)=e^{-iHt}$ at a chosen positive time $t$, and a final measurement. The allowed preparation and measurement operations depend on the access model; no control is interleaved with the unknown evolution. We count one preparation--evolution--measurement cycle as one experimental shot; sample and query counts refer to these shots. Our protocols use short times whose dependence on precision is given in Appendix~\ref{sec:derivative_estimation}. We consider the following definitions:
\begin{definition}[Uniform-control model]\label{def:uniform-local-control-model}
The $\hat a$-to-$\hat b$ model prepares $\ket{0_{\hat a}}^{\otimes n}$, where $\ket{0_{\hat a}}$ is the $+1$ eigenstate of $\hat a\cdot\vec\sigma$, and measures each qubit along $\hat b$. The unit vectors $\hat a,\hat b$ may be chosen for each experiment. This is equivalent to all-zero preparation and $Z$ readout with uniform single-qubit rotations before and after the evolution, as illustrated in Fig.~\ref{fig:control}(a).
\end{definition}
\begin{definition}[$Z$-basis models]\label{def:all-computational-basis-model}\label{def:all-zero-input-model}
The $Z$-to-$Z'$ model, illustrated in Fig.~\ref{fig:control}(b), prepares an arbitrary computational-basis string $\ket z$ and measures in the same basis. The $0$-to-$Z'$ model restricts the input to $\ket{0^n}$. The corresponding measurement statistics are
\begin{equation*}
p_{z\to z'}(t)=|\bra{z'}e^{-iHt}\ket z|^2.
\end{equation*}
\end{definition}
\begin{remark}[Local preparation and measurement]
\label{rem:local_access_model}
Suppose we measure a set of qubits $S$ under a geometrically
$k$-local Hamiltonian. To estimate the constant $r$th order time derivatives of the observable value at $t=0$, it is enough to prepare the prescribed product input on all sites within graph distance $r(k-1)$
of $S$ with respect to the interaction graph of the Hamiltonian, and measure only $S$. For fixed locality, maximum degree, and derivative order, this preparation region has size independent of the total system size.
\end{remark}

In this work, we define the \textit{time resolution} of a protocol to be the minimum requested evolution duration under the unknown Hamiltonian. In our protocol, the time resolution decreases only polylogarithmically with the target precision $\epsilon$.

\subsection{Unobservable gauges}\label{sec:gauges_general}
\begin{proposition}\label{thm:unobservable_gauges}\label{thm:XY_indistinguishability}\label{thm:Z_gauge}\label{prop:signflip}\label{thm:time-reversal-gauge}
For any Hamiltonian $H$, computational-basis transition probabilities are unchanged under $H\mapsto U_ZHU_Z^\dagger$ and $H\mapsto-H^*$, where $U_Z$ is a diagonal unitary and $H^*$ denotes entrywise complex conjugation in the computational basis.
For an IXZ Hamiltonian, $H$ is real, so the second transformation is simply the overall sign change $H\mapsto-H$. Conjugation by $Z^S=\prod_{j\in S}Z_j$ changes the coefficients according to
\begin{equation}\label{eq:alphaflip}
\alpha_{s_x,s_z}\longmapsto(-1)^{|s_x\cap S|}\alpha_{s_x,s_z}.
\end{equation}
\end{proposition}
\begin{proof}
Write $U_Z\ket z=e^{i\phi_z}\ket z$. Diagonal conjugation multiplies the transition amplitude from $z$ to $z'$ by $e^{i(\phi_{z'}-\phi_z)}$. This phase disappears when we take the absolute square, so the transition probability is unchanged. Also $e^{-i(-H^*)t}=(e^{-iHt})^*$, so replacing $H$ by $-H^*$ complex-conjugates each transition amplitude and leaves its absolute square unchanged. Finally, $ZXZ=-X$ and $ZZZ=Z$, so conjugation by $Z^S$ contributes one minus sign for each site in $s_x\cap S$, giving \eqref{eq:alphaflip}.
\end{proof}
Notice that the gauge caused by diagonal conjugations can map the Hamiltonian to leave its local family as arbitrary diagonal conjugations can destroy locality.

\subsection{Evolution time lower bound without interleaved control}\label{sec:no_heisenberg_scaling}
The lower bound below applies to the single-evolution restriction of Fig.~\ref{fig:control}(c), with no interleaved controls.
\begin{figure}[t!]
    \centering
    \includegraphics[width=0.99\textwidth]{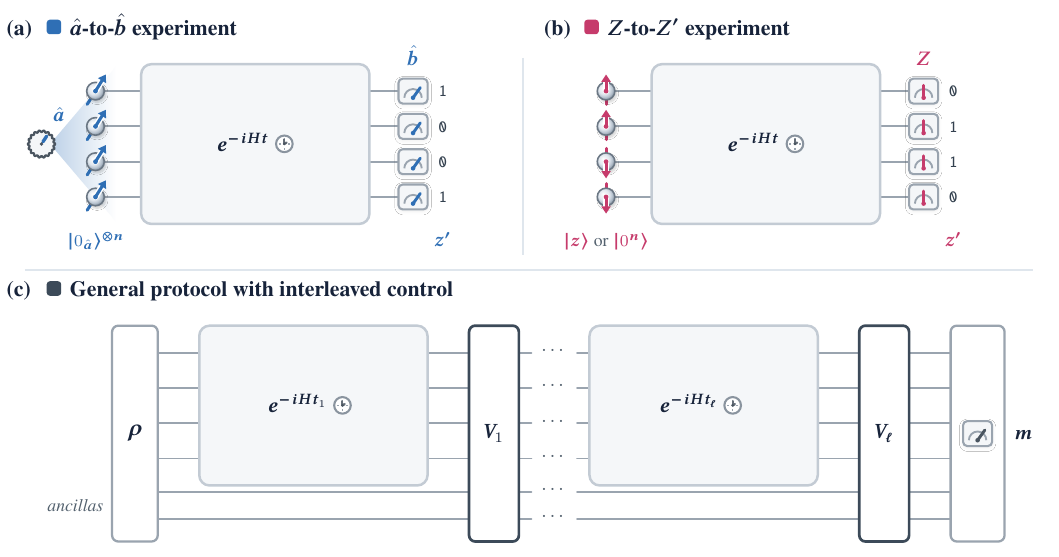}
    \caption{\textbf{Experimental access to the unknown Hamiltonian $H$.}
\textbf{(a)~Uniform control ($\hat a$-to-$\hat b$):}
Prepare every qubit along a common direction $\hat a$, evolve once under $H$, and measure along a common direction $\hat b$. Uniform rotations implement the preparation and readout, and the two directions may vary between experiments.
\textbf{(b)~Computational-basis access ($Z$-to-$Z'$):}
Prepare $\ket z$ in the computational basis, evolve once under $H$, and measure in the computational basis, obtaining $z'$ with probability $p_{z\to z'}(t)$. Arbitrary inputs require site-resolved bit flips, and the all-zero-input model fixes $z=0^n$. Our protocols vary input bits and measure output probabilities on constant-size regions (\Cref{rem:local_access_model}).
Neither (a) nor (b) allows control during evolution.
\textbf{(c)~General coherent control:}
An arbitrary joint state $\rho$ of the system and ancillas undergoes evolution segments $U(t_j)=e^{-iHt_j}$ on the system, alternating with joint control unitaries $V_j$, followed by a collective measurement with outcome $m$. Restricting each experiment to one uninterrupted evolution gives the setting of our lower bound (\Cref{thm:sql_2local}), even with arbitrary inputs, passive ancillas, and measurements.}
    \label{fig:control}
\end{figure}

\begin{theorem}\label{thm:sql_2local}
There is a 1D NN IXZ family for which reconstruction to coefficient accuracy $0<\epsilon\le\pi/48$, with success probability at least $2/3$ uniformly over the family, requires total evolution time $T\ge1/(144\epsilon^2)$. 
This holds for classically adaptive experiments, biased estimators, arbitrary inputs and measurements, and the use of extra ancillary qubits, provided each experiment uses one continuous time evolution and experiments are not coherently combined by interleaved controls.
\end{theorem}
\begin{proof}
Consider the 1D NN IXZ family on two qubits, with $\phi\in[\pi/6,\pi/4]$, given by
\begin{equation*}
H_\phi=e^{-i\phi Y_1/2}X_1X_2e^{i\phi Y_1/2}
=\cos\phi\,X_1X_2-\sin\phi\,Z_1X_2,
\end{equation*}
with all other coefficients zero. Include arbitrary passive ancillas from the outset and set $\widetilde H_\theta=H_\theta\otimes I_{\mathrm a}$. Since the Pauli terms anticommute and square to $I$, we have $\widetilde H_\theta^2=I$ and $(\widetilde H_\phi-\widetilde H_\psi)^2=4\sin^2(\Delta/2)I$, where $\Delta=|\phi-\psi|$. Separating even and odd powers in the exponential gives
\begin{align*}
\widetilde U_\theta(t)&=\cos t\,I-i\sin t\,\widetilde H_\theta,\\
\|\widetilde U_\phi(t)-\widetilde U_\psi(t)\|
&=2|\sin t|\,|\sin(\Delta/2)|\le\min\{t,1\}\Delta.
\end{align*}

To handle a mixed input $\rho$ on the system and passive ancillas, choose a purification $\ket\Psi$ with $\rho=\operatorname{Tr}_R\ket\Psi\!\bra\Psi$. Let $\ket{\Psi_\theta}=(\widetilde U_\theta(t)\otimes I_R)\ket\Psi$ and $\rho_\theta=\operatorname{Tr}_R\ket{\Psi_\theta}\!\bra{\Psi_\theta}$. Discarding $R$ cannot increase trace distance; using the pure-state formula and $\|A\otimes I_R\|=\|A\|$, we obtain
\begin{align*}
\frac12\|\rho_\phi-\rho_\psi\|_1
&\le\sqrt{1-|\langle\Psi_\phi|\Psi_\psi\rangle|^2}\\
&\le\|\ket{\Psi_\phi}-\ket{\Psi_\psi}\|
\le\|\widetilde U_\phi(t)-\widetilde U_\psi(t)\|.
\end{align*}
Thus the operator norm bound also applies to mixed inputs. For a joint measurement with POVM elements $E_m$ (extended by the identity on $R$, which preserves the outcome probabilities), let $p_\theta(m)=\|E_m^{1/2}\ket{\Psi_\theta}\|^2$. The reverse triangle inequality and $\sum_mE_m=I$ give
\begin{align*}
\sum_m\bigl(\sqrt{p_\phi(m)}-\sqrt{p_\psi(m)}\bigr)^2
&\le\sum_m\|E_m^{1/2}(\ket{\Psi_\phi}-\ket{\Psi_\psi})\|^2\\
&=\|\ket{\Psi_\phi}-\ket{\Psi_\psi}\|^2\le\min\{t^2,1\}\Delta^2.
\end{align*}

For adaptive experiments, let $P_{\theta,j}(h)$ be the probability of history $h$ after $j$ experiments, and let $d_j(h)$ be the preceding sum for the next experiment conditioned on $h$. The same history determines the same experiment under both hypotheses, so $d_j(h)\le\Delta^2t_{j+1}(h)$, using $\min\{t^2,1\}\le t$. Define $D_j=\sum_h(\sqrt{P_{\phi,j}(h)}-\sqrt{P_{\psi,j}(h)})^2$, with $D_0=0$. Expanding the squares and using $\sqrt{ab}\le(a+b)/2$ yields
\begin{align*}
D_{j+1}-D_j
&=\sum_h\sqrt{P_{\phi,j}(h)P_{\psi,j}(h)}\,d_j(h)\\
&\le\frac{\Delta^2}{2}\bigl(\mathbb E_\phi t_{j+1}+\mathbb E_\psi t_{j+1}\bigr).
\end{align*}
Summing and using $\sum_jt_j\le T$ gives $D\le\Delta^2T$, where $D$ is the same sum for the full transcript distributions $P_\phi,P_\psi$. Padding stopped histories with zero-duration rounds and taking limits also covers adaptive stopping.

Finally, Cauchy--Schwarz and $(\sqrt a+\sqrt b)^2\le2(a+b)$ give
\begin{align*}
\operatorname{TV}(P_\phi,P_\psi)
&=\frac12\sum_h|P_\phi(h)-P_\psi(h)|\\
&\le\frac12\sqrt{D}
\left[\sum_h\bigl(\sqrt{P_\phi(h)}+\sqrt{P_\psi(h)}\bigr)^2\right]^{1/2}\\
&\le\sqrt{D}\le\Delta\sqrt T.
\end{align*}

Any diagonal $U_Z=\diag(u_{00},u_{01},u_{10},u_{11})$ for which $U_ZH_\theta U_Z^\dagger$ is IXZ has, after removal of a common phase, $u_z\in\{\pm1\}$, since the nonzero off-diagonal entries of $H_\theta$ connect all four basis states.
The two matrix elements that flip both qubits must be equal in an IXZ Hamiltonian, forcing $u_{00}u_{11}=u_{01}u_{10}$.
Hence $U_Z$ is a product of Pauli-$Z$ operators up to a phase, so arbitrary diagonal conjugations give no further gauge within this family.

For $0<\epsilon\le\pi/48$, choose $\phi=\pi/6$ and $\psi=\phi+4\epsilon\le\pi/4$.
Then $\sin\psi-\sin\phi\ge2\sqrt2\epsilon>2\epsilon$, so the coefficient vectors remain separated by more than $2\epsilon$ even after allowing diagonal conjugation and overall-sign gauges within the IXZ family.
Reconstruction with success probability $2/3$ on both therefore gives a test requiring $\operatorname{TV}(P_\phi,P_\psi)\ge1/3$.
Thus $1/3\le4\epsilon\sqrt T$, which gives $T\ge1/(144\epsilon^2)$.
\end{proof}
\noindent The related general control-free lower bounds are also discussed in Refs.~\cite{dutkiewicz2024advantage,hu2025ansatz}.

\section{Mathematical tools}\label{sec:preliminaries}
Our reconstruction procedures use low-order time derivatives to obtain local equations for the Hamiltonian coefficients. To turn these equations into accurate estimates, we need both reliable derivative data and control over the amplification of errors during reconstruction. Lemma~\ref{lem:gaussian_polynomial} bounds the conditioning of the local systems under smoothing, and Theorem~\ref{thm:main_complexity4derivative_estimation} gives the cost of estimating the required derivatives. Lemma~\ref{lem:finite_stage_stability} then controls the errors through successive reconstruction steps, using the nonlinear certificates from Appendix~\ref{sec:certificate}.

\subsection{Smoothed analysis}\label{sec:smooth_analysis}
Several reconstruction steps involve division or inversion of a local matrix. These steps can amplify experimental errors when a denominator or matrix minor is small, even if the exact equations have a unique solution. To analyze this instability, we model the free coefficients as small random perturbations of an arbitrary fixed Hamiltonian, including one at which a reconstruction step is singular. Unlike average-case analysis, which draws the entire instance from a prescribed distribution, smoothed analysis allows any fixed base Hamiltonian and bounds the failure probability over the perturbation. The reconstruction target is the perturbed Hamiltonian itself.
\begin{definition}[Smoothed Hamiltonians]\label{def:Ham_family}\label{def:smoothed_hamiltonians}
Let $\params=(\alpha_1,\ldots,\alpha_N)$ be the vector of $N$ free coefficients specifying the chosen family. When translation invariance or isotropy requires several coefficients to be equal, they share one entry and receive one common perturbation. For any fixed base vector $\overline\params$, let
\begin{align*}
\params^*=\overline\params+\varsigma Z,\qquad Z\sim\mathcal N(0,I_N),\quad\varsigma>0.
\end{align*}
We reconstruct $H_{\params^*}$, the Hamiltonian obtained by replacing each free coefficient $\alpha_j$ with its perturbed value $\alpha_j^*$.
\end{definition}
\noindent Definition~\ref{def:smoothed_hamiltonians} follows the smoothed-analysis framework used in Refs.~\cite{spielman2004smoothed,chen2025quantum}. If $\|\overline\params\|_\infty\le R_0$, the known radius
\begin{equation}\label{eq:effective_radius}
R=1+R_0+\varsigma\sqrt{2\log(8N/\delta)}
\end{equation}
bounds $\|\params^*\|_\infty$ with probability at least $1-\delta/4$. On this event, we apply estimates that require coefficients bounded by one to $H_{\params^*}/R$. The normalization $R\ge1$ will be used throughout.

The denominators, minors, and certificates used below are controlled by polynomials in the free Hamiltonian coefficients. Showing that such a polynomial is not identically zero rules out an exact degeneracy with probability one under Gaussian perturbations. Stable reconstruction requires more: we need a quantitative bound on the probability that its value is close to zero. Lemma~\ref{lem:gaussian_polynomial} provides this bound without any restriction on the base coefficients.

\begin{lemma}[Polynomial nondegeneracy]\label{lem:gaussian_polynomial}
For a fixed nonzero polynomial $p$ of degree $s\ge1$ in independently smoothed coordinates, there is a constant $c_p>0$, independent of their means, such that
\begin{align*}
\Pr\bigl[|p(\params^*)|<c_p(\delta'\varsigma)^s\bigr]\le\delta',
\qquad 0<\delta'<1.
\end{align*}
For a nonzero constant polynomial, use its fixed absolute value instead.
\end{lemma}
\begin{proof}
Let $p_s$ be the sum of the terms of total degree $s$ in $p$. Since $p_s$ is nonzero, we can choose a real unit vector $v$ with $p_s(v)\ne0$. Decompose the standard Gaussian vector as $Z=Uv+Z_\perp$, where $U=v\cdot Z$ is a standard Gaussian scalar and $Z_\perp=Z-Uv$ is independent of $U$. Conditional on $Z_\perp$, consider the single-variable polynomial
\begin{align*}
q(u)=p(\overline\params+\varsigma Z_\perp+\varsigma uv).
\end{align*}
Its leading coefficient is $a=\varsigma^s p_s(v)$, which depends on neither the base vector $\overline\params$ nor the conditioned value of $Z_\perp$.

Factor $q(u)=a\prod_{j=1}^s(u-r_j)$ over the complex numbers. For $\xi>0$, if $|q(u)|<|a|\xi^s$, then $|u-r_j|<\xi$ for some $j$. For real $u$, this implies $|u-\operatorname{Re}r_j|<\xi$. Thus the set of real $u$ for which $q$ is small enough is contained in at most $s$ intervals, each of length $2\xi$. The standard Gaussian density is at most $1/\sqrt{2\pi}$, so
\begin{align*}
\Pr\bigl[|q(U)|<|a|\xi^s\mid Z_\perp\bigr]
\le\frac{2s\xi}{\sqrt{2\pi}}.
\end{align*}
Choose $\xi=\sqrt{2\pi}\,\delta'/(2s)$ and set
\begin{align*}
c_p=|p_s(v)|\left(\frac{\sqrt{2\pi}}{2s}\right)^s.
\end{align*}
Then $|a|\xi^s=c_p(\delta'\varsigma)^s$. Averaging the conditional bound over $Z_\perp$ proves the claim.
\end{proof}
In the reconstruction proofs, we establish polynomial nonzeroness by exhibiting a single exact assignment at which the relevant denominator, minor, or certificate is nonzero. Lemma~\ref{lem:gaussian_polynomial} turns this algebraic check into a conditioning bound for every fixed base Hamiltonian. The proof works uniformly over the base coefficients because they do not affect the leading coefficient of $q$. For $O(n)$ instances of a fixed finite collection of local polynomials, we assign each a failure probability of order $\delta'/n$ and take a union bound. Their values are then bounded away from zero simultaneously, with constants determined by that finite collection.

The linear systems used below inherit errors from two sources. Their right-hand sides use estimated derivatives, and their matrices may depend on coefficients recovered earlier. Lemma~\ref{lem:perturbed_least_squares} bounds the reconstruction error caused by both perturbations.
\begin{lemma}[Perturbed least squares]\label{lem:perturbed_least_squares}
Let $M\in\mathbb R^{m\times d}$ have full column rank, with smallest singular value $\sigma_{\min}(M)=\kappa>0$, and suppose $Mx=y$. Write $+$ for the Moore--Penrose pseudoinverse and use the operator norm for matrices. Then the least-squares solution $\widehat x=\widehat M^+\widehat y$ satisfies
\begin{align*}
\|\widehat x-x\|_2\le
\frac{2}{\kappa}\bigl(\|\widehat y-y\|_2+\|\widehat M-M\|\,\|x\|_2\bigr)
\end{align*}
whenever $\|\widehat M-M\|\le\kappa/2$.
\end{lemma}
\begin{proof}
Weyl's inequality gives $\|\widehat M^+\|\le2/\kappa$ and full column rank. Thus $\widehat x-x=\widehat M^+(\widehat y-y-(\widehat M-M)x)$.
\end{proof}

\subsection{Derivative estimation}\label{sec:derivative_estimation}
The polynomial system in our reconstruction algorithm involves estimating the derivatives of local moments and transition probabilities at $t=0$, while experiments give noisy samples at some positive nonzero evolution times. Generalizing the polynomial-fitting approach of Ref.~\cite{stilck2024efficient}, we estimate these derivatives from short-time data. Locality bounds the approximation error independently of the number of qubits, allowing a fixed-order derivative to be estimated with an inverse-square dependence on the desired precision, up to logarithmic factors. For fixed interaction strength, the time resolution is $\widetilde\Theta(1)$: the shortest queried evolution time decreases only as an inverse polylogarithm of the precision.

For a density matrix $\rho$ and a Hermitian observable $O$ with $\|O\|\le1$, write
\begin{align*}
p_{\rho,O}(t)=\tr(\rho e^{iHt}Oe^{-iHt}),\qquad
f(r,\rho,O)=p_{\rho,O}^{(r)}(0)=i^r\tr(\rho\,\operatorname{ad}_H^r(O)).
\end{align*}
We use the convention
\[
  [t^r]p_{\rho,O}(t)=\frac{f(r,\rho,O)}{r!}.
\]
For computational-basis input and output, we abbreviate
\begin{align*}
f(r,z,z')\coloneqq f(r,\ketbra{z}{z},\ketbra{z'}{z'}).
\end{align*}
For a local output probability, the observable is instead the projector on the measured bits, tensored with the identity on the unmeasured sites.
\begin{theorem}[Local derivative estimation]\label{thm:main_complexity4derivative_estimation}
Let $H=\sum_A H_A$ be a Hamiltonian on $n$ qubits, where each Hermitian interaction term $H_A$ acts on a nonempty subset $A\subseteq\{1,\ldots,n\}$ of at most $m$ qubits. Assume a known local energy bound $g\ge0$ satisfying
\[
  \max_{1\le v\le n}\sum_{A\ni v}\|H_A\|\le g,
\]
where $v$ indexes the qubits and $\|\cdot\|$ denotes the operator norm.
Let $\rho$ be a density matrix that can be prepared repeatedly, and let $O$ be a measurable Hermitian observable supported on at most $\chi$ qubits with $\|O\|\le1$.
For evolution time $t$, define the expected measurement outcome
\[
  p_{\rho,O}(t)=\tr(\rho e^{iHt}Oe^{-iHt}),
  \qquad G_{\mathrm{loc}}=\max\{1,\max\{1,m-1\}g\}.
\]
Here $\tr$ denotes the matrix trace.
For fixed positive integers $m,\chi,r$ and parameters $0<\eta,\delta'<1$, the derivatives $p_{\rho,O}^{(s)}(0)$ for $s=0,\ldots,r$ can be estimated with additive error at most $\eta$, simultaneously with probability at least $1-\delta'$, using
\[
  N=\widetilde O\!\left(G_{\mathrm{loc}}^{2r}\eta^{-2}\log(2/\delta')\right)
\]
independent measurements of $O$ after evolution under $H$ from input state $\rho$.
The estimator uses a polynomial interpolant of integer degree $d\ge\max\{2,r\}$, which can be chosen with $d=O(1+\log(2+G_{\mathrm{loc}}^r/\eta))$, and requires $O(N+d^2)$ classical arithmetic operations.
All queried evolution times lie in $[c/(G_{\mathrm{loc}}d^2),C/G_{\mathrm{loc}}]$, where $c=1/12$ and $C=1/4$ are universal positive constants.
The hidden constants depend only on $m,\chi,r$, and $\widetilde O$ suppresses factors polynomial in $\log(2+G_{\mathrm{loc}}^r/\eta)$.
In particular, for fixed $g$, the measurement bound is $\widetilde O(\eta^{-2}\log(2/\delta'))$, with constants depending only on $m,g,\chi,r$, and is independent of $n$.

For a two-body Hamiltonian on a graph of maximum degree $\Delta$, if $R\ge1$ bounds the absolute values of all one- and two-qubit Pauli coefficients, one may take $m=2$ and $g=G_{\mathrm{loc}}=C_0(\Delta+1)R$, where $C_0\ge1$ is a universal constant.
\end{theorem}

In Theorem~\ref{thm:main_complexity4derivative_estimation}, the degree $d$ is chosen to keep both the Taylor remainder and the sampling error under control. The remainder decreases exponentially with $d$, while extracting a fixed-order derivative from the fitted polynomial amplifies measurement errors by only a polynomial in $d$. A logarithmic degree therefore suffices to make the approximation accurate while preserving the inverse-square dependence of the shot count on precision. The shortest queried time is of order $d^{-2}$ for fixed $g$, which gives the $\widetilde\Theta(1)$ time resolution as stated above.

We defer the proof of Theorem~\ref{thm:main_complexity4derivative_estimation} to Appendix~\ref{sec:defer_derivative_estimation}, where Theorem~\ref{thm:complexity4derivative_estimation_formal} also keeps track of the interaction strength. This dependence matters when the graph degree or the coefficient radius grows. With $G_{\mathrm{loc}}=\max\{1,\max\{1,m-1\}g\}$, the shot count acquires a factor $G_{\mathrm{loc}}^{2r}$, up to logarithmic factors, and we can take $d=O_r(1+\chi+\log(2+G_{\mathrm{loc}}^r/\eta))$. The queried time interval is scaled by $G_{\mathrm{loc}}^{-1}$. For a nearest-neighbor Hamiltonian on a graph of maximum degree $\Delta$ with coefficient radius $R$, one may take $G_{\mathrm{loc}}=C_0(\Delta+1)R$.

We may also apply the estimator to the normalized Hamiltonian $H/R$. Under this normalization, a derivative of order $s$ is divided by $R^s$. Since $R\ge1$, choosing normalized precision $\eta/R^r$ controls the physical derivatives through order $r$ to accuracy $\eta$; a normalized evolution time $\tau$ is implemented by physical time $\tau/R$.

In the main text, this procedure supplies the moment coefficients through third order for uniform-control reconstruction and the transition-probability coefficients through fourth order for the $Z$-basis protocols. Each application estimates the full dynamics. Uniform-control experiments use the full uniform product input; in the local $Z$-basis constructions, only the specified local input region must be varied. When several observables are accessible from the same measured bitstrings, the union bound in Appendix~\ref{sec:defer_derivative_estimation} gives simultaneous guarantees with only logarithmic dependence of the sample count on their number.

\subsection{Constant-size reconstruction and stability}\label{sec:constant_size_stability}
The derivative estimates from Theorem~\ref{thm:main_complexity4derivative_estimation} supply the data for our reconstruction equations. To recover the Hamiltonian coefficients, we must also control how errors in those data propagate through the algebraic solves. Uniqueness alone does not provide this control: a small denominator or a nearly singular system can amplify even a small measurement error. Our protocols reduce reconstruction to local systems with a bounded number of unknowns and bounded algebraic degree, independently of the total number of qubits. They may solve many such systems, but each estimate passes through only a bounded number of reconstruction stages. This prevents successive conditioning losses from accumulating along a chain whose length grows with the system size.

Each nonlinear step solves a polynomial system $P_\theta(x) = P_\theta(t)$, where $\theta$ collects the quantities recovered in earlier stages, $t$ collects the coefficients that remain, and $x$ is a candidate for $t$.  See \Cref{sec:certificate} for details. \Cref{lem:exactcertificate} certifies that $t$ is the only solution, and \Cref{thm:quantitative} converts the residual of any candidate into a parameter error, including errors in the estimates of $\theta$ and of the right-hand side. \Cref{cor:certificate_partial_reconstruction} supplies the lower bound on the certificate under smoothing, even when $\theta$ and $t$ are correlated polynomial functions of the original coefficients. \Cref{prop:coarse_refinement} computes the solution without increasing the soft $\epsilon^{-2}$ dependence of the resource bounds.

Lemma~\ref{lem:finite_stage_stability} combines these nonlinear guarantees with Lemma~\ref{lem:perturbed_least_squares} and the elementary operations and branch tests used elsewhere in the protocols. On the conditioning events supplied by Lemma~\ref{lem:gaussian_polynomial}, a fixed number of stages amplifies the input error by at most a polynomial factor. We can therefore choose the derivative precision to be the desired coefficient precision divided by that factor, also below the stability threshold, and apply Theorem~\ref{thm:main_complexity4derivative_estimation} to obtain the main resource bounds. The individual reconstruction proofs verify the required nonzero polynomials and the bounded stage count.

\begin{lemma}[Finite-stage stability]\label{lem:finite_stage_stability}
Consider a local reconstruction with a fixed number of unknowns and stages. At each stage, the reconstruction may evaluate polynomials, perform arithmetic operations, recover a nonzero magnitude from its square, select among finitely many candidates separated by a positive gap, or solve a linear or polynomial system. Every division has a nonzero denominator, every linear system has a coefficient matrix of full column rank, and every polynomial system has the form treated in \Cref{sec:certificate}, with $\theta$ and the unknowns given by fixed polynomial maps of the original coefficients. All polynomial maps, matrix dimensions, and numbers of candidates are fixed.

Let $\params^*$ be the underlying coefficient vector, with $\|\params^*\|_\infty\le R$ and $R\ge1$. Suppose each denominator magnitude, nonzero radicand, branch separation, and certificate or matrix-minor magnitude used in the reconstruction is bounded below by
\begin{align*}
\frac{|p_j(\params^*)|}{Q_j(R)}
\end{align*}
for a fixed nonzero polynomial $p_j$ and a fixed polynomial $Q_j$ positive on $[1,\infty)$. Assume $|p_j(\params^*)|\ge\rho$ for every such polynomial, where $0<\rho\le1$. There are constants $C,a,b>0$, depending only on the reconstruction, such that input errors at most
\begin{align*}
\eta\le C^{-1}R^{-a}\rho^b
\end{align*}
give output error at most $CR^a\rho^{-b}\eta$. Errors are measured in the maximum norm.
\end{lemma}
\begin{proof}
We bound the error at each stage. On $\|\params^*\|_\infty\le R$, every fixed polynomial in the original coefficients and its derivatives is bounded by a polynomial in $R$. The lower bounds in the statement then give polynomial bounds in $R$ and $\rho^{-1}$ for the exact intermediate quantities: division uses a denominator of magnitude at least $\rho/Q_j(R)$, and the other operations preserve such bounds. For sufficiently small input error, the estimated quantities stay within distance one of their exact values, so polynomial evaluations and products also have polynomial error bounds.

For division, choose the propagated error small enough that the estimated denominator differs from the true one by less than half its magnitude. For a square root, the same choice keeps the estimated radicand positive, and
\begin{align*}
|\sqrt u-\sqrt v|=\frac{|u-v|}{\sqrt u+\sqrt v},\qquad u,v>0,
\end{align*}
bounds the error. A branch test makes the same choice as in the noiseless reconstruction once the errors in its comparisons are less than half the branch separation. Each of these steps therefore has a polynomial error amplification factor and an inverse-polynomial error threshold in $R$ and $\rho^{-1}$.

For a linear solve with matrix $M$ and $d$ unknowns, let $M_0$ be a nonsingular $d\times d$ submatrix obtained by selecting rows. Then
\begin{align*}
\sigma_{\min}(M)\ge\frac{|\det M_0|}{\|M_0\|_2^{d-1}},
\end{align*}
where $\|M_0\|_2$ is the operator norm. Since $d$ is fixed, Lemma~\ref{lem:perturbed_least_squares} gives the required polynomial error bound.

For a polynomial solve, apply \Cref{thm:quantitative} with $\Delta_0=\rho/Q_j(R)$ and with $\bar R\ge1$ a fixed polynomial in $R$ bounding the true parameters.
Let $\kappa\ge1$ be the resulting bound on the amplification of residual error into parameter error, after converting to Euclidean norms, so that for the true $\theta,t$
\begin{align*}
\|x-t\|_2\le\kappa\|P_\theta(x)-P_\theta(t)\|_2
\qquad\text{for }\|x\|_\infty\le2\bar R.
\end{align*}
Let $L\ge1$ bound the Jacobian of $P_\theta$, its Lipschitz constant in $x$, and the Lipschitz constants of $P_\theta$ and its Jacobian in $\theta$, on the box where all coordinates have magnitude at most $2\bar R$.
Define the refinement radius
\begin{align*}
r=\min\{1/2,\,1/(16\kappa L)\}.
\end{align*}
Both $\kappa$ and $L$ are polynomially bounded in $R$ and $\rho^{-1}$, and hence so is $r^{-1}$.
These are the error-amplification bound and refinement radius used in \Cref{sec:certificate_refinement}.
If the errors in the residual map and its Jacobian are at most a positive bound $0<\eta'\le r/(16\kappa)$, the coarse search and refinement of \Cref{prop:coarse_refinement}, with target precision $8\kappa\eta'$, returns a solution with error at most $8\kappa\eta'$. 
When the residual map is exact, the exact solve has zero reconstruction error by \Cref{lem:exactcertificate}. 
Finite-precision computation uses a positive accuracy budget as described in \Cref{sec:certificate_refinement}. 
Here $\eta'$ is at most a polynomial in $R$ times the errors in the previously recovered quantities and the right-hand side. 
Since the number of stages is fixed, composing these estimates gives the stated bounds.
\end{proof}

A local reconstruction may leave several allowed sign choices. We apply Lemma~\ref{lem:gaussian_polynomial} to the conditioning polynomials for every choice and use a union bound on their failure probabilities. Lemma~\ref{lem:finite_stage_stability} then applies whichever representative the reconstruction selects. These sign changes are invertible linear transformations of the original coefficients, so they cannot turn a nonzero polynomial into the zero polynomial.

\section{Learning NN Hamiltonians on graphs with uniform control and all-zero input}\label{sec:general_graph}
In this section, we give the detailed protocol for reconstructing generic NN Hamiltonians on every known simple graph $G=(V,E)$ with maximum degree $\Delta$.
We use the orientation convention $J_{vu}=J_{uv}^\top$ and decompose each coupling matrix as
\begin{align*}
J_{uv}=T_{uv}+[\vec d_{uv}]_\times+\lambda_{uv}I,
\quad T_{uv}=T_{uv}^\top,\quad \tr(T_{uv})=0,\quad
T_{vu}=T_{uv},\quad \vec d_{vu}=-\vec d_{uv},
\end{align*}
where $T_{uv}$ represents the traceless symmetric part, $[\vec d_{uv}]_\times$ represents the skew-symmetric part in the cross-product matrix form defined by Eq.~\eqref{eq:cross_matrix}, and $\lambda_{uv}$ represents the trace divided by three.
For a site set $S\subseteq[n]$ and a setting (input-output basis pair) $(\hat a, \hat b)$, as defined in Eq.~\eqref{eq:def_moment}, we let $N_S^{\hat b}=\prod_{v\in S}\hat b\cdot\vec\sigma_v$ and $m_S^{\hat a,\hat b}(t)=\bra{0^n_{\hat a}}e^{iHt}N_S^{\hat b}e^{-iHt}\ket{0^n_{\hat a}}$ be the corresponding moment of the measured $\pm1$ outcomes.
Expectations without a subscript are in $\ket\psi=\ket{0_{\hat a}}^{\otimes n}$. 
Its Taylor expansion is
\begin{align}\label{eq:def_L_Q}
m_S^{\hat a,\hat b}(t)=\mu^{|S|}+tL_S+t^2Q_S+t^3\mathcal{W}_S+O(t^4),
\end{align}
where $L_S$, $Q_S$, and $\mathcal{W}_S$ are the first-, second-, and third-order Taylor coefficients, given by
\begin{align*}
L_S=i\langle[H,N_S^{\hat b}]\rangle,\qquad
Q_S=-\tfrac12\langle[H,[H,N_S^{\hat b}]]\rangle,\qquad
\mathcal{W}_S=-\tfrac i6\langle[H,[H,[H,N_S^{\hat b}]]]\rangle.
\end{align*}
We write $\mu=\hat a\cdot\hat b$, $\hat k=\hat b\times\hat a$, and $\hat l=\hat b-\mu\hat a$.
One global bitstring gives a bounded sample of every required local moment.
In Figures~\ref{fig:overview_uniform}, \ref{fig:graph_first_order}, and \ref{fig:graph_solve_curl}, $A_e=[\vec d_e]_\times$ denotes the skew-symmetric part of an edge coupling, and $\Gamma$ denotes a skew-symmetric matrix added to every edge of a cycle, with its sign reversed on edges oriented against the direction of traversal.

The protocol has three rounds.
First-order Taylor coefficients determine the single-site fields $\vec h_v$, the symmetric traceless couplings $T_{uv}$, and the divergences $\vec\delta_v$.
We then recover the skew-symmetric edge vectors $\vec d_{uv}$ using connected three-site second-order responses and the special solvers for complete connected components.
Finally, connected four-site third-order responses, supplemented by single-vertex second-order responses at vertices of degree at most two, determine the edge traces $\lambda_{uv}$.

\subsection{First-order reconstruction and its complete kernel}\label{sec:general_first}\label{sec:general_curl}
The first round reconstructs the single-site fields $\vec h_v$, the symmetric traceless couplings $T_{uv}$, and the divergences $\vec\delta_v$ from first-order Taylor coefficients.
The following identities give the required linear systems.
\begin{proposition}\label{lem:first_graph}
The first-order Taylor coefficients determine every $T_{uv}$, every $\vec h_v$, and the divergence $\vec\delta_v=\sum_{w\sim v}\vec d_{vw}$. Specifically,
\begin{equation}\label{eq:first_order_compact}
L_v=-2\hat k\cdot\left(\vec h_v+\sum_{w\sim v}J_{vw}\hat a\right),
\qquad L_{uv}-\mu(L_u+L_v)=-4\hat k^\top T_{uv}\hat l.
\end{equation}
When every queried $L$ has error at most $\eta$, fixed settings $\{(\hat a_i,\hat b_i)\}_i$ give linear systems that reconstruct each symmetric traceless coupling $T_{uv}$ with error at most $C\eta$, and each single-site field $\vec h_v$ and divergence $\vec\delta_v$ with error at most $C(\Delta+1)\eta$.
\end{proposition}
\begin{proof}
Eq.~\eqref{eq:first_order_compact} can be derived from the definition in Eq.~\eqref{eq:def_L_Q} using $[\vec x\cdot\vec\sigma,\vec y\cdot\vec\sigma]=2i(\vec x\times\vec y)\cdot\vec\sigma$ with $\vec\sigma=(X,Y,Z)^\top$ and the factorization property of product states. 
For $L_{uv}$ with $u\sim v$, local terms for each site give $\mu(L_u+L_v)$, and the remaining internal-edge correction is $-4\hat k^\top T_{uv}\hat l$.

We then consider the five orthogonal settings
\begin{align*}
(\hat e_x,\hat e_y),\ (\hat e_x,(\hat e_y+\hat e_z)/\sqrt2),\
(\hat e_y,\hat e_x),\ (\hat e_y,(\hat e_x+\hat e_z)/\sqrt2),\ (\hat e_z,\hat e_x).
\end{align*}
If $\mathtt z_1,\ldots,\mathtt z_5$ are the edge coefficients in that order, then
\begin{align*}
T_{yz}=\mathtt z_1/4,\quad 
T_{xz}=-\mathtt z_3/4,\quad 
T_{xy}=\mathtt z_5/4,\quad
T_{xx}=(\mathtt z_2+2\mathtt z_4)/6,\quad 
T_{yy}=-(2\mathtt z_2+\mathtt z_4)/6,\quad
T_{zz}=(\mathtt z_2-\mathtt z_4)/6.
\end{align*}
Let $e_j=\widehat{\mathtt z}_j-\mathtt z_j$ for $j=1,\ldots,5$.
Since these settings are orthogonal, $\mu=0$ and each $\mathtt z_j$ is a queried $L_{uv}$, so $|e_j|\le\eta$.
The reconstruction formulas therefore give:
\begin{align*}
\|\widehat T-T\|_F^2=(e_2^2+e_4^2+e_2e_4)/6+(e_1^2+e_3^2+e_5^2)/8\le7\eta^2/8.
\end{align*}
At an orthogonal setting, subtract this known symmetric part and consider the residual
\begin{align*}
r_v(\hat a,\hat b)=-L_v/2-\sum_{w\sim v}\hat k^\top T_{vw}\hat a=\hat k\cdot\vec h_v+\hat b\cdot\vec\delta_v.
\end{align*}
For each cyclic triple $(i,j,l)$ of $(x,y,z)$, use $(\hat e_j,\hat e_l)$ and $(-\hat e_j,\hat e_l)$. 
Their residuals $r^+,r^-$ satisfy
\begin{align*}
h_{v,i}=(r^--r^+)/2,\qquad \delta_{v,l}=(r^-+r^+)/2.
\end{align*}
Each residual error is at most $(1/2+\deg v)\eta$ given the error on each estimate on $L$ is $\eta$, proving the claim.
\end{proof}

We now describe which Hamiltonian parameters remain undetermined after the first round. 
First, measuring larger sets of sites does not give additional first-order information. For any $|S|\ge2$, the same calculation gives
\begin{align*}
L_S=\mu^{|S|-1}\sum_{v\in S}L_v+\mu^{|S|-2}\sum_{\{u,v\}\in E(S)}\bigl(L_{uv}-\mu(L_u+L_v)\bigr).
\end{align*}
Here $E(S)$ is the set of edges with both endpoints in $S$. 
Thus every first-order moment is determined by the single-site and edge moments at the same setting.

By Proposition~\ref{lem:first_graph}, two Hamiltonians with identical first-order data must have the same fields $\vec h_v$, symmetric traceless couplings $T_{uv}$, and divergences $\vec\delta_v$. 
Conversely, keeping these quantities fixed leaves Eq.~\eqref{eq:first_order_compact} unchanged: the skew-symmetric couplings enter the vertex expression only through $\sum_{w\sim v}\vec d_{vw}=\vec\delta_v$, and the trace terms vanish because $\hat k\cdot\hat a=0$. 
Together with the identity above, this shows that equality of these three collections is both necessary and sufficient for equality of all first-order data.
To express the remaining freedom in $\vec d$, we choose one orientation for each edge as in Eq.~\eqref{eq:graph-ham-general}. 
Define the incidence matrix $B_G\in\mathbb R^{|V|\times|E|}$ to have one row per vertex and one column per oriented edge, with
\begin{align*}
(B_G)_{w,e}=\begin{cases}
+1,& e=(u\to v)\text{ and }w=u,\\
-1,& e=(u\to v)\text{ and }w=v,\\
0,& \text{otherwise}.
\end{cases}
\end{align*}
Stack the edge vectors $\vec d_e$ in the column order of $B_G$ into $\vec d\in\mathbb R^{3|E|}$, and stack the vertex divergences into $\vec\delta\in\mathbb R^{3|V|}$. 
Since $\vec d_{vu}=-\vec d_{uv}$, the divergence equations can be written as
\begin{align*}
\vec\delta=(B_G\otimes I_3)\vec d,
\end{align*}
where $\otimes$ denotes the Kronecker product and $I_3$ is the $3\times3$ identity matrix. 
Write $\Delta\vec h$, $\Delta T$, $\Delta\vec d$, and $\Delta\lambda$ for the differences between the corresponding parameter collections of two Hamiltonians. 
The two Hamiltonians therefore have identical first-order data if and only if
\begin{align*}
\Delta\vec h=0,\qquad \Delta T=0,\qquad
(B_G\otimes I_3)\Delta\vec d=0,\qquad \Delta\lambda\ \hbox{arbitrary}.
\end{align*}
The above equation describes the first-order kernel, namely the space of parameter changes invisible to all first-order measurements.
The condition on $\Delta\vec d$ means that the changes in the axial vectors sum to zero at every vertex. 
For example, adding the same vector around a cycle, with a minus sign on any edge whose chosen orientation opposes the traversal, leaves every divergence unchanged.  To count their dimension, let $c$ be the number of connected components of $G$. 
The equation $B_G^\top x=0$ says that $x$ has equal values at the endpoints of every edge, so $x$ is constant on each connected component and $\operatorname{rank}B_G=|V|-c$. 
Hence the ambiguity has dimension $3(|E|-|V|+c)$, or three parameters per independent cycle. 
Each edge also has an arbitrary trace parameter $\Delta\lambda_e$, giving a total first-order kernel dimension of $|E|+3(|E|-|V|+c)$.
For NN Hamiltonians on any graph $G$, the following two rounds remove these first-order ambiguities.

\begin{figure}[t!]
    \centering
\includegraphics[width=0.99\textwidth]{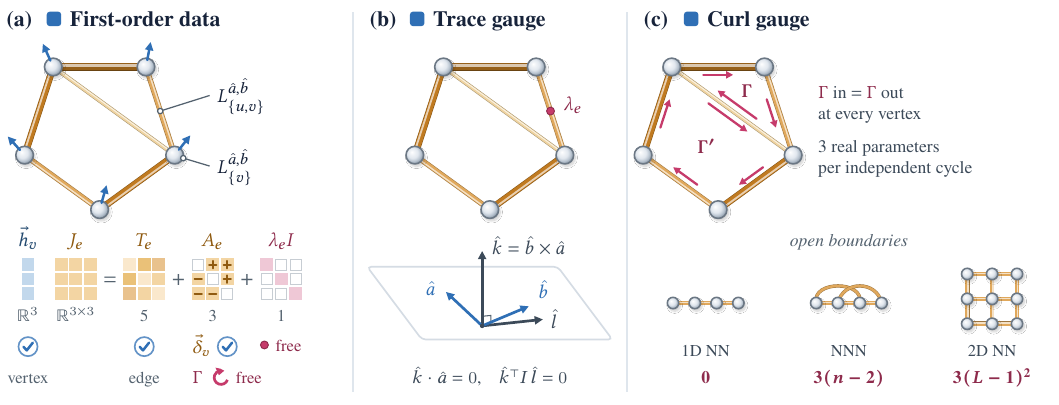}
\caption{\textbf{First-order reconstruction under uniform control.}
\textbf{(a)}~Single-site and two-site responses determine the local fields $\vec h_v$ and the symmetric traceless parts $T_e$ of the couplings. The skew-symmetric parts $A_e$ enter only through their divergence at each vertex.
\textbf{(b)}~The trace terms $\lambda_e I$ are invisible at first order, leaving one undetermined parameter per edge.
\textbf{(c)}~Adding a skew-symmetric matrix around a cycle, with signs following the edge orientations, preserves these divergences: the incoming and outgoing contributions cancel at every vertex. The illustrated graph has two independent cycles, supporting the circulations $\Gamma$ and $\Gamma'$.
Together, the trace and circulation ambiguities account for $|E|+3(|E|-|V|+c)$ undetermined parameters on a graph with $c$ connected components. For the edge couplings on the plot, darker colors indicating stronger interactions, and blue arrows show the local fields.}
\label{fig:graph_first_order}
\end{figure}

\subsection{Removing curl gauges with second-order data}\label{sec:general_second}\label{sec:isolated_triangle}
After the first round, the fields $\vec h_v$, symmetric traceless couplings $T_{uv}$, and divergences $\vec\delta_v$ are known, but the individual skew-symmetric edge vectors $\vec d_{uv}$ are constrained only through their sums at the vertices.
We now show how second-order data remove this ambiguity for generic Hamiltonians on any graph.
We first derive a connected three-site second-order moment that gives linear equations in the skew-symmetric edge vectors. 
We then use these equations to reconstruct the edges of induced paths, triangles, and four-vertex cliques, and show that these local systems cover every component except isolated triangles. 
A separate affine system then handles those components.

\paragraph{Connected path second-order moment calculations.}
For three distinct sites $u,v,w$, take their measured $\pm1$ outcomes at the same setting $(\hat a,\hat b)$ and evolution time $t$.
In terms of the moments defined in Eq.~\eqref{eq:def_L_Q}, we consider the following correlation and its quadratic Taylor coefficient
\begin{align*}
C_{uvw}^{\hat a,\hat b}(t)&=m_{uvw}^{\hat a,\hat b}(t)-m_{uv}^{\hat a,\hat b}(t)m_w^{\hat a,\hat b}(t)-m_{uw}^{\hat a,\hat b}(t)m_v^{\hat a,\hat b}(t)-m_{vw}^{\hat a,\hat b}(t)m_u^{\hat a,\hat b}(t)+2m_u^{\hat a,\hat b}(t)m_v^{\hat a,\hat b}(t)m_w^{\hat a,\hat b}(t),\\
R_{uvw}&=[t^2]C_{uvw}^{\hat a,\hat b}(t).
\end{align*}
The connected correlation subtracts products of lower-order moments, so it vanishes whenever the outcomes split into two independent groups.
Substituting the Taylor expansions from Eq.~\eqref{eq:def_L_Q} gives the following expression for $R_{uvw}$ in terms of the measured first- and second-order coefficients:
\begin{align}\label{eq:connected_residual}
R_{uvw}={}&Q_{uvw}-\mu(Q_{uv}+Q_{uw}+Q_{vw})+\mu^2(Q_u+Q_v+Q_w)\notag\\
&-L_{uv}L_w-L_{uw}L_v-L_{vw}L_u+
2\mu(L_uL_v+L_uL_w+L_vL_w),
\end{align}
which has the following property.
\begin{lemma}[Connected locality and the path identity]\label{lem:Quvw_path}\label{lem:Quvw_triangle}
Without loss of generality, we compute the contribution of the ordered path $u$-$v$-$w$.
Let $s=\sqrt{1-\mu^2}$ and assume $s>0$.
We express $M=J_{uv}$ and $N=J_{vw}$ in the right-handed frame $\hat e_1=\hat l/s$, $\hat e_2=-\hat k/s$, and $\hat e_3=\hat a$, and write
\begin{align*}
\alpha=M_{11}-M_{22}+i(M_{12}+M_{21}),\quad
\beta=N_{11}-N_{22}+i(N_{12}+N_{21}),\quad
p=M_{13}+iM_{23},\quad r=N_{31}+iN_{32}.
\end{align*}
The contribution of these two edges is
\begin{equation}\label{eq:path_B_compact}
\mathcal B(M,N)=2s^3\Re(\alpha r+p\beta)-4\mu s^2\Re(\overline\alpha\beta).
\end{equation}
For $s=0$, this contribution is zero.
If $u$-$v$-$w$ is an induced path, meaning that $u$ and $w$ are not adjacent, then $R_{uvw}=\mathcal B(J_{uv},J_{vw})$.
If the three vertices form a triangle, $R_{uvw}$ is the sum of the three cyclic path contributions.
\end{lemma}
\begin{proof}
We first show why only pairs of distinct edge interactions spanning $\{u,v,w\}$ can contribute to $R_{uvw}$.
Each $Q$ coefficient is quadratic in the Hamiltonian, and each $L$ coefficient is linear. 
Thus Eq.~\eqref{eq:connected_residual} expresses $R_{uvw}$ as a sum of contributions involving two local Hamiltonian terms, possibly the square of the same term.
To isolate the contribution of two distinct terms $K_1,K_2$, consider the scaled Hamiltonian $H_{x,y}=xK_1+yK_2$, with arbitrary real coefficients $x,y$. 
The coefficient of $xy$ in $R_{uvw}$ is the contribution of this pair in the full Hamiltonian. 
For the square of the term, we can consider $xK_1$ and take the coefficient of $x^2$. 
Because $K_1$ and $K_2$ each act on at most two sites, they can connect all three sites $u,v,w$ only if they act on two distinct edges spanning $\{u,v,w\}$.
Otherwise, we can partition the qubits into two sets separating one of $u,v,w$ from the other two, with neither $K_1$ nor $K_2$ acting across the partition.
Evolution under $H_{x,y}$ preserves the product state across this partition, so one measured outcome is independent of the other two.
The correlation $C_{uvw}^{\hat a,\hat b}(t)$ is the mean of the product of the three outcomes after subtracting their respective means, and the above independence makes it zero for all $x,y,t$.
Hence the coefficient of $xy$ in $R_{uvw}$ vanishes. The same argument for $xK_1$ eliminates the coefficient of $x^2$.

Therefore, for the path $u$-$v$-$w$, we can set $H=xH_M+yH_N$, where $H_M=\vec\sigma_u^\top M\vec\sigma_v$ and $H_N=\vec\sigma_v^\top N\vec\sigma_w$.
In the stated frame, the input is $\ket\psi=\ket{000}$ and the measured observable is $sX+\mu Z$.
Subtracting $\mu$ from each outcome leaves $C_{uvw}^{\hat a,\hat b}(t)$ unchanged, so we use $O_j=sX_j+\mu(Z_j-I)$ for $j\in\{u,v,w\}$.
Writing $\langle A\rangle_t$ for the expectation of $A$ under evolution by $H$ and then expanding $e^{-iHt}$, we can obtain
\begin{align*}
[xyt^2]\langle O_uO_vO_w\rangle_t=2\Re\langle H_M\psi|O_uO_vO_w|H_N\psi\rangle-\Re\langle\psi|O_uO_vO_w\{H_M,H_N\}|\psi\rangle.
\end{align*}
Here $\{H_M,H_N\}=2\sum_{i,l}(MN)_{il}\sigma_u^i\sigma_w^l$, so the second term vanishes because $\langle0|O_v|0\rangle=0$.
The edge states are
\begin{align*}
H_M\ket\psi&=M_{33}\ket{000}+p\ket{100}+(M_{31}+iM_{32})\ket{010}+\alpha\ket{110},\\
H_N\ket\psi&=N_{33}\ket{000}+(N_{13}+iN_{23})\ket{010}+r\ket{001}+\beta\ket{011}.
\end{align*}
Using $O_j\ket0=s\ket1$ and $O_j\ket1=s\ket0-2\mu\ket1$ gives
\begin{align*}
\langle H_M\psi|O_uO_vO_w|H_N\psi\rangle=s^3(\overline p\beta+\overline\alpha r)-2\mu s^2\overline\alpha\beta.
\end{align*}

We must also subtract the products of lower-order moments in $C_{uvw}^{\hat a,\hat b}(t)$.
Every nonempty product of the $O_j$'s at distinct sites has zero expectation at $t=0$, so only products of two first-order coefficients contribute at order $t^2$.
The identity $[t]\langle A\rangle_t=2\Im\langle\psi|AH|\psi\rangle$ gives
\begin{align*}
[t]\langle O_u\rangle_t=2xs\Im p,\
[t]\langle O_w\rangle_t=2ys\Im r,\
[t]\langle O_uO_v\rangle_t=2xs^2\Im\alpha,\
[t]\langle O_vO_w\rangle_t=2ys^2\Im\beta,\
[t]\langle O_uO_w\rangle_t=0.
\end{align*}
Subtracting $\langle O_uO_v\rangle_t\langle O_w\rangle_t$ and $\langle O_vO_w\rangle_t\langle O_u\rangle_t$ therefore yields
\begin{align*}
\mathcal B(M,N)&=2s^3\Re(\overline\alpha r+\overline p\beta)-4\mu s^2\Re(\overline\alpha\beta)-4s^3(\Im\alpha\,\Im r+\Im p\,\Im\beta)\\
&=2s^3\Re(\alpha r+p\beta)-4\mu s^2\Re(\overline\alpha\beta),
\end{align*}
where we used the fact that $\Re(\overline z w)-2\Im z\,\Im w=\Re(zw)$.
For $s=0$, we choose any right-handed frame with $\hat e_3=\hat a$, and the same calculation gives zero.
An induced path has one contributing edge pair, and a triangle has three, one for each middle vertex.
\end{proof}

We next use Eq.~\eqref{eq:path_B_compact} to separate the known and unknown coupling parameters.
Note that
\begin{align*}
\mathcal B(I,N)=\mathcal B(M,I)=0,\qquad
\mathcal B(A_1,A_2)=0
\end{align*}
for skew-symmetric matrices $A_1,A_2$. 
Thus the trace terms and the terms involving two skew-symmetric couplings vanish.
For a symmetric traceless matrix $T$, define
\begin{align*}
\Phi_T=2\bigl[(\hat l^\top T\hat l-\hat k^\top T\hat k)\hat k+2(\hat l^\top T\hat k)\hat l\bigr].
\end{align*}
Substitution gives $\mathcal B([\vec d]_\times,T)=-\Phi_T\cdot\vec d$ and $\mathcal B(T,[\vec d]_\times)=\Phi_T\cdot\vec d$.
Since $T$ is known from the first round, subtracting its contribution on an induced path gives the linear equation
\begin{align*}
R_{uvw}-\mathcal B(T_{uv},T_{vw})=-\Phi_{T_{vw}}\cdot\vec d_{uv}+\Phi_{T_{uv}}\cdot\vec d_{vw}.
\end{align*}
For an oriented triangle $u\to v\to w\to u$, the part involving the skew-symmetric edge vectors is
\begin{equation}\label{eq:graph_triangle_linear}
(\Phi_{T_{wu}}-\Phi_{T_{vw}})\cdot\vec d_{uv}
+(\Phi_{T_{uv}}-\Phi_{T_{wu}})\cdot\vec d_{vw}
+(\Phi_{T_{vw}}-\Phi_{T_{uv}})\cdot\vec d_{wu}.
\end{equation}
All remaining terms depend only on the known symmetric traceless couplings and can be subtracted. We therefore obtain linear equations for the skew-symmetric edge vectors without knowing the traces.

\paragraph{Local reconstruction and coverage of the graph.}
We now choose fixed measurement settings for these equations and explain how the resulting systems recover every edge outside isolated triangles.
For each selected path or triangle, use the six settings
\begin{equation}\label{eq:graph_six_path_bases}
(\hat e_x,\hat e_y),\ (\hat e_x,\hat e_z),\ (\hat e_y,\hat e_x),\ (\hat e_y,\hat e_z),\ (\hat e_z,\hat e_x),\ (\hat e_z,\hat e_y).
\end{equation}
An induced path gives a $6\times6$ system in the three coordinates of each of its two skew-symmetric edge vectors.
The coefficients depend linearly on the known $T$'s. To show that the system is generically invertible, it suffices to give one choice of $T$ for which its determinant is nonzero. 
At
\begin{align*}
T_{uv}=\diag(1,2,-3),\qquad
T_{vw}=\begin{pmatrix}2&1&3\\1&-1&4\\3&4&-1\end{pmatrix},
\end{align*}
its matrix and determinant are
\begin{align*}
\begin{pmatrix}
0&16&0&0&0&-10\\0&0&-16&0&-10&0\\
-12&0&-6&0&0&8\\-6&0&12&8&0&0\\
4&6&0&0&2&0\\6&-4&0&2&0&0
\end{pmatrix},\qquad -465600.
\end{align*}
A triangle with two known skew-symmetric edge vectors gives a $6\times3$ system for the third by Eq.~\eqref{eq:graph_triangle_linear}.
If $\vec d_{uv}$ is the unknown vector, its coefficient is $\Phi_{T_{wu}}-\Phi_{T_{vw}}$. At $T_{wu}-T_{vw}=\diag(1,2,-3)$, rows $1,2,4$ have determinant $-800$, so this system also has full column rank generically.

For a four-vertex clique $K_4$, all six edge vectors can be recovered together.
Query its triangles in the order $123,124,134,234$, each in the six settings above. 
Orient the edges from smaller to larger indices and order their variables lexicographically.
This gives a $24\times18$ system, with a nonzero minor at
\begin{align*}
T_{12}=\diag(1,-1,0),\ \ T_{13}=\diag(0,1,-1),\ \
T_{14}=E_{12}+E_{21},\ \ T_{23}=E_{13}+E_{31},\ \
T_{24}=E_{23}+E_{32},\ \ T_{34}=0,
\end{align*}
where $E_{ij}$ has a single $1$ in entry $(i,j)$ and zeros elsewhere.
The $18\times18$ submatrix formed by rows $1,\ldots,14,16,17,18,22$ has determinant $-2^{30}\cdot3^3\cdot5$.
These three explicit specializations show that the path, triangle-completion, and clique systems have full column rank for generic $T$.

It remains to show that these local equation systems recover every edge of any simple graph.
For a vertex $v$, define its closed neighborhood by $N[v]=\{v\}\cup\{w:w\sim v\}$.
An edge $\{u,v\}$ lies in an induced three-site path if and only if $N[u]\ne N[v]$: a vertex in the difference of these sets is adjacent to exactly one endpoint, and conversely the third vertex of such a path witnesses this difference.
For each such edge, select one induced path containing it and solve the corresponding system, and all these systems can be solved in parallel.
The remaining edges have endpoints with the same closed neighborhood. 
Group vertices according to equality of their closed neighborhoods. 
Each such group is a clique.
If a group has a neighbor $w$ outside it, then $w$ is adjacent to every vertex of the group.
For any edge $\{u,v\}$ within the group, $N[w]\ne N[u]=N[v]$, so the path systems have already recovered $\vec d_{uw}$ and $\vec d_{vw}$.
The triangle equation system on $u,v,w$ therefore recovers $\vec d_{uv}$.
These triangle systems can all be solved in parallel using only the estimates from the path systems.

After these triangle solves, only complete connected components of the original graph remain. A single vertex has no edge variables. On a two-vertex component $K_2$, the divergence equation directly gives $\vec d_{uv}=\vec\delta_u$.  A three-vertex component $K_3$ uses the affine system below, while every edge of a larger complete component belongs to a $K_4$ and can be recovered by the $K_4$ clique equation system. Selecting at most one path, triangle-completion, or $K_4$ equation system for each edge, together with one affine system for each isolated triangle, requires $O(n+|E|)$ equation systems. The whole procedure is summarized in \Cref{fig:graph_solve_curl}.

\begin{figure}[t!]
\centering
\includegraphics[width=0.99\textwidth]{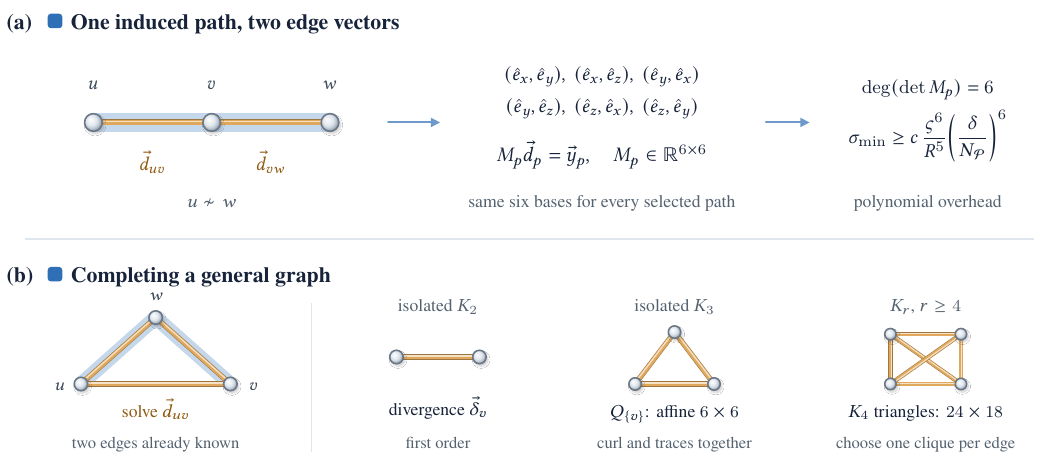}
\caption{\textbf{Local reconstruction of the skew-symmetric edge vectors.}
\textbf{(a)}~For an induced path $u$-$v$-$w$, six fixed settings (basis pairs) give a $6\times6$ system in $(\vec d_{uv},\vec d_{vw})$. 
Its determinant is a nonzero polynomial of degree six, so independent Gaussian smoothing gives the simultaneous lower bound in Eq.~\eqref{eq:graph_path_smoothed_bound}. 
Each selected path is solved independently. 
In particular, every edge of a triangle-free component with at least three vertices has a path.
\textbf{(b)}~For general graphs, first solve every edge lying in an induced path, then recover each remaining edge in a non-complete connected component from a triangle whose other two edges are already known.
Complete connected components use the divergence equation for $K_2$, the six-variable affine vertex system for $K_3$, and four-vertex clique systems for $K_r$, $r\ge4$.}
\label{fig:graph_solve_curl}
\end{figure}

\paragraph{Affine reconstruction on an isolated triangle.}
An isolated triangle has neither an induced three-site path nor a four-vertex clique.
After the first round, its unknown parameters consist of the curl gauge and three edge traces.
We recover these six parameters together from single-vertex second-order moments. 
The following lemma explains why these moments are affine in the remaining unknowns.

\begin{lemma}[The first-order kernel fixes every uniform input]\label{lem:kernel_eigenstate}
Let $G$ be a Hamiltonian perturbation belonging to the first-order kernel described in Appendix~\ref{sec:general_first}.
Then $G\ket\psi=c\ket\psi$ for every uniform product input $\ket\psi$, where $c$ is a scalar, since the skew-symmetric contributions cancel and the swaps in the trace terms leave the uniform input unchanged.
For any reference Hamiltonian $H_0$ and observable $O$, let $Q_O(H)$ denote the quadratic Taylor coefficient of its expectation under evolution by $H$ from this input. Then
\begin{align*}
Q_O(H_0+G)-Q_O(H_0)=-\tfrac12\langle[H_0,[G,O]]\rangle.
\end{align*}
\end{lemma}
\begin{proof}
The perturbation $G$ consists of skew-symmetric couplings whose edge vectors sum to zero at each vertex and arbitrary edge traces. To evaluate its action, let $\mathrm{SWAP}_{uv}$ swap qubits $u$ and $v$. The identities $(\vec\sigma_u\times\vec\sigma_v)_j=i[\sigma_u^j,\mathrm{SWAP}_{uv}]$ and $\mathrm{SWAP}_{uv}\ket\psi=\ket\psi$ give
\begin{align*}
\vec\sigma_u^\top[\vec d_{uv}]_\times\vec\sigma_v\ket\psi =-i\vec d_{uv}\cdot(\vec\sigma_u-\vec\sigma_v)\ket\psi.
\end{align*}
When we sum over the edges, the coefficient at each vertex is its divergence, which is zero by the defining constraint on $G$. 
Thus the skew-symmetric part of $G$ annihilates the input.
Each trace term $\lambda_{uv}\vec\sigma_u\cdot\vec\sigma_v=\lambda_{uv}(2\mathrm{SWAP}_{uv}-I)$ acts as a scalar, proving $G\ket\psi=c\ket\psi$.
It follows that $\langle[G,B]\rangle=0$ for every operator $B$.
Expanding the double commutator that defines $Q_O(H_0+G)$, all terms with $G$ in the outer commutator therefore vanish. 
The remaining mixed term is the displayed expression.
\end{proof}
We apply the lemma to an isolated triangle with vertices $1,2,3$, orienting its edges as $1\to2$, $2\to3$, and $1\to3$.
The known divergence equations are satisfied by the parametrization
\begin{align*}
\vec d_{12}=\vec z,\qquad \vec d_{23}=\vec z+\vec\delta_2,\qquad
\vec d_{13}=\vec\delta_1-\vec z.
\end{align*}
The third divergence equation follows from $\vec\delta_1+\vec\delta_2+\vec\delta_3=0$.
The six unknowns are $\vec x=(z_1,z_2,z_3,\lambda_{12},\lambda_{23},\lambda_{13})$.
Setting these six variables to zero, while retaining the recovered fields, $T$'s, and divergences, defines a known Hamiltonian $H_0$ with the same first-order data.
Thus $H-H_0$ belongs to the first-order kernel, and Lemma~\ref{lem:kernel_eigenstate} makes each difference $Q_O(H)-Q_O(H_0)$ linear in $\vec x$.

Use the six queries $(1;y,x),(1;z,x),(2;x,y),(2;x,z),(3;x,y),(3;z,y)$, where $(v;i,j)$ means measuring the single-vertex coefficient $Q_{\{v\}}$ with input direction $\hat e_i$ and measurement direction $\hat e_j$.
Let $q$ collect these measured coefficients and let $q_0$ collect the corresponding coefficients computed from $H_0$, in the same order.
They give the affine system $B_\triangle \vec x=q-q_0$. The entries of $B_\triangle$ are linear in the known first-order parameters, and those of $q_0$ are quadratic.
To verify generic invertibility, set $T=0$, $\vec\delta=0$, and choose the fields
$\vec h_1=(-1,-1,-1)$, $\vec h_2=(-1,-1,0)$, and $\vec h_3=(-1,0,-1)$. 
The system is then
\begin{align*}
B_\triangle=\begin{pmatrix}
0&-2&0&0&0&0\\0&0&-2&0&0&0\\0&0&0&0&-2&0\\
-2&0&0&2&2&0\\-2&0&0&0&2&2\\0&0&0&0&2&2
\end{pmatrix},
\quad \det B_\triangle=-64,\quad q_0=(2,2,2,0,0,0)^\top.
\end{align*}
The nonzero determinant proves that the system is generically invertible, so these six measurements recover $\vec z$ and all three traces of the isolated triangle.
With noisy data, first project the three estimated divergences onto their zero-sum subspace by subtracting their average from each estimate. 
This makes the parametrization consistent and changes the error by at most a constant factor.

\subsection{Trace recovery from four-site third-order moments and single-site second moments}\label{sec:general_third}
After the second round, the fields and the traceless parts of every edge coupling are known, and the remaining unknowns are the edge traces $\lambda_{uv}$.
We recover them from connected four-site third-order moments when an endpoint has degree at least three, and from single-vertex second-order moments for the remaining edges (see \Cref{fig:graph_trace}).
Each system has at most six unknowns, which is independent of the maximal degree $\Delta$ of the graph $G$.

\begin{figure}[t!]
\centering
\includegraphics[width=0.99\textwidth]{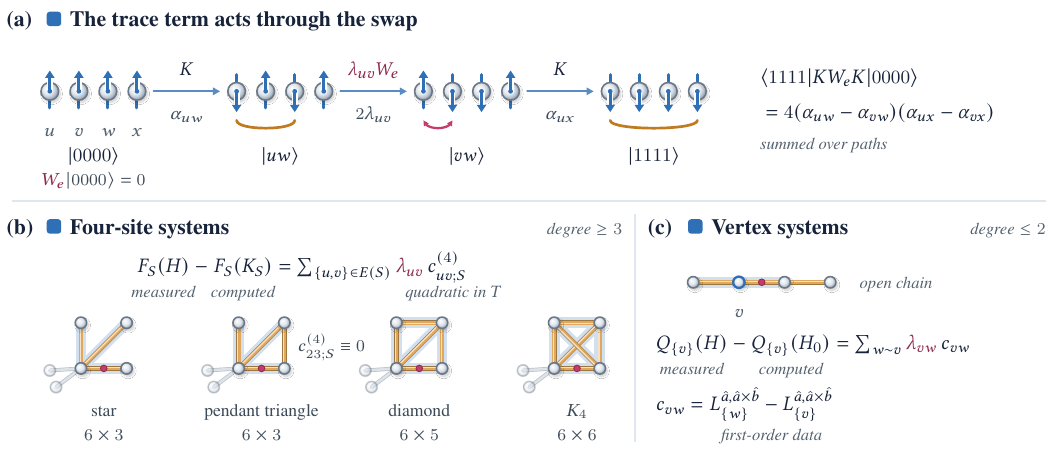}
\caption{\textbf{Recovering the edge traces from local responses.}
\textbf{(a)}~The trace term (crimson) exchanges two spins between pair flips generated by the known couplings (amber), making it visible in the four-site third-order response.
\textbf{(b)}~A vertex and three of its neighbors give a linear system for at most six edge traces. Subtracting the known traceless contribution from the connected third-order response isolates these traces and exterior neighbors do not contribute.
\textbf{(c)}~At vertices with one or two neighbors, single-site second-order responses recover the incident traces after the known contribution is subtracted.
Blue halos mark the edges whose traces are recovered, and the blue rim marks the measured vertex. Crimson dots indicate representative trace terms. See \Cref{lem:graph_four_trace} and Eq.~\eqref{eq:third-round-equations}.}
\label{fig:graph_trace}
\end{figure}

\paragraph{Connected four-site third-order moments.}
Fix four sites $S$ and measure their $\pm1$ outcomes at setting $(\hat a,\hat b)$ at time $t$.
Let $\Pi(S)$ be the set of partitions of $S$ into nonempty blocks, and let $\Pi_j(S)$ contain those partitions with exactly $j$ blocks.
We define the connected four-site correlation and its third-order response by
\begin{equation}\begin{gathered}\label{eq:graph_four_cumulant}
C_S^{\hat a,\hat b}(t)=\sum_{\pi\in\Pi(S)}(-1)^{|\pi|-1}(|\pi|-1)!\prod_{B\in\pi}m_B^{\hat a,\hat b}(t),\qquad
F_S(H)=\frac{(C_S^{\hat a,\hat b})'''(0)}{32}.
\end{gathered}\end{equation}
As in the three-site case, this correlation vanishes whenever the outcomes split into two independent groups.
We use orthogonal settings throughout this subsection such that every nonempty moment is zero at $t=0$.
Substituting the Taylor expansions into Eq.~\eqref{eq:graph_four_cumulant} gives
\begin{equation}\label{eq:graph_four_data}
[t^3]C_S^{\hat a,\hat b}(t)=\mathcal{W}_S-\sum_{\{B,C\}\in\Pi_2(S)}(L_BQ_C+Q_BL_C)
+2\sum_{\{B,C,D\}\in\Pi_3(S)}L_BL_CL_D,
\quad F_S=\frac3{16}[t^3]C_S^{\hat a,\hat b}(t).
\end{equation}
Each unordered partition is counted once.
Thus $F_S(H)$ is computed from the first-, second-, and third-order coefficients on subsets of $S$, obtained from the same measured bitstrings.

\begin{lemma}[Connected four-site moment]\label{lem:graph_four_trace}
Only triples of distinct edges forming a tree on $S$ contribute to $[t^3]C_S^{\hat a,\hat b}(t)$.
For each pair of distinct sites $i,j\in S$, define
\begin{equation*}
\beta_{ij}=\hat b^\top T_{ij}\hat b-\hat k^\top T_{ij}\hat k,
\qquad\gamma_{ij}=-\hat b^\top T_{ij}\hat k,
\qquad\alpha_{ij}=\beta_{ij}+2i\gamma_{ij},
\end{equation*}
with $T_{ij}=0$ when $\{i,j\}$ is not an edge.
Let $K_S$ be the known Hamiltonian consisting of the traceless edge couplings within $S$:
\begin{align*}
K_S=\sum_{\{i,j\}\in E(S)}\vec\sigma_i^\top(T_{ij}+[\vec d_{ij}]_\times)\vec\sigma_j.
\end{align*}
Then
\begin{equation}\label{eq:graph_four_trace}
F_S(H)-F_S(K_S)=\sum_{\{u,v\}\in E(S)}\lambda_{uv}c_{uv;S}^{(4)},
\end{equation}
where, for $\{w,x\}=S\setminus\{u,v\}$,
\begin{align}\label{eq:graph_four_coeff}
c_{uv;S}^{(4)}={}&(\beta_{uw}-\beta_{vw})(\gamma_{ux}-\gamma_{vx})+(\gamma_{uw}-\gamma_{vw})(\beta_{ux}-\beta_{vx})=\tfrac12\Im\bigl[(\alpha_{uw}-\alpha_{vw})(\alpha_{ux}-\alpha_{vx})\bigr].
\end{align}
The coefficients $c_{uv;S}^{(4)}$ depend only on the known symmetric traceless couplings $T$.
The subtraction $F_S(K_S)$ can be computed classically given the previous two rounds.
\end{lemma}
\begin{proof}
As in the proof of Lemma~\ref{lem:Quvw_path}, we consider multiplying each local Hamiltonian term by an independent scalar.
Eq.~\eqref{eq:graph_four_data} is cubic in these scalars, so each coefficient involves at most three local Hamiltonian terms and can be computed with all other terms set to zero.
If those terms do not connect all four sites, the evolution preserves the product input across a partition separating the measured sites into two nonempty groups.
The groups of outcomes are independent, so $C_S^{\hat a,\hat b}(t)$ and the coefficient both vanish.
Three one- or two-site terms can connect all four sites only if they act on three distinct edges within $S$ forming a tree. 
Fields, repeated edges, and interactions leaving $S$ therefore do not contribute.

We can therefore restrict the calculation to the four measured sites.
In the right-handed frame $(\hat b,-\hat k,\hat a)$, the input is $\psi=\ket{0000}$ and the four-site observable is $N_S^{\hat b}=X^{\otimes4}$.
Write $\ket f$ for the computational-basis state with qubits in $f$ in state $1$ and all other qubits in state $0$.
For a two-body Hamiltonian $K$ with symmetric traceless couplings $T$, the amplitude for flipping qubits $i,j$ is $\bra{ij}K\ket\varnothing=\alpha_{ij}$.
The skew-symmetric and trace parts of the edge coupling do not contribute to this amplitude.

We vary the trace on $e=\{u,v\}$.
Subtracting a scalar multiple of the identity gives the interaction $W_e=2(\mathrm{SWAP}_{uv}-I)$, where $\mathrm{SWAP}_{uv}$ swaps qubits $u,v$.
It satisfies
\begin{align*}
W_e\psi=0,\qquad[W_e,N_S^{\hat b}]=0,\qquad W_eN_S^{\hat b}\psi=0.
\end{align*}
Expansion of the state through third order gives
\begin{align*}
\mathcal{W}_S(K)=\Im\bra\psi KN_S^{\hat b}K^2\ket\psi-\tfrac13\Im\bra\psi N_S^{\hat b}K^3\ket\psi.
\end{align*}
Under $K\mapsto K+\lambda W_e$, the first term changes by
$\lambda\Im\bra\psi KN_S^{\hat b}W_eK\ket\psi=0$, since $N_S^{\hat b}W_e$ is Hermitian.
In the second term, every term containing $W_e$ vanishes except $\lambda KW_eK$.
Thus
\begin{align*}
\mathcal{W}_S(K+\lambda W_e)-\mathcal{W}_S(K)
=-\tfrac\lambda3\Im\bra{1111}KW_eK\ket{0000}.
\end{align*}
The swap preserves excitation number, so only intermediate states with two qubits in state $1$ contribute.
Let $\tau_e$ exchange $u,v$, and put
$A=\alpha_{uw}$, $B=\alpha_{vw}$, $C=\alpha_{ux}$, $D'=\alpha_{vx}$.
Since $W_e\ket f=2(\ket{\tau_e(f)}-\ket f)$ and
$\bra{1111}K\ket f=\alpha_{S\setminus f}$ for $|f|=2$, a direct summation gives
\begin{align*}
\bra{1111}KW_eK\ket{0000}=2\{A(C-D')+B(D'-C)+C(A-B)+D'(B-A)\}=4(A-B)(C-D').
\end{align*}
The coefficient of $\lambda$ in $\mathcal{W}_S$ is therefore $-\frac43\Im[(A-B)(C-D')]$.

We must also subtract the lower-order products in Eq.~\eqref{eq:graph_four_data}.
For a two-site set $f$, their coefficients satisfy
\begin{align*}
L_f=2\Im\alpha_f,\qquad
Q_f(K+\lambda W_e)-Q_f(K)=2\lambda\Re(\alpha_f-\alpha_{\tau_e(f)}),
\end{align*}
where the second identity follows from
$Q_f=\langle K\psi|N_f^{\hat b}|K\psi\rangle-\Re\bra\psi N_f^{\hat b}K^2\ket\psi$.
All first-order coefficients are independent of the traces.
For $|f|=3$, both $L_f$ and the trace-dependent part of $Q_f$ vanish because the required three flips cannot occur in these terms.
Therefore, partitions into a singleton and a three-site block, and partitions into three blocks, give no trace contribution in Eq.~\eqref{eq:graph_four_data}.
The pairing $uv\mid wx$ also gives none, since $\tau_e$ leaves both pairs unchanged.
For the two remaining pairings, the coefficient of $\lambda$ in the sum of the pair products is
\begin{align*}
4\{\Im A\Re(D'-C)+\Im D'\Re(A-B)+\Im C\Re(B-A)+\Im B\Re(C-D')\}=-4\Im[(A-B)(C-D')].
\end{align*}
Subtracting these pair products leaves the trace-dependent part of $[t^3]C_S^{\hat a,\hat b}(t)$ equal to $\frac83\lambda\Im[(A-B)(C-D')]$.
Multiplying by $3/16$ gives Eq.~\eqref{eq:graph_four_coeff}.
The other traces in $K$ were arbitrary, but the coefficient depends only on $T$.
Adding the trace terms to $K_S$ one edge at a time proves Eq.~\eqref{eq:graph_four_trace}, with no products of trace variables.
\end{proof}

\paragraph{Local reconstruction and nonzero determinants.}
We use Eq.~\eqref{eq:graph_four_trace} at the six orthogonal settings
\begin{equation}\label{eq:graph_four_bases}
\begin{array}{c|c|c}
r&\hat a&\hat b\\\hline
1&\hat e_z&\hat e_x\\2&-\hat e_y&\hat e_x\\3&\hat e_x&-\hat e_z\\
4&\hat e_z&(3\hat e_x+4\hat e_y)/5\\
5&(-4\hat e_y+3\hat e_z)/5&\hat e_x\\
6&(4\hat e_x+3\hat e_z)/5&(3\hat e_x-4\hat e_z)/5
\end{array}
\end{equation}
Label a vertex of degree at least three by $1$ and three of its neighbors by $2,3,4$.
Up to relabeling the neighbors, the induced graph on these four sites is a star, a triangle with a pendant edge, a diamond, or a $K_4$, with respective edge sets
\begin{align*}
\{12,13,14\},\quad\{12,13,14,23\},\quad
\{12,13,14,23,24\},\quad\{12,13,14,23,24,34\}.
\end{align*}
For the triangle with a pendant edge, $c_{23;S}^{(4)}=0$ identically because $T_{24}=T_{34}=0$.
Thus the equations involve only the three traces on edges incident to vertex $1$.
For each of the other three graphs, we solve for all edge traces within the four-site set.

To prove generic recovery, we exhibit a square submatrix with nonzero determinant.
Use the following symmetric traceless matrices, omitting absent edges:
\begin{align*}
T_{12}&=\diag(1,-1,0),&T_{13}&=\diag(0,1,-1),&T_{14}&=E_{12}+E_{21},\\
T_{23}&=E_{13}+E_{31},&T_{24}&=E_{23}+E_{32},&
T_{34}&=\begin{pmatrix}1&1&1\\1&0&1\\1&1&-1\end{pmatrix}.
\end{align*}
Order the columns lexicographically by edge labels.
For the star and the triangle with a pendant edge, selecting rows $1,4,5$ and columns $12,13,14$ gives
\begin{align*}
M_\star=\begin{pmatrix}
-1&2&0\\
\frac{527}{625}&-\frac{1054}{625}&-\frac{336}{625}\\
\frac{21}{125}&\frac{102}{125}&0
\end{pmatrix},
\qquad\det M_\star=-\frac{48384}{78125}.
\end{align*}
For the diamond, rows $1,\ldots,5$ give
\begin{align*}
M_\diamond=\begin{pmatrix}
-1&2&0&0&0\\0&0&0&0&1\\-2&0&-2&-3&0\\
\frac{527}{625}&-\frac{1054}{625}&-\frac{336}{625}&0&0\\
-\frac35&\frac{102}{125}&0&0&\frac{136}{125}
\end{pmatrix},\qquad\det M_\diamond=\frac{48384}{78125}.
\end{align*}
For the $K_4$, all six rows give
\begin{align*}
M_{K_4}=\begin{pmatrix}
-1&0&-2&-3&-1&0\\0&1&-1&0&0&1\\-2&1&0&0&1&-2\\
\frac{527}{625}&-\frac{336}{625}&\frac{382}{625}&\frac{1077}{625}&-\frac{481}{625}&0\\
-\frac35&-\frac{68}{125}&-\frac{406}{125}&-\frac{189}{125}&-\frac{51}{125}&-\frac{44}{125}\\
-\frac{119}{125}&\frac{232}{125}&-\frac6{125}&-\frac95&\frac{13}{125}&-\frac{92}{125}
\end{pmatrix},\qquad\det M_{K_4}=\frac{349235712}{9765625}.
\end{align*}
These nonzero determinants prove generic recovery in all four cases.
For polynomials in the original independent Pauli coefficients, the selected determinants have degrees $6,6,10,12$, respectively.

\paragraph{Single-vertex responses at degree at most two.}
For the remaining edges, second-order moments at a single vertex suffice.
Remove all trace terms from $H$ to obtain the known Hamiltonian $H_0$, so $H=H_0+V$ with
$V=\sum_{\{v,w\}}\lambda_{vw}(2\mathrm{SWAP}_{vw}-I)$.
The perturbation $V$ acts as a scalar on the uniform input, so Lemma~\ref{lem:kernel_eigenstate} gives
\begin{equation}\label{eq:third-round-equations}
Q_v(H)-Q_v(H_0)=\sum_{w\sim v}\lambda_{vw}c_{vw}.
\end{equation}
Only traces on edges incident to $v$ appear, because all other trace terms commute with $N_{\{v\}}^{\hat b}$.
To compute the coefficient for $\{v,w\}$, we note that $[2\mathrm{SWAP}_{vw}-I,N_{\{v\}}^{\hat b}]=-2(N_{\{v\}}^{\hat b}-N_{\{w\}}^{\hat b})\mathrm{SWAP}_{vw}$ and $(N_{\{v\}}^{\hat b}-N_{\{w\}}^{\hat b})\mathrm{SWAP}_{vw}=-\mathrm{SWAP}_{vw}(N_{\{v\}}^{\hat b}-N_{\{w\}}^{\hat b})$.
We thus have
\begin{align*}
c_{vw}=\langle\{H_0,N_{\{v\}}^{\hat b}-N_{\{w\}}^{\hat b}\}\rangle.
\end{align*}
Writing $S_v=\sum_{t\sim v}T_{vt}$ and
$\vec f_v=\vec h_v+\vec\delta_v\times\hat a+S_v\hat a$, product-state factorization gives
$\langle\{H_0,N_{\{v\}}^{\hat b}\}\rangle=2\mu\langle H_0\rangle+2\hat l\cdot\vec f_v$.
Consequently
\begin{align*}
c_{vw}=2(\vec h_v-\vec h_w)\cdot\hat l+
2(\vec\delta_w-\vec\delta_v)\cdot\hat k+
2\hat l^\top(S_v-S_w)\hat a.
\end{align*}
In an orthogonal setting, this coefficient is a difference of two measured first-order coefficients:
\begin{equation}\label{eq:trace_direct}
c_{vw}(\hat a,\hat b)=L_w^{\hat a,\hat a\times\hat b}
-L_v^{\hat a,\hat a\times\hat b}.
\end{equation}
Thus, first-order data directly determine the trace matrix, and errors in the reconstructed skew-symmetric edge vectors enter only through the subtraction $Q_v(H_0)$.

At degree one, use $(\hat a,\hat b)=(\hat e_z,\hat e_x)$ to solve for the incident trace.
At degree two, use $(\hat e_z,\hat e_x)$ and $(\hat e_z,\hat e_y)$ to solve for both incident traces.
To verify generic invertibility, label the neighbors $w_1$ and, when present, $w_2$.
Set $T=\vec d=0$, $\vec h_v=0$, $\vec h_{w_1}=-\hat e_x/2$, and $\vec h_{w_2}=-\hat e_y/2$ when $w_2$ is present.
The resulting matrix is the identity, regardless of the edges outside this set.
Its determinant is therefore a nonzero polynomial of degree one or two in the original coefficients.

\paragraph{Coverage of all traces.}
For an edge with an endpoint of degree at least three, choose that endpoint as the center and include two of its other neighbors.
The resulting four-site system recovers the edge trace in each of the four graph cases above.
If both endpoints have degree at most two, the single-vertex systems recover the trace instead.
For isolated triangles, we may keep the trace estimates already obtained in the second round.
Selecting at most one system for each edge gives $O(n+|E|)$ equation systems.

\subsection{Conditioning, precision, and complexity}\label{sec:general_full}
We now analyze the propagation of measurement error on the reconstruction.
Under Gaussian smoothing, we first bound from below the smallest singular values of the coefficient matrices in the linear systems for the skew-symmetric edge vectors and edge traces, then choose the measurement precision needed for final error $\epsilon$.
\begin{theorem}\label{thm:graph_ham_uniform_local}
Let $H$ be the NN Hamiltonian \eqref{eq:graph-ham-general} on a known simple graph of maximum degree $\Delta$, with independently Gaussian-smoothed Pauli coefficients.
For any $\varsigma>0$ and $0<\epsilon,\delta<1$, there is a uniform-control protocol that reconstructs every field and edge matrix to error $\epsilon$, with probability at least $1-\delta$.
Its sample complexity, total evolution time, and classical post-processing are
\begin{align*}
\widetilde O\bigl(\epsilon^{-2}\poly(n,\Delta,R,\varsigma^{-1},\delta^{-1})\bigr),
\end{align*}
where $R$ is the coefficient radius from Eq.~\eqref{eq:effective_radius}.
The protocol uses a fixed number of global basis pairs, moments estimated from measuring local observables on at most four sites, and derivatives through order three.
The queried times range from $c/((\Delta+1)Rd^2)$ to $C/((\Delta+1)R)$, where the interpolation degree $d$ is logarithmic in $(\Delta+1)R$ and the inverse measurement precision.
Since $(\Delta+1)\le n$, the polynomial dependence on $\Delta$ may also be absorbed into the dependence on $n$.
\end{theorem}
\begin{proof}
\emph{Conditioning of the linear systems.}
The event that all coefficients are bounded by $R$ has failure probability at most $\delta/4$, by Eq.~\eqref{eq:effective_radius}.
We allocate $\delta/4$ to the conditioning bounds and $\delta/2$ to the measurement estimates, as specified below.
Let $N_0$ be the number of selected equation systems, or $1$ if none are needed, so $N_0\le C(n+|E|+1)$ for some constant $C$, and put $s=\min\{1,\varsigma\}$.
For each equation system, use the square submatrix specified in the preceding subsections.
The determinants for paths, triangle completion, four-vertex cliques, and isolated triangles have degrees $6,3,18,6$, respectively.
The four-site trace determinants have degrees $6,6,10,12$, and the remaining vertex trace determinants have degree at most two.

For each reconstruction system, we have displayed a choice of the symmetric traceless matrices $T$ and fields $\vec h$ at which the selected determinant is nonzero (see Appendices~\ref{sec:general_second} and~\ref{sec:general_third}).
These examples show that none of the determinant polynomials is identically zero.
Except for the vertex trace systems at degree one or two in Appendix~\ref{sec:general_third}, these polynomials form a fixed finite collection up to relabeling sites, so Lemma~\ref{lem:gaussian_polynomial} gives a common positive constant $c_p$ for them.
For the vertex trace systems at degree one or two, the neighbors may have arbitrarily many exterior edges, so we verify this uniformity directly.
Let $w_1$ and, at degree two, $w_2$ be the neighbors of the measured vertex $v$, and write $h_{w,x},h_{w,y}$ for the $x$- and $y$-components of $\vec h_w$.
The formula for $c_{vw}$ in Appendix~\ref{sec:general_third} gives field contributions $2(h_{v,x}-h_{w,x})$ and $2(h_{v,y}-h_{w,y})$ in the settings $(\hat e_z,\hat e_x)$ and $(\hat e_z,\hat e_y)$, respectively.
Set all Pauli coefficients to zero except $h_{w_1,x}$ at degree one, or $h_{w_1,x}$ and $h_{w_2,y}$ at degree two.
The corresponding determinants are
\begin{align*}
\det\begin{pmatrix}-2h_{w_1,x}\end{pmatrix}&=-2h_{w_1,x},\qquad\text{(degree one)},\\
\det\begin{pmatrix}-2h_{w_1,x}&0\\0&-2h_{w_2,y}\end{pmatrix}&=4h_{w_1,x}h_{w_2,y},\qquad\text{(degree two)}.
\end{align*}
Thus the coefficients $-2$ and $4$ of these terms in the full determinant polynomials do not depend on the exterior edges.
These determinants are homogeneous of degrees one and two, respectively.
Taking $h_{w_1,x}=1$ in the first case, or $h_{w_1,x}=h_{w_2,y}=1/\sqrt2$ in the second, gives a unit vector in the original coefficient space at which the absolute value of the determinant is $2$.
The construction of $c_p$ in the proof of Lemma~\ref{lem:gaussian_polynomial} therefore gives a positive lower bound independent of the graph for these systems as well.

Apply Lemma~\ref{lem:gaussian_polynomial} to each selected determinant with failure probability $\delta/(8N_0)$.
On the coefficient-radius event, matrix entries are bounded by $CR$ for the systems recovering the skew-symmetric edge vectors, $CR^2$ for the four-site trace systems, and $C(\Delta+1)R$ for the vertex trace systems.
Combining the determinant bounds with the matrix-norm inequality in Appendix~\ref{sec:constant_size_stability} gives the bound on the smallest singular value over all equation systems mentioned so far of the template $M\vec x=\vec y$:
\begin{equation}\label{eq:graph_conditioning_degree_free}
\sigma_{\min}(M)\ge\kappa,\qquad
\kappa=c\,\frac{s^{18}}{(\Delta+1)R^{17}}\left(\frac{\delta}{N_0}\right)^{18}\le1.
\end{equation}
For a rectangular matrix, the smallest singular value is at least that of the selected square submatrix formed from its rows.
The determinant bounds fail with total probability at most $\delta/8$, by a union bound that does not require independence between overlapping systems.

If any path systems are used, let $N_{\mathcal P}\ge1$ be their number.
Their determinants have degree six and their entries are bounded by $CR$ on the coefficient-radius event.
Applying Lemma~\ref{lem:gaussian_polynomial} with failure probability $\delta/(8N_{\mathcal P})$ per path determinant gives, on that event, the sharper bound
\begin{equation}\label{eq:graph_path_smoothed_bound}
\sigma_{\min}(M_p)\ge\kappa_{\mathcal P},\qquad
\kappa_{\mathcal P}=c\,\frac{\varsigma^6}{R^5}\left(\frac{\delta}{N_{\mathcal P}}\right)^6.
\end{equation}
The additional path determinant bounds fail with total probability at most $\delta/8$.
Thus the two collections of determinant bounds fail with total probability at most $\delta/4$.
Intersecting their success events with the coefficient-radius event gives both displayed singular-value bounds simultaneously with probability at least $1-\delta/2$.

\emph{Propagation of measurement errors.}
Suppose each queried $L,Q,\mathcal{W}$ has error at most $\eta\le1$.
By Proposition~\ref{lem:first_graph}, the first-round errors are at most $C\eta$ in $T$ and $C(\Delta+1)\eta$ in the fields and divergences.
The sum of interaction norms incident to a vertex is at most $C(\Delta+1)R$, so each order-$j$ coefficient used by the protocol has magnitude at most $C_j((\Delta+1)R)^j$.
Eq.~\eqref{eq:connected_residual} therefore has error at most $C(\Delta+1)R\eta$.
The path and clique matrices have error at most $C\eta$, while the isolated-triangle matrix has error at most $C(\Delta+1)\eta$.
The known quadratic terms subtracted from their right-hand sides have error at most $C(\Delta+1)R\eta$.
Lemma~\ref{lem:perturbed_least_squares} then bounds the initial error in the reconstructed skew-symmetric edge vectors by
\begin{align*}
e_1\le C(\Delta+1)^2R\eta/\kappa.
\end{align*}
Each triangle-completion system uses only two of these recovered edge vectors, so the same perturbation bound gives
\begin{align*}
e_2\le C(\Delta+1)^2R^2\eta/\kappa^2.
\end{align*}
This also bounds the field errors and the initial edge-vector errors.
Only these two successive stages are needed, regardless of the graph size.

For a four-site system, the measured expression in Eq.~\eqref{eq:graph_four_data} has error at most $C(\Delta+1)^2R^2\eta$, and its matrix from Eq.~\eqref{eq:graph_four_coeff} has error at most $CR\eta$.
The subtraction $F_S(K_S)$ is a fixed cubic polynomial in the traceless edge matrices within $S$, so its error is at most $CR^2e_2$ when $e_2\le1$.
Lemma~\ref{lem:perturbed_least_squares} therefore bounds the trace error by
\begin{align*}
C\kappa^{-1}((\Delta+1)^2R^2\eta+R^2e_2+R^2\eta).
\end{align*}
At a vertex of degree one or two, Eq.~\eqref{eq:trace_direct} gives matrix error at most $C\eta$.
The double commutator defining $Q_v(H_0)$ contains at most $C(\Delta+1)^2$ ordered pairs of local Hamiltonian terms, so the subtraction has error at most $C(\Delta+1)^2(Re_2+e_2^2)$.
For $e_2\le1$, solving this system gives error at most
$C\kappa^{-1}(\eta+(\Delta+1)^2Re_2+R\eta)$.
Substituting the bound on $e_2$ gives final coefficient error at most $C(\Delta+1)^4R^4\eta/\kappa^3$.
It therefore suffices to choose
\begin{equation}\label{eq:graph_precision_rule}
\eta\le c\,\epsilon\kappa^3/((\Delta+1)^4R^4).
\end{equation}
For a sufficiently small universal $c$, this choice also ensures $e_2\le1$ and bounds every matrix perturbation by $\kappa/2$, as required by Lemma~\ref{lem:perturbed_least_squares}.
Reconstructing $J=T+[\vec d]_\times+\lambda I$ changes the error by at most a fixed norm factor.

\emph{Experimental and classical costs.}
Apply Theorem~\ref{thm:complexity4derivative_estimation_formal}, the explicit derivative-estimation bound underlying Theorem~\ref{thm:main_complexity4derivative_estimation}, proved in Appendix~\ref{sec:dev_2_body_const_deg_graph}.
Here the interaction size is $m=2$, each observable has support $\chi\le4$, the highest derivative order is $r=3$, and the local energy bound is $G_{\mathrm{loc}}=C(\Delta+1)R$.
Estimating the required coefficients to the precision in Eq.~\eqref{eq:graph_precision_rule} takes
\begin{align*}
N_{\rm shot}=\widetilde O(G_{\mathrm{loc}}^6\eta^{-2})
=\widetilde O(\epsilon^{-2}(\Delta+1)^{14}R^{14}\kappa^{-6}).
\end{align*}
The interpolation degree is $d=O(1+\log(2+((\Delta+1)R)^3/\eta))$, and the queried times lie in the stated interval.
Substituting the common conditioning bound from Eq.~\eqref{eq:graph_conditioning_degree_free} gives the conservative explicit shot bound
\begin{align*}
\widetilde O\!\left(\epsilon^{-2}(\Delta+1)^{20}R^{116}s^{-108}(N_0/\delta)^{108}\right).
\end{align*}
The polynomial exponents are therefore independent of the graph degree.
Since each experiment lasts at most $C/((\Delta+1)R)$, the total evolution time is at most $CN_{\rm shot}/((\Delta+1)R)$.

The settings used in all three rounds form a fixed finite collection: the first-order settings, the path and isolated-triangle settings, the six four-site settings, and the degree-one/two vertex settings together with the rotated measurement directions in Eq.~\eqref{eq:trace_direct}.
Let $N_{\mathrm{bas}}$ be the number of these basis pairs.
For each pair, apply Theorem~\ref{thm:complexity4derivative_estimation_formal} to all required moments with $\delta'=\delta/(2N_{\mathrm{bas}})$ and derivative error at most $\eta$.
Dividing by the corresponding factorials gives the same error bound for $L,Q,\mathcal W$.
For every fixed Hamiltonian in the coefficient-radius event, a union bound over the basis pairs bounds the probability of any inaccurate coefficient estimate by $\delta/2$.
At a fixed setting and time, one measured bitstring provides a sample of every required moment.
Simultaneous accuracy therefore costs only a logarithmic factor in $n+|E|$, included in $\widetilde O$, without separate experimental batches for different site sets.

The graph operations take polynomial time.
The constant-size linear solves and four-qubit subtractions require $O(n+|E|)$ operations, the vertex subtractions require at most $O(n(\Delta+1)^2)$, and accumulating the measured moments requires $O((n+|E|+1)N_{\rm shot})$.
The interpolation work is also polynomial, giving the stated classical cost.
Combining the radius, conditioning, and measurement guarantees bounds the overall failure probability by $\delta/4+\delta/4+\delta/2=\delta$.

We summarize the uniform-control reconstruction in Algorithm~\ref{algo:graph_ham_uniform_local}.
\end{proof}

\begin{algorithm}[t!]
\caption{Hamiltonian learning with uniform control and all-zero input}
\label{algo:graph_ham_uniform_local}
\algorithmicrequire{Accuracy $\epsilon$, failure probability $\delta$, smoothing parameter $\varsigma>0$, a graph and Hamiltonian as in Theorem~\ref{thm:graph_ham_uniform_local}, and uniform-control query access.}

\algorithmicensure{Estimates $\{\widehat{\vec h}_v\}_{v\in V}$ and $\{\widehat J_e\}_{e\in\vec E}$ with maximum error at most $\epsilon$.}

\textbf{Choose precision and bases. }Use Eq.~\eqref{eq:graph_precision_rule}, with the effective radius and simultaneous conditioning bounds from the theorem. 
Allocate failure probabilities $\delta/4$, $\delta/4$, and $\delta/2$ to the radius, conditioning, and sampling events, respectively, as in the proof.
Use the first-order bases, the six path and triangle bases in Eq.~\eqref{eq:graph_six_path_bases}, the six vertex queries for each isolated triangle, the six four-site settings in Eq.~\eqref{eq:graph_four_bases}, and the degree-one/two vertex settings with their rotated-output partners from Appendix~\ref{sec:general_third}.

\textbf{First round. }Estimate the required first-order coefficients and recover $\widehat T_e$, $\widehat{\vec h}_v$, and $\widehat{\vec\delta}_v$ using \Cref{lem:first_graph}.

\textbf{Second round. }Estimate the required second-order coefficients and form the path and triangle residuals. 
Solve one induced-path system for every edge lying in such a path. 
For a remaining edge whose equal-closed-neighborhood class has an outside neighbor, use its triangle with that neighbor and the two already solved edges. 
Treat complete connected components by the divergence equation for two vertices, the affine system in Appendix~\ref{sec:isolated_triangle} for three vertices, and the four-vertex clique systems for larger components.

\textbf{Third round. }For each edge with an endpoint of degree at least three, select a four-site set consisting of that endpoint and three neighbors, including the edge. 
Estimate the required coefficients through third order, form Eq.~\eqref{eq:graph_four_data}, subtract $F_S(K_S)$, and solve the selected system in Eq.~\eqref{eq:graph_four_trace}. 
At degree-one/two vertices use the first-order differences in Eq.~\eqref{eq:trace_direct}, subtract $Q_v(H_0)$, and solve Eq.~\eqref{eq:third-round-equations}.
Where more than one local system estimates an edge, retain any one estimate. 
Output $\widehat J_e=\widehat T_e+[\widehat{\vec d}_e]_\times+\widehat\lambda_e I$ and the reconstructed fields.
\end{algorithm}

\section{Learning IXZ Hamiltonians on lattices with \texorpdfstring{$Z$}{Z}-basis access}\label{sec:zbasis}
In this section, we reconstruct IXZ Hamiltonians on chains, square lattices, and cubic lattices using computational-basis preparation and measurement.
We write $H=D+K$, where $D$ is the diagonal part and $K$ is the off-diagonal part in the computational basis.
With arbitrary $Z$-basis input, we first recover the coefficients of $K$ up to their gauge signs, choose one representative, and then recover the coefficients of $D$ in the same representative.
With all-zero input, the isotropic TI protocol instead solves jointly for the parameters left undetermined by second-order data.
The detailed equations and proofs are given in Appendix~\ref{sec:defer_proof}.

The TI protocols use open lattices large enough to contain the bulk and boundary measurement regions specified below.
For spatially varying Hamiltonians, we describe how to complete the reconstruction at the endpoints of an open chain and along the sides and corners of an open rectangle.
Periodic lattices use the same interior equations.
Throughout, we use the smoothing model of Definition~\ref{def:smoothed_hamiltonians} and the coefficient radius $R$ from Eq.~\eqref{eq:effective_radius}.

\paragraph{Local transition probabilities.}\label{sec:local_transition_bridge}
We first explain how to obtain the required equations by measuring only a constant-size region.
The equations are written for transitions between full computational-basis strings, whereas a local measurement also counts transitions with additional flips outside the measured region.
At the fourth order, these additional transitions can be removed by subtracting a quantity computed from $K$.

Fix an input string $z$ and a nonempty flip set $\sigma$.
Choose a measured region $B$ containing every site reached from $\sigma$ by at most two overlapping Hamiltonian terms, and let $p_H^B(t)$ be the probability that the measured bits equal $(z\oplus1_\sigma)|_B$.
Then
\begin{align}\label{eq:local_transition_bridge}
[t^2]p_H^B(t)&=[t^2]p_{H,z\to z\oplus1_\sigma}(t),\nonumber\\
[t^4]\bigl(p_H^B(t)-p_K^B(t)\bigr)
&=[t^4]\bigl(p_{H,z\to z\oplus1_\sigma}(t)-p_{K,z\to z\oplus1_\sigma}(t)\bigr).
\end{align}
To prove these identities, we sum the full-string transition probabilities over the flip sets $\sigma\cup\tau$, where $\tau\subseteq B^c$.
For $\tau\ne\emptyset$, no single Hamiltonian term can produce all these flips: every term meeting $\sigma$ is supported in $B$.
The direct transition amplitude is therefore zero.
Moreover,
$(H^2)_{z\oplus1_{\sigma\cup\tau},z}=(K^2)_{z\oplus1_{\sigma\cup\tau},z}$, because the terms containing $D$ multiply that zero direct amplitude.
Thus these extra transitions contribute nothing at order $t^2$ and contribute equally for $H$ and $K$ at order $t^4$.
This gives both identities.

To implement the measurement, enlarge $B$ by four interaction layers to a preparation region $S$.
Prepare the required input bits on $S$ and measure the output on $B$.
The nested commutators of the output projector are supported in $S$, so their expectations do not depend on the state outside $S$.
They are also unchanged when the Hamiltonian is restricted to terms supported in $S$.
We can therefore compute $[t^4]p_K^B$ using only the recovered couplings on this constant-size region.
When the coefficients have magnitude at most $R\ge1$, errors at most $\eta\le1$ in the recovered couplings change the subtraction by at most $CR^3\eta$, since it is a fixed polynomial of degree four.
Some sign tests below use fourth-order data without this subtraction.
For the four-site flips in the three-local TI IXZ Hamiltonian learning protocol, both terms in a two-step transition must meet $\sigma$, since one term cannot flip all four sites.
Both terms are then supported in $B$, so no additional flip outside $B$ is possible at second order in the amplitude.
The local and full-string fourth derivatives therefore again agree.
For the $XX$ sign correlators in learning TI NN IXZ Hamiltonians, we also include in $B$ one interaction layer around every varied input bit.
A two-step transition with an additional flip outside $B$ must consist of a direct flip of $\sigma$ and a disjoint outside flip.
Its amplitude is twice the product of these two direct amplitudes.
The outside-flip amplitude does not depend on the varied input bits, and the single-site flip amplitude does not depend on the second-neighbor bits used in the correlator.
The Fourier average over these bits therefore removes every extra contribution.
Thus the sign correlators can be again computed from local probabilities.
For the nonadjacent-pair tests in Appendix~\ref{sec:local_z_common}, take $\sigma=\{u,w\}$ with $u,w$ nonadjacent and include both sites and their neighbors in $B$.
Any two NN terms that produce both flips must each meet $\sigma$; otherwise one term would have to flip both nonadjacent sites.
Both terms are supported in $B$, so a transition with an additional flip outside $B$ has neither a direct nor a second-order amplitude.
Its fourth-order probability is therefore zero, and the raw local and full-string fourth derivatives agree for these tests as well.
The diagonal reconstruction equations use the subtracted quantity in Eq.~\eqref{eq:local_transition_bridge}.

\subsection{Isotropic TI NN Hamiltonians from all-zero input}\label{sec:RI_TI_IXZ}

We first consider an isotropic TI NN IXZ Hamiltonian on an open cubic lattice in dimension $\mathfrak D\in\{1,2,3\}$:
\begin{equation}\label{eq:1D_TI_RI_IXZ}
H=\sum_v(h_ZZ_v+h_XX_v)+
\sum_{\{v,w\}}\bigl[J_{ZZ}Z_vZ_w+J_{XX}X_vX_w+
J_{XZ}(X_vZ_w+Z_vX_w)\bigr].
\end{equation}
\begin{theorem}\label{thm:RI_TI_IXZ}
Let $H$ be the Hamiltonian in Eq.~\eqref{eq:1D_TI_RI_IXZ} on an open chain, square lattice, or cubic lattice containing the bulk and boundary measurement regions specified below.
For generic coefficients, second- and fourth-order derivatives of local $Z$ moments from the all-zero input determine the five parameters up to independent common signs on $(h_X,J_{XZ})$ and $(h_Z,J_{ZZ},J_{XX})$.
For any $\varsigma>0$ and $0<\epsilon,\delta<1$, under independent Gaussian smoothing of the five coefficients, there is a protocol that estimates every parameter to error $\epsilon$ up to these signs with probability at least $1-\delta$.
Its sample complexity, total evolution time, and classical post-processing complexity are
\begin{align*}
\widetilde O\bigl(\epsilon^{-2}\poly(R,\varsigma^{-1},\delta^{-1})\bigr).
\end{align*}
The protocol uses only all-zero preparation and computational-basis measurements on constant-size regions, and the bound is independent of the number of sites.
\end{theorem}
We now specify the measured moments.
For a site set $s$, write $\widetilde f(r,0^n,s)$ for the $r$th derivative at zero of $\langle Z^s\rangle$, with $Z^s=\prod_{v\in s}Z_v$ and all-zero input.
Each output bitstring gives a sample of this product, so its derivatives are Fourier combinations of the local output-probability derivatives.
It suffices to prepare zeros on the region supporting the required nested commutators.
The second derivatives for $s$ a bulk vertex and for $s$ an adjacent pair determine $(h_X+2\mathfrak D J_{XZ})^2$ and $J_{XX}^2$.
We supplement them with the four fourth derivatives listed below for different site sets $s$:
\begin{center}
\begin{tabular}{c|llll}
$\mathfrak D$ & bulk 1 & bulk 2 & fourth derivative 3 & fourth derivative 4\\ \hline
1 & $\{i\}$ & $\{i,i+1\}$ & boundary $\{1,3\}$ & boundary $\{1,4\}$\\
2 & $\{r\}$ & $\{r,r+\hat e_x\}$ & bulk $\{r,r+\hat e_x+\hat e_y\}$ & corner $\{(1,1)\}$\\
3 & $\{r\}$ & $\{r,r+\hat e_y\}$ & corner $\{(1,1,1)\}$ & corner unit plaquette
\end{tabular}
\end{center}
The corner plaquette consists of the four sites $(1,1,1),(1,1,2),(1,2,1),(1,2,2)$.
For the bulk measurements, choose sites whose required two-step neighborhoods do not meet the boundary.
We assume that the lattice contains all these measurement regions.
The two second derivatives and four fourth derivatives give the polynomial system
\begin{equation*}
\mathcal P_{\mathfrak D}(h_Z,h_X,J_{ZZ},J_{XZ},J_{XX})
=\hbox{the two second and four fourth derivatives above}.
\end{equation*}
Appendix~\ref{sec:equation_system_TIRI_allzero} gives the coefficients of these equations and proves that, after choosing the signs of $h_X+2\mathfrak D J_{XZ}$ and $J_{XX}$, the remaining parameters have a unique solution for generic coefficients.
Thus there are exactly four generic solutions.
They differ by independent common sign changes of $(h_X,J_{XZ})$ and $(h_Z,J_{ZZ},J_{XX})$, as listed in Table~\ref{tab:1_2and3D_TI_2-local_gauge-configs}.
\begin{table}[htbp]
\centering
\renewcommand{\arraystretch}{1.2}
\begin{tabular}{c|ccccc|c}
\hline
Index & $h_Z$ & $h_X$ & $J_{ZZ}$ & $J_{XZ}$ & $J_{XX}$ & Gauge Type \\
\hline
1 & $\tilde{h}_Z$ & $\tilde{h}_X$ & $\tilde{J}_{ZZ}$ & $\tilde{J}_{XZ}$ & $\tilde{J}_{XX}$ & --- \\
2 & $-\tilde{h}_Z$ & $-\tilde{h}_X$ & $-\tilde{J}_{ZZ}$ & $-\tilde{J}_{XZ}$ & $-\tilde{J}_{XX}$ & (1) \\
3 & $\tilde{h}_Z$ & $-\tilde{h}_X$ & $\tilde{J}_{ZZ}$ & $-\tilde{J}_{XZ}$ & $\tilde{J}_{XX}$ & (2) \\
4 & $-\tilde{h}_Z$ & $\tilde{h}_X$ & $-\tilde{J}_{ZZ}$ & $\tilde{J}_{XZ}$ & $-\tilde{J}_{XX}$ & (1)(2) \\
\hline
\end{tabular}
\caption{The four generic gauge-transformed parameter configurations for 1D, 2D, and 3D isotropic translationally invariant nearest-neighbor IXZ Hamiltonians. The tilded coefficients denote one fixed representative. 
The  gauges are labeled by: (1) time-reversal gauge as stated in \Cref{thm:time-reversal-gauge}, and 
(2) conjugation by $Z^{\otimes n}$, a special case of \Cref{thm:Z_gauge}.}
\label{tab:1_2and3D_TI_2-local_gauge-configs}
\end{table}
The same proof controls errors in the two quantities recovered from second-order data and in the remaining fourth-order equations by Corollary~\ref{cor:certificate_partial_reconstruction}.
Combining this bound with the derivative estimates of Theorem~\ref{thm:main_complexity4derivative_estimation} and the refinement procedure in Appendix~\ref{sec:certificate_refinement} gives the stated $\widetilde O(\epsilon^{-2}\poly(R,\varsigma^{-1},\delta^{-1}))$ cost.

\subsection{TI NN Hamiltonians from arbitrary \texorpdfstring{$Z$}{Z} input}\label{sec:TI_IXZ}
We next allow arbitrary computational-basis input and drop the isotropy assumption.
On a chain, square lattice, or cubic lattice, orient the bonds along the positive coordinate directions and write
\begin{align}\label{eq:2D_TI_z_to_zp_DK}
D&=\sum_r h_ZZ_r+\sum_{r,\alpha}J^\alpha_{ZZ}Z_rZ_{r+\hat e_\alpha},\notag\\
K&=\sum_r h_XX_r+\sum_{r,\alpha}
\bigl(J^\alpha_{XZ}X_rZ_{r+\hat e_\alpha}
+J^\alpha_{ZX}Z_rX_{r+\hat e_\alpha}
+J^\alpha_{XX}X_rX_{r+\hat e_\alpha}\bigr).
\end{align}
Here $\alpha$ ranges over the lattice axes.
In one dimension, we omit the superscript $\alpha$.

\begin{theorem}\label{thm:TI_NN_z_to_zp}\label{thm:1D_TI_z_to_zp}\label{thm:2D3D_TI_z_to_zp}
Let $H=D+K$ be the TI NN IXZ Hamiltonian in Eq.~\eqref{eq:2D_TI_z_to_zp_DK} on an open chain, square lattice, or cubic lattice containing the bulk measurement regions used in Appendix~\ref{sec:TI_Zbasis_proofs}.
For generic coefficients, constant-size $Z$-basis experiments determine the parameters up to independent common signs on
\begin{align*}
(h_X,\{J^\alpha_{XZ},J^\alpha_{ZX}\}_\alpha),\qquad
(h_Z,\{J^\alpha_{ZZ},J^\alpha_{XX}\}_\alpha).
\end{align*}
For any $\varsigma>0$ and $0<\epsilon,\delta<1$, under independent Gaussian smoothing of the TI coefficients, there is a protocol that estimates every parameter to error $\epsilon$ up to these signs with probability at least $1-\delta$.
Its sample complexity, total evolution time, and classical post-processing complexity are
\begin{align*}
\widetilde O\bigl(\epsilon^{-2}\poly(R,\varsigma^{-1},\delta^{-1})\bigr).
\end{align*}
The protocol uses arbitrary $Z$-basis input and derivatives through order four on constant-size regions, and the bound is independent of the number of sites.
\end{theorem}

The proof in Appendix~\ref{sec:TI_Zbasis_proofs} specifies a constant number of local input patterns and linear systems.
Second-order data determine the single-$X$ coefficients and the magnitudes of the $XX$ coefficients, and in two and three dimensions, fourth-order data also fix the relative signs of the $XX$ coefficients along different axes.
After choosing representatives for these signs, we subtract the $K$-only fourth-order coefficients and solve the resulting equations for the diagonal parameters.
The two remaining common signs are generated by conjugation by $Z^{\otimes n}$ and time reversal.

We summarize the 1D protocol in Algorithm~\ref{algo:1D_TI_z_basis} and the 2D/3D protocol in Algorithm~\ref{algo:2D3D_TI_z_basis}.

\begin{algorithm}[t!]
\caption{Reconstructing a 1D TI NN IXZ Hamiltonian with $Z$-to-$Z'$ access}
\label{algo:1D_TI_z_basis}
\algorithmicrequire{Accuracy $\epsilon$, failure probability $\delta$, and smoothing parameter $\varsigma>0$; query access to $U(t)=e^{-iHt}$ for an unknown 1D TI NN IXZ Hamiltonian in the one-dimensional specialization of Eq.~\eqref{eq:2D_TI_z_to_zp_DK}.}

\algorithmicensure{$\epsilon$-accurate estimates of the tuples $(h_X,J_{XZ},J_{ZX})$ and $(h_Z,J_{ZZ},J_{XX})$, up to a global sign for each tuple.}

\textbf{Choose precision. }Use derivative error $\epsilon/\poly(R,\varsigma^{-1},\delta^{-1})$, also below the stability threshold specified in Appendix~\ref{sec:TI_Zbasis_proofs}, with the effective radius from Eq.~\eqref{eq:effective_radius}.

\textbf{Choose the local regions. }Fix neighboring sites $\{i^*,i^*+1\}$ such that the radius-two neighborhood of $i^*$ lies in the chain. 
Let $B$ contain the sites within distance $4$ of this pair, intersected with the finite chain.
Enlarge $B$ by four interaction layers to obtain $S$. Prepare $z_{\mathrm{loc}}=z|_S$ and measure $z'_{\mathrm{loc}}=z'|_B$, as in Eq.~\eqref{eq:local_transition_bridge}.

\textbf{Define measured Taylor coefficients. }For the local output probability $p_H^B$ of Eq.~\eqref{eq:local_transition_bridge},
\begin{align*}
[t^2]p_H^B(t)=\frac{(p_H^B)''(0)}{2!},\qquad
[t^4]p_H^B(t)=\frac{(p_H^B)^{(4)}(0)}{4!}.
\end{align*}

\textbf{Estimate second-order and fourth-order terms. }For the single-site flip $\sigma=\{i^*\}$, estimate $[t^2]p_H^B(t)$ and $[t^4]p_H^B(t)$ for all four choices of $(z_{i^*-1},z_{i^*+1})$, with all other input bits zero.
For the adjacent-pair flip $\sigma=\{i^*,i^*+1\}$, estimate $[t^2]p_H^B(t)$ from the all-zero input.
Here $\sigma$ is the set of bit flips from $z|_B$ to $z'|_B$.

\textbf{Estimate parameters. }Recover $(h_X,J_{XZ},J_{ZX})$ from the four single-flip second-order equations and $|J_{XX}|$ from the adjacent-pair equation. 
Choose one common sign for the former tuple and one sign for $J_{XX}$.
Compute the fourth-order $K$-only local probabilities on $S$ and subtract them from the measured fourth-order coefficients. 
By Eq.~\eqref{eq:local_transition_bridge}, these differences give the displayed diagonal equations. 
Solve them to obtain $(\hat h_X,\hat J_{XZ},\hat J_{ZX})$ and $(\hat h_Z,\hat J_{ZZ},\hat J_{XX})$ up to a global sign for each tuple.
\end{algorithm}

\begin{algorithm}[t!]
\caption{Reconstructing a 2D/3D TI NN IXZ Hamiltonian with $Z$-to-$Z'$ access}
\label{algo:2D3D_TI_z_basis}
\algorithmicrequire{Accuracy $\epsilon$, failure probability $\delta$, and smoothing parameter $\varsigma>0$; query access to $U(t)=e^{-iHt}$ for an unknown 2D TI NN IXZ Hamiltonian in Eq.~\eqref{eq:2D_TI_z_to_zp_DK}, or its 3D TI NN IXZ extension with an additional axis $\hat e_z$.}

\algorithmicensure{$\epsilon$-accurate estimates of the tuples $(h_X,\{J_{XZ}^\alpha,J_{ZX}^\alpha\}_\alpha)$ and $(h_Z,\{J_{ZZ}^\alpha,J_{XX}^\alpha\}_\alpha)$, with $\alpha\in\{x,y\}$ in 2D or $\alpha\in\{x,y,z\}$ in 3D, up to a global sign for each tuple.}

\textbf{Choose precision. }Use derivative error $\epsilon/\poly(R,\varsigma^{-1},\delta^{-1})$, also below the stability threshold specified in Appendix~\ref{sec:TI_Zbasis_proofs}, with the effective radius from Eq.~\eqref{eq:effective_radius}.

\textbf{Choose the local regions. }Fix a bulk site $r^*$ whose radius-two neighborhood is contained in the lattice, and include its nearest neighbors along every axis, including the $z$ axis in 3D.
Let $B$ contain the sites within distance $4$ of this set, intersected with the finite lattice. 
Enlarge $B$ by four interaction layers to obtain $S$. Prepare $z_{\mathrm{loc}}=z|_S$ and measure $z'_{\mathrm{loc}}=z'|_B$.

\textbf{Define measured Taylor coefficients. }For the local output probability $p_H^B$ of Eq.~\eqref{eq:local_transition_bridge},
\begin{align*}
[t^2]p_H^B(t)=\frac{(p_H^B)''(0)}{2!},\qquad
[t^4]p_H^B(t)=\frac{(p_H^B)^{(4)}(0)}{4!}.
\end{align*}

\textbf{Estimate second-order and fourth-order terms. }For the single-site flip $\sigma=\{r^*\}$, estimate $[t^2]p_H^B(t)$ and $[t^4]p_H^B(t)$ over all $16$ nearest-neighbor input patterns in 2D, or all $64$ in 3D, with all other input bits zero.
For each adjacent-pair flip $\sigma=\{r^*,r^*+\hat e_\alpha\}$, estimate $[t^2]p_H^B(t)$ from the all-zero input.
For each pair of distinct axes $\alpha,\beta$, also estimate the single-flip fourth-order correlator in Eq.~\eqref{eq:fourthordertransition}: vary the two input bits at $r^*+2\hat e_\alpha$ and $r^*+2\hat e_\beta$ over their four choices, keep all other input bits zero, and flip only $r^*$ in the output.

\textbf{Estimate parameters. }Use the second-order equations and the correlators above to recover the off-diagonal coefficients, choosing one common sign for the single-$X$ coefficients and one for the $XX$ coefficients.
Subtract the fourth-order $K$-only local coefficients computed on $S$, as in Eq.~\eqref{eq:local_transition_bridge}, and form the five diagonal-product equations in 2D or the nine in 3D from Appendix~\ref{sec:TI_Zbasis_proofs}.
Their solutions determine the diagonal coefficients up to a common sign.
Fix this sign relative to the chosen $XX$ sign by comparing the two predictions for the single-flip fourth-order coefficient at the all-zero input, already included among the measured patterns.
Output $(h_X,\{J_{XZ}^\alpha,J_{ZX}^\alpha\}_\alpha)$ and $(h_Z,\{J_{ZZ}^\alpha,J_{XX}^\alpha\}_\alpha)$ up to a global sign for each tuple.
\end{algorithm}

\paragraph{Known local depolarizing SPAM.}
\label{par:known_depolarizing_spam}
We can also account for known local depolarizing noise in preparation and measurement when arbitrary $Z$-basis input is available.
For depolarizing strength $\lambda<1$, the effect on a computational-basis bit is the classical channel
\begin{align*}
B_\lambda=\begin{pmatrix}1-\lambda/2&\lambda/2\\\lambda/2&1-\lambda/2\end{pmatrix}.
\end{align*}
Tensor products of this channel act on the table of local transition probabilities, with one factor for each noisy input or output bit.
The eigenvalues of $B_\lambda$ are $1$ and $1-\lambda$, and the sum of the absolute values in each row of its inverse is $(1-\lambda)^{-1}$.
If $q$ input and output bits are involved in total, inverting these channels therefore amplifies the error in each table entry by at most $(1-\lambda)^{-q}$ and the sample cost by at most $(1-\lambda)^{-2q}$.
The number $q$ is constant for the local regions used here.
To invert the preparation channel, we use all input choices on the entire region supporting the required nested commutators.

\subsection{TI three-local chains}\label{sec:3TI_IXZ}
We now allow interactions supported within three consecutive sites of a TI chain.
The Hamiltonian $H=D+K$ with $D$ and $K$ representing the diagonal part and the off-diagonal part has
\begin{align}\label{eq:ham_3_local_TI}
\begin{split}
&D=\sum_i \left(h_Z Z_i+J_{ZZ}Z_iZ_{i+1}+J_{ZIZ}Z_iZ_{i+2}+J_{ZZZ}Z_iZ_{i+1}Z_{i+2}\right),\\
&K=\sum_i\Big(h_X X_i+J_{XZ}X_iZ_{i+1}+J_{ZX}Z_iX_{i+1}+J_{XX}X_iX_{i+1}+J_{XIZ}X_iZ_{i+2}+J_{ZIX}Z_iX_{i+2}+J_{ZZX}Z_iZ_{i+1}X_{i+2}\\
&\qquad\qquad+J_{ZXZ}Z_iX_{i+1}Z_{i+2}+J_{XZZ}X_iZ_{i+1}Z_{i+2}+J_{XIX}X_iX_{i+2}+J_{XXZ}X_iX_{i+1}Z_{i+2}+J_{XZX}X_iZ_{i+1}X_{i+2}\\
&\qquad\qquad+J_{ZXX}Z_iX_{i+1}X_{i+2}+J_{XXX}X_iX_{i+1}X_{i+2}\Big).
\end{split}
\end{align}
There are four diagonal parameters and fourteen off-diagonal parameters.
We group them according to whether their Pauli strings contain an odd or even number of $X$ operators:
\begin{align*}
\Xi_o=(h_X,J_{XZ},J_{ZX},J_{XIZ},J_{ZIX},J_{ZZX},J_{ZXZ},J_{XZZ},J_{XXX}),\\
\Xi_e=(J_{XX},J_{XIX},J_{XXZ},J_{XZX},J_{ZXX},h_Z,J_{ZZ},J_{ZIZ},J_{ZZZ}).
\end{align*}
The two gauge transformations allow an independent common sign on each of these two families.
\begin{theorem}\label{thm:TI_3_local}
Let $H$ be the 1D $n$-qubit TI $3$-\gl IXZ Hamiltonian in Eq.~\eqref{eq:ham_3_local_TI} on an open chain containing the bulk measurement regions used in Appendix~\ref{sec:proof_TI_3_local}, with parameters $\Xi=\Xi_e\cup\Xi_o$.
Under independent Gaussian smoothing, there is a protocol that estimates every parameter in $\Xi_o$ and $\Xi_e$ to error $\epsilon$, up to an independent common sign on each tuple, with probability at least $1-\delta$.
Its sample complexity, total evolution time, and classical post-processing complexity are each $\widetilde O(\epsilon^{-2}\poly(R,1/\varsigma,1/\delta))$, using constant-size local $Z$-basis experiments.
\end{theorem}

The proof in Appendix~\ref{sec:proof_TI_3_local} first reconstructs the coefficients that produce a single flip or an adjacent pair of flips, up to a common sign for each group.
After choosing these signs, second-order data leave eight choices for the coefficients that produce a distance-two pair or a triple flip.
Two fourth-order moments distinguish these choices for generic coefficients.
The remaining fourth-order equations are linear in four energy differences and their ten quadratic products, where each energy difference is a linear combination of the four diagonal parameters.
A $14\times14$ linear system determines these quantities and hence the diagonal parameters.
The proof gives explicit coefficient choices at which the signs separate all candidates and the determinant of this system is nonzero.
We summarize the TI three-local reconstruction in Algorithm~\ref{algo:TI_3_local_z_basis}.

\begin{algorithm}[t!]
\caption{Reconstructing a 1D TI $3$-\gl IXZ Hamiltonian with $Z$-to-$Z'$ access}
\label{algo:TI_3_local_z_basis}
\algorithmicrequire{Accuracy $\epsilon$, failure probability $\delta$, and smoothing parameter $\varsigma>0$; query access to $U(t)=e^{-iHt}$ for an unknown 1D TI $3$-\gl IXZ Hamiltonian in Eq.~\eqref{eq:ham_3_local_TI}.}

\algorithmicensure{$\epsilon$-accurate estimates of the tuples $\Xi_o$ and $\Xi_e$, up to a global sign for each tuple.}

\textbf{Choose precision. }Use derivative error $\epsilon/\poly(R,\varsigma^{-1},\delta^{-1})$, also below the stability threshold specified in Appendix~\ref{sec:proof_TI_3_local}, with the effective radius from Eq.~\eqref{eq:effective_radius}.

\textbf{Choose the local regions. }For the flip sets and input patterns below, choose $B$ to contain their two-interaction-layer neighborhoods, intersected with the finite chain. 
Enlarge $B$ by four more interaction layers to obtain $S$. 
Prepare $z_{\mathrm{loc}}=z|_S$ and measure $z'_{\mathrm{loc}}=z'|_B$. 
Both regions have constant size.

\textbf{Define measured Taylor coefficients. }For the local output probability $p_H^B$ of Eq.~\eqref{eq:local_transition_bridge},
\begin{align*}
[t^2]p_H^B(t)=\frac{(p_H^B)''(0)}{2!},\qquad
[t^4]p_H^B(t)=\frac{(p_H^B)^{(4)}(0)}{4!}.
\end{align*}

\textbf{Estimate second-order and fourth-order terms. }Estimate the second-order terms for the four flip sets $\{i^*\}$, $\{i^*,i^*+1\}$, $\{i^*,i^*+2\}$, and $\{i^*,i^*+1,i^*+2\}$, varying the input bits appearing in their transition amplitudes. 
Estimate the fourth-order terms for these flip sets and the input choices in Appendix~\ref{sec:proof_TI_3_local}. 
Also estimate the fourth-order terms from the all-zero input for the two flip sets $\{i^*,i^*+1,i^*+2,i^*+3\}$ and $\{i^*,i^*+1,i^*+2,i^*+4\}$.

\textbf{Estimate parameters. }Recover the off-diagonal candidates from the second-order equations and distinguish them with the two local four-site fourth-order tests. 
Subtract the $K$-only fourth-order local coefficients computed on $S$ from the remaining fourth-order data. Eq.~\eqref{eq:local_transition_bridge} gives the diagonal system in Appendix~\ref{sec:proof_TI_3_local}. 
Solve it and output $\Xi_o$ and $\Xi_e$ up to a global sign for each tuple.
\end{algorithm}

\paragraph{Higher fixed locality.}\label{sec:higher_locality_criterion}
For higher fixed locality, the same method applies if a fixed collection of local derivatives determines the parameters up to finitely many gauge choices and has a certificate that is nonzero in the original coefficient space, as required by Corollary~\ref{cor:certificate_partial_reconstruction}.
We have not established this condition for every higher locality.

\subsection{Spatially varying NN Hamiltonians}\label{sec:2gl_IXZ}
Finally, we allow the fields and edge couplings on a chain or square lattice to vary from site to site.
Orient each nearest-neighbor edge once, with $\vec E$ denoting the resulting set of oriented edges. 
We again write $H=D+K$ with
\begin{equation}\label{eq:local_IXZ_DK}
\begin{split}
D=\sum_{v\in V}h_Z^vZ_v+\sum_{(u\to v)\in\vec E}J_{ZZ}^{uv}Z_uZ_v,\
K=\sum_{v\in V}h_X^vX_v+\sum_{(u\to v)\in\vec E}
\bigl(J_{XZ}^{uv}X_uZ_v+J_{ZX}^{uv}Z_uX_v+J_{XX}^{uv}X_uX_v\bigr).
\end{split}
\end{equation}
All coefficients are real, and each independent coefficient is smoothed separately. 
On a chain, the edge $i\to i+1$ is represented as $J_{AB}^{i,i+1}=J_{AB}^i$, in the notation of Appendix~\ref{sec:proof_1D_local_Z_basis}.

There are two sources of unavoidable signs: Pauli-$Z$ conjugation at each site and the common sign from time reversal. 
Explicitly, choose independent signs $s_v\in\{\pm1\}$ for the sites and one sign $\tau\in\{\pm1\}$ for the whole Hamiltonian. 
All $Z$-basis transition probabilities remain unchanged under
\begin{equation}\label{eq:local_IXZ_signs}
\begin{split}
(h_X^v,h_Z^v)\longmapsto(\tau s_vh_X^v,\tau h_Z^v),
(J_{XZ}^{uv},J_{ZX}^{uv})\longmapsto(\tau s_uJ_{XZ}^{uv},\tau s_vJ_{ZX}^{uv}),
(J_{XX}^{uv},J_{ZZ}^{uv})\longmapsto(\tau s_us_vJ_{XX}^{uv},\tau J_{ZZ}^{uv}).
\end{split}
\end{equation}
Indeed, these transformations give $H\mapsto\tau U_ZHU_Z^\dagger$, where $U_Z=\prod_{v:s_v=-1}Z_v$: conjugation by $Z_v$ changes the sign of each term containing $X_v$, while $H\mapsto-H$ leaves the transition probabilities unchanged because $H$ is real (see Proposition~\ref{thm:unobservable_gauges}). 
Thus the signs of different coefficients cannot be chosen independently: they must obey Eq.~\eqref{eq:local_IXZ_signs}. 
For generic coefficients, their $n+1$ sign choices give $2^{n+1}$ distinct Hamiltonians with identical data.

\begin{theorem}\label{thm:local_NN_z_to_zp}\label{thm:1D_local_z_to_zp}\label{thm:2D_local_z_to_zp}
Let $H$ be a general spatially varying NN Hamiltonian of the form in Eq.~\eqref{eq:local_IXZ_DK} on a chain or square lattice. 
For a chain, assume an open chain of length at least $4$ or a periodic chain of length at least $9$. 
For a square lattice, assume an open rectangle with both side lengths at least $3$ or a periodic lattice with each period at least $9$.
For generic coefficients, local $Z$-basis experiments determine $H$ up to exactly the sign choices in Eq.~\eqref{eq:local_IXZ_signs}.
Under independent Gaussian smoothing, there is a protocol that estimates every coefficient to error $\epsilon$ up to these signs with probability at least $1-\delta$.
It uses constant-size local experiments, and its sample complexity, total evolution time, and classical post-processing complexity are
\begin{align*}
\widetilde O\bigl(\epsilon^{-2}\poly(n,R,\varsigma^{-1},\delta^{-1})\bigr).
\end{align*}
\end{theorem}
The proofs first recover the off-diagonal coefficients throughout the lattice and align their signs using Appendix~\ref{sec:local_z_common}.
They then use this same representative of $K$ in every local subtraction and diagonal solve: four-site windows on a chain in Appendix~\ref{sec:proof_1D_local_Z_basis}, and a vertex with its neighbors on a square lattice in Appendix~\ref{sec:proof_2D_local_Z_basis}.
On a four-site chain window of the 1D chain, eliminating two boundary sums leaves ten polynomial equations in five diagonal variables.
The certificate proves that these five variables are uniquely determined, after which the two eliminated sums are recovered by division.
At an interior square-lattice vertex, eliminating four sums associated with its neighbors leaves seven polynomial equations for the field at that vertex and its four incident $ZZ$ couplings.
The corresponding certificate proves uniqueness for these five parameters.
As the window or center moves through the lattice, the local solutions determine all individual Hamiltonian coefficients.
Arguments similar to those in Appendix~\ref{sec:general_full} bound the propagation of measurement errors and establish the stated resource costs.

\paragraph{Open-boundary reconstruction.}
At a chain endpoint, second-order data do not distinguish the two single-flip coefficients from their interchange.
Fourth-order transitions that flip a nonadjacent pair distinguish these choices and align the signs of all $XX$ couplings, including those incident to boundary vertices.
For the diagonal step, the chain certificate remains nonzero when either outside edge of a window is absent.
On a rectangle, a separate system using three neighbors recovers each vertex on a straight side.
Once those coefficients are known, each corner is recovered from a linear equation.
Appendices~\ref{sec:local_z_common}, \ref{sec:proof_1D_local_Z_basis}, and~\ref{sec:proof_2D_local_Z_basis} give these boundary equations and their conditioning bounds.

The lower bounds of $9$ on the system size in 1D are sufficient to preserve all local equations used in this reconstruction.
The second-order equations use a vertex and its neighbors, or an edge.
The nonadjacent-pair sign tests use two neighbors of a common center and their neighbors, so every site entering those equations is within distance two of that center.
For the diagonal equations, Eq.~\eqref{eq:fourth-order-expansion} shows that subtracting the $K$-only terms leaves only the direct transition and two-step intermediate transitions.
For a single-site or adjacent-pair flip set $\sigma$, every intermediate flip is within distance one of $\sigma$, and the amplitudes and energy differences then involve only coefficients within distance two of $\sigma$.
Thus all equations for a four-site chain window lie within distance three of one of its middle sites, and all equations at a square-lattice center lie within distance two of it.
In a radius-three neighborhood, any coordinate difference has magnitude at most $6$.
Reducing coordinates modulo periods of at least $9$ therefore neither identifies distinct sites nor introduces new nearest-neighbor edges: such an identification requires a nonzero coordinate difference of magnitude at least $9$, and a new edge requires one of magnitude at least $8$.
The local coefficient polynomials and their nonzero certificates are consequently the same as in the interior calculations.
The measurement and preparation regions $B,S$ are formed in the actual periodic graph and may wrap around it, and as a result, Eq.~\eqref{eq:local_transition_bridge} still applies there.

\section{Deferred proofs}\label{sec:defer_proof}
The Hamiltonian reconstruction protocols in this work use short-time derivatives to obtain equations for the unknown Hamiltonian coefficients. In this section, we begin by proving that these derivatives can be estimated accurately from measurements at evolution times independent of the system size, with the cost controlled by the local interaction strength rather than the norm of the full Hamiltonian. This gives the derivative-estimation bounds used in both access models.

We then turn to the $Z$-basis reconstruction proofs from Appendix~\ref{sec:zbasis}. The transition identities in Appendix~\ref{sec:local_z_common} explain how second- and fourth-order data determine the off-diagonal coefficients up to gauge and provide equations for the diagonal coefficients. The subsequent proofs apply these identities to each Hamiltonian family, specify the local measurements, and verify the algebraic conditions needed for uniqueness and stability. For the nonlinear systems, the explicit certificates established here allow us to apply the general error bounds of Appendix~\ref{sec:certificate}.

\subsection{Derivative estimation by stable polynomial interpolation}\label{sec:defer_derivative_estimation}\label{sec:dev_2_body_const_deg_graph}
Here we prove Theorem~\ref{thm:main_complexity4derivative_estimation} by fitting a polynomial to measurements at a finite set of short evolution times. Locality bounds the Taylor remainder independently of the total system size. A suitable choice of sampling times controls how measurement errors affect the derivatives extracted from the fitted polynomial. We keep the local interaction strength explicit, since it sets both the evolution times and the sampling cost. The argument applies to all the Hamiltonian families considered here, with interaction size $m=2$ for the NN models and $m=3$ for the three-local chain.

\begin{theorem}[Simultaneous local derivative estimation]\label{thm:complexity4derivative_estimation_formal}\label{thm:bound_d_poly}\label{thm:main_bound_d_poly}
Let $m,\chi$ and $r$ be fixed positive integers, and let $H=\sum_AH_A$ be a Hamiltonian on $n$ qubits, where each Hermitian term $H_A$ is supported on a nonempty set $A\subseteq\{1,\ldots,n\}$ with $|A|\le m$. Suppose a known bound $g\ge0$ satisfies
\begin{equation}\label{eq:ham_local_bounded}
\sum_{A:A\cap B\ne\varnothing}\|H_A\|\le g|B|
\qquad\text{for every }B\subseteq\{1,\ldots,n\},
\end{equation}
where $\|\cdot\|$ is the operator norm. Equivalently, $\max_v\sum_{A\ni v}\|H_A\|\le g$.
Let $\rho$ be a density matrix that can be prepared repeatedly. For an integer $K\ge1$, let $O_1,\ldots,O_K$ be Hermitian observables, each supported on at most $\chi$ qubits and satisfying $\|O_k\|\le1$. Define
\begin{align*}
p_k(t)=\tr(\rho e^{iHt}O_ke^{-iHt}),\qquad
q=\max\{1,m-1\},\qquad G_{\mathrm{loc}}=\max\{1,qg\}.
\end{align*}
Assume that, after evolution for any requested time $t$, one measurement setting supplies real outcomes $Y_1,\ldots,Y_K\in[-1,1]$ with $\mathbb E[Y_k]=p_k(t)$. Repetitions of this experiment are independent, while the $K$ outcomes within one repetition need not be independent.

For $0<\eta,\delta'<1$, there is an estimator satisfying
\begin{align*}
\Pr\!\left\{
\max_{\substack{1\le k\le K\\0\le s\le r}}
\bigl|\widehat f_{k,s}-p_k^{(s)}(0)\bigr|\le\eta
\right\}\ge1-\delta',
\end{align*}
where the probability is over the measurement outcomes. The estimator uses a polynomial interpolant of integer degree
\begin{align*}
d\ge\max\{2,r\},\qquad
d=O_r\bigl(1+\chi+\log(2+G_{\mathrm{loc}}^r/\eta)\bigr),
\end{align*}
and a total of at most
\begin{align*}
N\le C_rG_{\mathrm{loc}}^{2r}\eta^{-2}d^{4r+1}\log^2(d+1)
\log\frac{2K(d+1)}{\delta'}
\end{align*}
shots in this setting, summed over all queried times. Here $C_r$ and the constant in $O_r$ depend only on $r$. Every queried time lies in
\begin{align*}
\left[\frac{1}{12G_{\mathrm{loc}}d^2},\frac{2+d^{-2}}{12G_{\mathrm{loc}}}\right].
\end{align*}
Accumulating the supplied outcomes and computing the interpolants requires $O(KN+Kd^2)$ classical arithmetic operations. For a single observable, take $K=1$.
\end{theorem}
\begin{proof}
We first bound the Taylor coefficients. Write $\operatorname{ad}_H^j(O_k)$ for the $j$-fold nested commutator with $H$, with $\operatorname{ad}_H^0(O_k)=O_k$. In its expansion into ordered interaction terms, a nonzero commutator can only add an interaction that meets the support already present. Each such interaction adds at most $m-1$ sites. Thus, after $j$ commutators, each nonzero term has support size at most $\chi+qj$, and Eq.~\eqref{eq:ham_local_bounded} bounds the sum of the norms of its possible next interactions by $g(\chi+qj)$. Iterating $\|[A,B]\|\le2\|A\|\|B\|$ and using $\chi+qj\le q(\chi+j)$ gives
\begin{align*}
\|\operatorname{ad}_H^j(O_k)\|
\le(2qg)^j\chi(\chi+1)\cdots(\chi+j-1)
=(2qg)^j j!\binom{\chi+j-1}{j},
\end{align*}
with the empty product interpreted as $1$ when $j=0$.

Rescale time by setting $u=12G_{\mathrm{loc}}t$, and let $P_{k,d}$ be the degree-$d$ Taylor polynomial at zero of $F_k(u)=p_k(u/(12G_{\mathrm{loc}}))$. Since $p_k^{(j)}(0)=i^j\tr(\rho\operatorname{ad}_H^j(O_k))$, the commutator bound implies
\begin{align*}
\|F_k-P_{k,d}\|_{[0,3]}
&\le\sum_{j>d}2^{-j}\binom{\chi+j-1}{j}\\
&\le(3/4)^{d+1}\sum_{j\ge0}(2/3)^j\binom{\chi+j-1}{j}
=3^\chi(3/4)^{d+1}.
\end{align*}
Here $\|\cdot\|_I$ denotes the supremum norm on the interval $I$, and the last equality uses the negative-binomial series. The estimate holds uniformly in the system size. Moreover, for $0\le s\le d$,
\begin{align*}
P_{k,d}^{(s)}(0)=(12G_{\mathrm{loc}})^{-s}p_k^{(s)}(0).
\end{align*}

We next specify the positive evolution times at which to estimate the expectations. For $d\ge\max\{2,r\}$ and $a=d^{-2}$, choose the rescaled sampling times
\begin{align*}
u_j=a+1-\cos(j\pi/d),\qquad j=0,\ldots,d,
\end{align*}
and perform the experiments at $t_j=u_j/(12G_{\mathrm{loc}})$. The points $u_j$ lie in $[a,a+2]\subseteq[0,3]$ and are the standard Chebyshev points on this interval, including its endpoints. Their closer spacing near the endpoints keeps the amplification of errors in the sampled values logarithmic in $d$. To state this precisely, let $\ell_j$ be the degree-$d$ polynomial that is $1$ at $u_j$ and $0$ at the other sampling times. The factor controlling this amplification, called the Lebesgue constant, satisfies
\begin{align*}
L_d\coloneqq\max_{u\in[a,a+2]}\sum_{j=0}^d|\ell_j(u)|
\le C\log(d+1)
\end{align*}
for an absolute constant $C$. Indeed, interpolating errors $e_j$ gives the polynomial $\sum_j e_j\ell_j(u)$, whose absolute value on the interval is at most $L_d\max_j|e_j|$. Let $\widehat P_{k,d}$ interpolate the empirical means at the sampling times $u_j$. If the Taylor remainder and every empirical mean error are at most $\nu$, then each interpolated value differs from $P_{k,d}$ by at most $2\nu$. Consequently,
\begin{align*}
\|\widehat P_{k,d}-P_{k,d}\|_{[a,a+2]}\le2L_d\nu.
\end{align*}

To control the derivatives at zero, consider any polynomial $v$ of degree at most $d$. Repeated application of Markov's inequality on the length-two interval $[a,a+2]$ gives
$|v^{(j)}(a)|\le d^{2j}\|v\|_{[a,a+2]}$.
Taylor expansion of $v^{(s)}$ from $a$ to zero therefore yields
\begin{align*}
|v^{(s)}(0)|
&\le\|v\|_{[a,a+2]}
\sum_{j=s}^d\frac{a^{j-s}d^{2j}}{(j-s)!}\\
&\le e d^{2s}\|v\|_{[a,a+2]},\qquad 0\le s\le r.
\end{align*}
The choice $a=d^{-2}$ keeps the extrapolation factor bounded by $e$; the remaining factor $d^{2s}$ is the cost of taking $s$ derivatives.

Set
\begin{align*}
\nu=\frac{\eta}{4e(12G_{\mathrm{loc}})^r d^{2r}L_d}
\end{align*}
and choose $d$ large enough that $3^\chi(3/4)^{d+1}\le\nu$. The exponential decay of the remainder dominates the factor $d^{2r}L_d$, so this condition holds for
$d=O_r(1+\chi+\log(2+G_{\mathrm{loc}}^r/\eta))$.
Define the derivative estimates by
\begin{align*}
\widehat f_{k,s}=(12G_{\mathrm{loc}})^s\widehat P_{k,d}^{(s)}(0).
\end{align*}
On the event that all empirical mean errors are at most $\nu$, applying the preceding derivative bound to $\widehat P_{k,d}-P_{k,d}$ shows that, simultaneously for every $k$ and $0\le s\le r$,
\begin{align*}
|\widehat f_{k,s}-p_k^{(s)}(0)|
\le2e(12G_{\mathrm{loc}})^s d^{2s}L_d\nu\le\eta.
\end{align*}

It remains to ensure that the empirical means are accurate. At each evolution time $t_j$, take
\begin{align*}
M=\left\lceil2\nu^{-2}\log\frac{2K(d+1)}{\delta'}\right\rceil
\end{align*}
independent shots. Hoeffding's inequality and a union bound over the $K(d+1)$ means give the required event with probability at least $1-\delta'$. No independence between observables in the same shot is needed, and this one event controls all derivative orders. Substituting the choice of $\nu$ and the bound on $L_d$ into $N=(d+1)M$ proves the stated shot count.

Finally, the chosen evolution times $t_j=u_j/(12G_{\mathrm{loc}})$ lie in the stated interval and are all at most $1/(4G_{\mathrm{loc}})$. The total evolution time is therefore at most $N/(4G_{\mathrm{loc}})$. Accumulating the $K$ outcomes costs $O(KN)$ arithmetic operations, and direct polynomial interpolation and extraction of the fixed number of derivatives cost $O(Kd^2)$.
\end{proof}

In our reconstruction protocols, the observables are constant-size parities or local output projectors, whose outcomes can be computed from a measured bitstring in constant time. Within each preparation and readout setting, the same shots therefore supply many local observables, with only a logarithmic dependence of the shot count on their number. Taking $r=3$ provides the moment coefficients used in uniform-control reconstruction, while $r=4$ provides the transition-probability coefficients used in the $Z$-basis protocols; in both cases the Taylor coefficient is $p_k^{(s)}(0)/(s!)$. The logarithmic interpolation degree also explains the time-resolution bound in Appendix~\ref{sec:derivative_estimation}: at fixed local interaction strength and support size, the shortest evolution time decreases only as an inverse polylogarithm of the required precision. These estimates concern the full Hamiltonian evolution, even though the observables are local.

\subsection{Transition identities and gauge fixing}\label{sec:local_z_common}
The $Z$-basis protocols in Appendix~\ref{sec:zbasis} use short-time transitions to recover both off-diagonal couplings and diagonal energies.
Second-order data give the squared amplitudes of direct transitions, leaving signs that must be resolved or fixed by a gauge choice.
Fourth-order data also contain interference between different transition paths and therefore provide information about these signs and the diagonal coefficients.
We first derive an identity that separates the two contributions, then describe the sign choices and boundary tests used for the NN models.
The same identity underlies the model-specific reconstruction equations in the remaining subsections.

Write $H=D+K$ in the computational basis, with $D$ diagonal and $K$ having zero diagonal, and let $E_s=\langle s|D|s\rangle$ be the energy of a bitstring $s\in\{0,1\}^n$ under $D$.
For distinct input and output strings $a,b\in\{0,1\}^n$, define
\begin{align*}
p_{H,a\to b}(t)&=|\langle b|e^{-iHt}|a\rangle|^2,\\
A&=K_{b,a}=\langle b|K|a\rangle,\qquad \Delta=E_a-E_b.
\end{align*}
Thus $A$ is the direct transition amplitude, while $\Delta$ is the diagonal energy difference between the two strings.
We use the derivative notation $f(r,a,b)=p_{H,a\to b}^{(r)}(0)$; the corresponding Taylor coefficient is $[t^r]p_{H,a\to b}(t)=f(r,a,b)/(r!)$.

An IXZ Hamiltonian is real symmetric, so all matrix elements of its powers are real.
Expanding the amplitude through third order gives
\begin{align*}
\langle b|e^{-iHt}|a\rangle
=-itA-\frac{t^2}{2}(H^2)_{b,a}+\frac{it^3}{6}(H^3)_{b,a}+O(t^4).
\end{align*}
Its squared modulus has second-order coefficient $A^2$, so $f(2,a,b)=2A^2$.
At fourth order, the squared second-order amplitude and the interference between the first- and third-order amplitudes give
\begin{align*}
[t^4]p_{H,a\to b}(t)
&=\frac14(H^2)_{b,a}^2-\frac13A(H^3)_{b,a},\\
f(4,a,b)&=6(H^2)_{b,a}^2-8A(H^3)_{b,a}.
\end{align*}
To expose the diagonal dependence, expand the two matrix powers using $H=D+K$:
\begin{align*}
(H^2)_{b,a}&=(E_a+E_b)A+(K^2)_{b,a},\\
(H^3)_{b,a}&=A(E_a^2+E_aE_b+E_b^2)
+(E_a+E_b)(K^2)_{b,a}\\
&\quad +(KDK)_{b,a}+(K^3)_{b,a}.
\end{align*}
Here $(K^2)_{b,a}=\sum_s K_{b,s}K_{s,a}$ and $(KDK)_{b,a}=\sum_s K_{b,s}E_sK_{s,a}$, where the sums run over computational-basis strings.
Only intermediate strings connected to both $a$ and $b$ by nonzero matrix elements of $K$ contribute.
Substituting these expressions and collecting the energy terms yields
\begin{equation}\label{eq:fourth-order-expansion}
f(4,a,b)=-2A^2\Delta^2+
4A\sum_sK_{b,s}K_{s,a}(E_a+E_b-2E_s)
+6(K^2)_{b,a}^2-8A(K^3)_{b,a}.
\end{equation}
The first two terms contain all dependence on $D$: a squared energy difference for the direct transition, and an energy-weighted sum over two-step paths.
The last two terms are exactly the fourth derivative for evolution under $K$ alone.
Once a representative of the off-diagonal coefficients has been recovered, we can compute and subtract this $K$-only contribution.
The residual is at most quadratic in the diagonal coefficients and gives the diagonal reconstruction equations below.
Experimentally, this subtraction uses local probabilities: Eq.~\eqref{eq:local_transition_bridge} identifies $[t^4](p_H^B-p_K^B)$ with the full-string residual divided by $4!$.
Thus the identity supplies the required equations without measuring the entire output string.

Before using this subtraction to recover the diagonal coefficients, we must choose a consistent representative of $K$. For the NN models, we first recover the coefficients that flip a single site and the magnitudes of the $XX$ couplings from second-order data. Fourth-order interference then determines the relative signs of the $XX$ couplings. We give one test for the bulk TI models and a second test using nonadjacent sites that also works at open boundaries. The latter resolves the additional ambiguity at a chain endpoint and allows the local sign choices to be combined across the lattice. We conclude with a general statement about recovering an amplitude from its square, which supplies the corresponding second-order reconstruction step for the three-local chain.

\paragraph{Single-flip blocks.}
Fix a vertex $v$ with neighbors $v_1,\ldots,v_q$. Write $e_v$ for the bitstring with a single $1$ at $v$, so that $z\oplus e_v$ differs from the input $z$ only at that site. The terms producing this flip are the transverse field $X_v$ and the mixed couplings $X_vZ_{v_j}$. If $a=h_X^v$, $b_j$ is the coefficient of $X_vZ_{v_j}$, and $s_j=(-1)^{z_{v_j}}$, their combined amplitude is
\begin{align*}
A_v(z)\coloneqq K_{z\oplus e_v,z}=a+\sum_{j=1}^q b_js_j.
\end{align*}
We call $(a,b_1,\ldots,b_q)$ the single-flip block at $v$. Second-order data give $A_v(z)^2$. Varying the neighboring input bits determines its Boolean Fourier coefficients, which can be read directly from
\begin{align*}
A_v(z)^2=a^2+\sum_jb_j^2+2\sum_jab_js_j
+2\sum_{j<l}b_jb_ls_js_l.
\end{align*}
Thus we know $a^2+\sum_jb_j^2$ and all products $ab_j$ and $b_jb_l$. When $q\ge2$ and these coefficients are nonzero, the products determine the block up to a common sign: for example,
\begin{align*}
a^2=\frac{(ab_1)(ab_2)}{b_1b_2},\qquad
b_j=\frac{ab_j}{a}.
\end{align*}
Conjugation by $Z_v$ changes the sign of the entire block, so we fix this gauge by choosing $a=h_X^v>0$. At a degree-one endpoint, the data give only $a^2+b_1^2$ and $ab_1$, leaving an interchange of $a$ and $b_1$ as well as the common sign. We resolve this extra ambiguity below.

Flipping the two endpoints of an edge directly gives the squared $XX$ coupling, and hence its magnitude. To determine relative signs in the TI models, choose a bulk site $v$ and two distinct lattice directions $j,l$. Let $v_j,v_l$ be the neighboring sites in those directions, and let $w_j,w_l$ be the next sites along the same respective directions. All four sites must be present. Denote by $k_j$ the coefficient of $X_vX_{v_j}$ and by $c_j^{\rm out}$ the coefficient of $X_{v_j}Z_{w_j}$, and define $k_l,c_l^{\rm out}$ similarly. The two $c^{\rm out}$ coefficients are already known from the neighboring single-flip blocks. Let $z(s,t)$ have $Z$ signs $s,t\in\{\pm1\}$ at $w_j,w_l$ and all other input bits zero. The two-variable Fourier coefficient of the single-flip fourth derivative is
\begin{equation}\label{eq:fourthordertransition}
\frac14\sum_{s,t=\pm1}st\,f(4,z(s,t),z(s,t)\oplus e_v)
=48k_jk_l c_j^{\rm out}c_l^{\rm out}.
\end{equation}
To see why this average isolates the desired product, consider Eq.~\eqref{eq:fourth-order-expansion}. On the chain, square lattice, and cubic lattice, the diagonal-dependent terms and $(K^3)_{z\oplus e_v,z}$ contain no term involving both outer signs $s$ and $t$, so the Fourier average removes them. In contrast, $(K^2)_{z\oplus e_v,z}$ contains $2k_jc_j^{\rm out}s$ and $2k_lc_l^{\rm out}t$, from flipping an incident pair and then its neighboring site, in either order. Their cross term in $6(K^2)_{z\oplus e_v,z}^2$ gives the factor $48$. Appendix~\ref{sec:local_transition_bridge} shows that the same Fourier coefficient is obtained from local output probabilities.

Dividing by the known, generically nonzero $c_j^{\rm out}c_l^{\rm out}$ determines $k_jk_l$, fixing the relative $XX$ signs among the TI axes. Near an open boundary, one of the required outer sites may be absent. The next identity provides a sign test that does not require these extensions of the incident edges.

\paragraph{Nonadjacent-pair identity.}\label{sec:nonadjacent_pair_boundary}
Consider two distinct, nonadjacent vertices $u,w$, and the transition from $z$ to $z\oplus e_u\oplus e_w$. A single NN term cannot flip both sites, so this transition first appears through two applications of $K$. There are two ways to produce it: flip $u$ and $w$ separately, or flip the two edges of a path $u-v-w$. The two mechanisms interfere, making this transition sensitive to products of $XX$ couplings.

Let $N(u)$ denote the neighbors of $u$, and write $k_{uv}=J_{XX}^{uv}$ for the coefficient of $X_uX_v$ on an edge. Using the single-flip amplitudes $A_v(z)$ defined above, set
\begin{align*}
T_{uw}\coloneqq\sum_{v\in N(u)\cap N(w)}k_{uv}k_{vw},\qquad
Q_{uw}(z)\coloneqq\frac{f(4,z,z\oplus e_u\oplus e_w)}{24}.
\end{align*}
Here $T_{uw}$ sums the edge products over all common neighbors, and $Q_{uw}(z)$ is the measured fourth-order Taylor coefficient. They satisfy
\begin{equation}\label{eq:nonadjacent_pair}
Q_{uw}(z)=\bigl(A_u(z)A_w(z)+T_{uw}\bigr)^2.
\end{equation}
Indeed, flipping $u$ does not change $A_w(z)$, and flipping $w$ does not change $A_u(z)$, because the sites are not adjacent. The two orders of the single flips therefore contribute $2A_u(z)A_w(z)$ to $(K^2)_{z\oplus e_u\oplus e_w,z}$. For each common neighbor $v$, the two orders of the pair flips contribute $2k_{uv}k_{vw}$. These exhaust the two-step paths, giving
\begin{align*}
(K^2)_{z\oplus e_u\oplus e_w,z}=2\bigl(A_u(z)A_w(z)+T_{uw}\bigr).
\end{align*}
Since the direct amplitude vanishes, Eq.~\eqref{eq:fourth-order-expansion} reduces to $f(4,z,z\oplus e_u\oplus e_w)=6(K^2)_{z\oplus e_u\oplus e_w,z}^2$, proving Eq.~\eqref{eq:nonadjacent_pair}. In particular, no diagonal coefficient enters this identity.

The quantity $Q_{uw}(z)$ is also accessible from a local measurement. Take $B\supseteq\{u,w\}\cup N(u)\cup N(w)$ and measure the probability that only $u,w$ flip within $B$. Any two NN terms producing both flips must each meet $\{u,w\}$ and are therefore supported in $B$. They cannot produce additional flips outside $B$, so the local and full-string fourth-order coefficients agree, as explained in Appendix~\ref{sec:local_transition_bridge}. Varying the input bits changes $A_u(z)A_w(z)$ while leaving $T_{uw}$ fixed. This is how we extract the edge products from the squared quantity in Eq.~\eqref{eq:nonadjacent_pair}.

\paragraph{Chain endpoints.}
We first use this identity to resolve the interchange in an endpoint block. Number the sites $1,\ldots,n$ of an open chain, with $n\ge4$, and write $s_j=(-1)^{z_j}$. The single-flip block at site $3$ is already known because it has degree two, whereas the block at site $1$ is still ambiguous. Write
\begin{align*}
A_1=a+bs_2,\qquad A_3=c+ds_2+es_4,\qquad T=k_{12}k_{23}.
\end{align*}
For each fixed $s_2=\pm1$, set $B_\pm=c+ds_2\pm e$, and let $Q_\pm$ denote $Q_{13}(z)$ for the choices $s_4=\pm1$, respectively. Thus Eq.~\eqref{eq:nonadjacent_pair} gives $Q_\pm=(A_1B_\pm+T)^2$. Second-order data already give $A_1^2$, so subtracting these two equations isolates the product $A_1T$:
\begin{equation*}
W(s_2)=\frac{Q_+-Q_--A_1^2(B_+^2-B_-^2)}{2(B_+-B_-)}=A_1T.
\end{equation*}
The denominator is $4e$, and $|T|=|k_{12}||k_{23}|$ is known. Let $W_0=(W(+1)+W(-1))/2=aT$ and $W_1=(W(+1)-W(-1))/2=bT$. Choose the endpoint gauge $a>0$, set $T=\operatorname{sgn}(W_0)|T|$, and recover $a=W_0/T$, $b=W_1/T$. This removes the interchange left by second-order data. Reflection gives the right endpoint, using the recovered block at $n-2$ and varying $s_{n-3}$. This works also for $n=4$: both blocks used in these tests are interior blocks, and all conditioning polynomials are nonzero generically.

\paragraph{Global $XX$ sign alignment.}
With every single-flip block fixed in the positive-field gauge, we can use Eq.~\eqref{eq:nonadjacent_pair} to align the $XX$ couplings throughout the lattice. For a nonadjacent pair with a common neighbor, choose two local inputs $z,z'$ for which the known products $B=A_u(z)A_w(z)$ and $B'=A_u(z')A_w(z')$ differ. Subtracting the two measured equations cancels $T_{uw}^2$ and gives
\begin{equation}\label{eq:nonadjacent_T_recovery}
T_{uw}=\frac{Q_{uw}(z)-Q_{uw}(z')-(B^2-B'^2)}{2(B-B')}.
\end{equation}
One may fix all input bits to zero and flip any one neighbor of $u$ to obtain $z'$. The denominator is a nonzero polynomial: specialize $A_w=1$ and give $A_u$ a nonzero coefficient on that varied bit. Thus each fixed choice works generically. On a chain, every pair of edges at a common center has exactly one such common neighbor, so \eqref{eq:nonadjacent_T_recovery} gives their product, including the two endpoint bonds.

On a rectangular square lattice, opposite neighbors of a center have one common neighbor. Perpendicular neighbors $u,w$ have exactly two, the center $v$ and the opposite plaquette corner $v'$, even at an open boundary. Consequently
\begin{align*}
T_{uw}=p+p',\qquad p=k_{uv}k_{vw},\quad p'=k_{uv'}k_{v'w}.
\end{align*}
Their magnitudes $m=|p|$, $m'=|p'|$ are known. The four possibilities $\pm m\pm m'$ are distinct when $mm'\ne0$ and $m^2\ne(m')^2$, so $T_{uw}$ uniquely specifies both signs. The last inequality is a nonzero degree-four polynomial in the four distinct plaquette-edge coefficients. These tests therefore recover every incident-edge product, including those at sides and corners.

These products determine the relative signs of all edges meeting at a vertex. Because the lattice is connected, they determine all edge signs once the sign of one reference edge is chosen. Formally, the graph whose vertices are lattice edges, with two vertices adjacent when their edges share an endpoint, is connected; a spanning tree in this graph specifies an order for assigning the signs. We propagate only these discrete sign choices and keep the magnitude of each coupling from its own second-order estimate. The one remaining common $XX$ sign corresponds to time reversal combined with $Z$ conjugation at every site, which preserves the chosen positive single-flip fields.

To make these choices reliably from noisy data, we need lower bounds on the quantities used in the divisions and sign tests. The required local conditioning list consists of the single-flip block products, the endpoint quantities $e,aT,T$, the denominators $B-B'$, and the plaquette polynomials $p,p',p^2-(p')^2$. For example, $|p^2-(p')^2|\ge\rho>0$ implies $|m-m'|\ge\rho/(m+m')$, where $m+m'$ is polynomially bounded by the coefficient radius. The four plaquette candidates are then separated by at least $2\min\{m,m',|m-m'|\}$. Sufficiently accurate data therefore select the correct discrete signs. The remaining square roots, products, and divisions involve only a fixed number of stages, so their errors are controlled by the same local bounds. A union bound over the finitely many sign choices in each constant-size region and over the $O(n)$ regions gives the simultaneous conditioning needed for the spatially varying models.

\paragraph{Amplitudes involving several neighboring signs.}
The NN single-flip amplitude is affine in the neighboring $Z$ signs, which allowed us to recover its coefficients from pairwise products. For the three-local chain, a transition amplitude can also contain products of two neighboring signs. We therefore need a more general way to recover an amplitude from the second-order data giving its square. The following lemma uses the known set of allowed products of signs to rule out spurious choices of the amplitude's sign at different inputs. In Appendix~\ref{sec:proof_TI_3_local}, it recovers the eight coefficients in the single-site-flip amplitude and the three coefficients in the adjacent-pair-flip amplitude, each up to a common sign.
\begin{lemma}[Boolean amplitude reconstruction]\label{lem:boolean_amplitude}
Let $q\ge1$ be a fixed integer, write $[q]=\{1,\ldots,q\}$, and let $s=(s_1,\ldots,s_q)\in\{\pm1\}^q$ denote the $Z$ signs of the input bits on which a real transition amplitude depends. For $\alpha\subseteq[q]$, define the Boolean character
\begin{align*}
\chi_\alpha(s)=\prod_{j\in\alpha}s_j,\qquad \chi_\varnothing(s)=1.
\end{align*}
Let $\mathcal T\subseteq2^{[q]}$ be a known, nonempty collection of allowed subsets, and consider amplitudes of the form
\begin{align*}
A(s)=\sum_{\alpha\in\mathcal T}u_\alpha\chi_\alpha(s),\qquad u_\alpha\in\mathbb R.
\end{align*}
For subsets $\alpha,\beta\subseteq[q]$, write $\alpha\triangle\beta$ for their symmetric difference, namely the elements belonging to exactly one of the two sets. Suppose that
\begin{align*}
\{\alpha\triangle\beta:\alpha\in\mathcal T\}=\mathcal T
\quad\Longrightarrow\quad\beta=\varnothing.
\end{align*}
Then, outside a finite union of proper linear subspaces of the coefficient space $\mathbb R^{\mathcal T}$, the values $A(s)^2$ for all $s\in\{\pm1\}^q$ determine $(u_\alpha)_{\alpha\in\mathcal T}$ up to one common sign. More precisely, any real amplitude $\widetilde A(s)=\sum_{\alpha\in\mathcal T}\widetilde u_\alpha\chi_\alpha(s)$ satisfying $\widetilde A(s)^2=A(s)^2$ at every input has either $\widetilde u_\alpha=u_\alpha$ for all $\alpha$, or $\widetilde u_\alpha=-u_\alpha$ for all $\alpha$.
\end{lemma}
\begin{proof}
The data give $|A(s)|$ at every input, so the ambiguity is a choice of sign at each of the $2^q$ inputs. Each equation $A(s)=0$ defines a proper hyperplane in the coefficient space. Away from these hyperplanes, a competing amplitude with the same square has the form
\begin{align*}
\widetilde A(s)=\omega(s)A(s),\qquad \omega(s)\in\{\pm1\}.
\end{align*}
We show that every nonconstant choice of $\omega$ is inconsistent with the allowed characters for generic coefficients.

Let $\mathcal V_{\mathcal T}$ be the real linear span of the characters $\chi_\alpha$ with $\alpha\in\mathcal T$. For a fixed function $\omega$, the condition $\omega A\in\mathcal V_{\mathcal T}$ is a set of homogeneous linear constraints on $(u_\alpha)$: all Fourier coefficients outside $\mathcal T$ must vanish. Its solutions therefore form a linear subspace. This subspace can equal the whole coefficient space only if multiplication by $\omega$ maps every element of $\mathcal V_{\mathcal T}$ back into $\mathcal V_{\mathcal T}$.

To examine that possibility, expand
\begin{align*}
\omega(s)=\sum_{\beta\subseteq[q]}\widehat\omega(\beta)\chi_\beta(s),\qquad
\widehat\omega(\beta)=2^{-q}\sum_{s\in\{\pm1\}^q}\omega(s)\chi_\beta(s).
\end{align*}
Since $\chi_\alpha\chi_\beta=\chi_{\alpha\triangle\beta}$, multiplying a single allowed character by $\omega$ gives
\begin{align*}
\omega\chi_\alpha=\sum_{\beta\subseteq[q]}\widehat\omega(\beta)\chi_{\alpha\triangle\beta}.
\end{align*}
For fixed $\alpha$, distinct $\beta$ give distinct characters, so no terms in this sum cancel. Thus, if multiplication by $\omega$ preserves $\mathcal V_{\mathcal T}$, every $\beta$ with $\widehat\omega(\beta)\ne0$ satisfies $\alpha\triangle\beta\in\mathcal T$ for every $\alpha\in\mathcal T$. Symmetric difference with $\beta$ is a bijection, so this implies $\{\alpha\triangle\beta:\alpha\in\mathcal T\}=\mathcal T$. By hypothesis, only $\beta=\varnothing$ is possible. Hence $\omega$ is constant and is either $+1$ or $-1$.

Every nonconstant sign function therefore gives a proper linear subspace of exceptional coefficients. There are only finitely many such functions, since $q$ is fixed. Excluding their subspaces and the amplitude-zero hyperplanes leaves exactly the two amplitudes $A$ and $-A$, proving the claim.

The argument also gives a reconstruction procedure: assign signs to the measured magnitudes $|A(s)|$, compute the Boolean Fourier coefficients of each assignment, and retain those with all coefficients outside $\mathcal T$ equal to zero. This requires a number of operations depending only on $q$. For stability under noisy data, the relevant nonzero linear forms are the amplitudes $A(s)$ and, for each excluded sign function, one Fourier coefficient of $\omega A$ outside $\mathcal T$ that is not identically zero. Lower bounds on the former control the square roots; lower bounds on the latter separate the rejected sign assignments from the allowed ones. These are precisely the finite polynomial conditioning bounds supplied by Lemma~\ref{lem:gaussian_polynomial} in the reconstruction proofs.
\end{proof}
The support condition also explains why the endpoint amplitudes required fourth-order data. For $A(s)=a+bs$, multiplication by $s$ gives $sA(s)=b+as$ without changing the square, leaving $a$ and $b$ interchangeable at second order. More generally, multiplication by a character $\chi_\beta$ preserves $A^2$ and replaces each allowed character $\chi_\alpha$ by $\chi_{\alpha\triangle\beta}$. If a nonempty $\beta$ preserves $\mathcal T$, the resulting amplitude is still admissible, and the squared data cannot distinguish it from the original one. The two character supports used in the three-local reconstruction have no such symmetry, as verified in Appendix~\ref{sec:proof_TI_3_local}. Lemma~\ref{lem:boolean_amplitude} therefore recovers each coefficient block from second-order data up to a common sign for generic coefficients.

\subsection{Learning TI Hamiltonians}

We now prove the reconstruction guarantees for the TI IXZ models in Appendix~\ref{sec:zbasis}.
We begin with isotropic NN Hamiltonians in one, two, and three dimensions, using only the all-zero input, and give explicit certificates for the resulting polynomial systems.
We then allow arbitrary computational-basis inputs and treat general TI NN Hamiltonians in these dimensions, followed by TI three-local chains.
For each family, we derive the second- and fourth-order reconstruction equations, identify the remaining sign gauges, and verify the polynomial conditions that ensure generic recovery and stability under Gaussian smoothing.

\subsubsection{Isotropic TI systems and their certificates}\label{sec:equation_system_TIRI_allzero}
We first derive the equations for the isotropic models in Appendix~\ref{sec:RI_TI_IXZ}.
For the all-zero input, flipping a set $\sigma$ changes the measured parity $Z^s$ precisely when $|\sigma\cap s|$ is odd.
Thus the second and fourth derivatives of this parity expectation are
\begin{align*}
\widetilde f(2,0^n,s)=-4\sum_{\sigma:\,|\sigma\cap s|\ {\rm odd}}|K_{1_\sigma,0}|^2,\qquad
\widetilde f(4,0^n,s)=-48\sum_{\sigma:\,|\sigma\cap s|\ {\rm odd}}[t^4]p_{0\to1_\sigma}.
\end{align*}
In spatial dimension $\mathfrak D\in\{1,2,3\}$, the single-flip amplitude at a bulk site is $S_{\mathfrak D}=h_X+2\mathfrak D J_{XZ}$.
Writing $u$ and $v$ for the second derivatives measured on a single site and an adjacent pair, respectively, we obtain
\begin{align*}
u=-4S_{\mathfrak D}^2-8\mathfrak D J_{XX}^2,\qquad
v=-8S_{\mathfrak D}^2-8(2\mathfrak D-1)J_{XX}^2.
\end{align*}
These equations determine $J_{XX}^2=(v-2u)/8$ and $S_{\mathfrak D}^2=-u/4-2\mathfrak D J_{XX}^2$.
To determine the remaining parameters, we use fourth-order data.
For each dimension, let $F_1,\ldots,F_4$ denote the four fourth derivatives in the order listed in Appendix~\ref{sec:RI_TI_IXZ}. Table~\ref{tab:isotropic_fourth_order_coefficients} gives their monomial coefficients, with columns ordered as $F_1,\ldots,F_4$ within each dimensional block.
\begin{table}[t]
\centering
\begingroup\scriptsize\setlength{\tabcolsep}{3pt}\renewcommand{\arraystretch}{1.12}
\begin{tabular}{lrrrrrrrrrrrr}
\toprule
&\multicolumn{4}{c}{1D}&\multicolumn{4}{c}{2D}&\multicolumn{4}{c}{3D}\\
Monomial &1&2&3&4&1&2&3&4&1&2&3&4\\ \midrule
$J_{XX}^{4}$ & 128 & 128 & 336 & 336 & 640 & 1536 & 3200 & 128 & 1536 & 4480 & 336 & 2816 \\
$J_{XZ}^{2} J_{XX}^{2}$ & 512 & 768 & 1136 & 1136 & 4992 & 12800 & 29184 & 512 & 18048 & 54272 & 1968 & 43360 \\
$h_X J_{XZ} J_{XX}^{2}$ & 512 & 1024 & 1728 & 1344 & 2560 & 7424 & 15872 & 512 & 6144 & 19968 & 1344 & 22528 \\
$J_{ZZ}^{2} J_{XX}^{2}$ & 128 & 128 & 144 & 144 & 2304 & 3456 & 4608 & 288 & 9600 & 16000 & 1200 & 6464 \\
$h_Z J_{ZZ} J_{XX}^{2}$ & 256 & 256 & 320 & 320 & 1536 & 2304 & 3072 & 384 & 3840 & 6400 & 960 & 3584 \\
$h_X^{2} J_{XX}^{2}$ & 192 & 384 & 960 & 576 & 384 & 1152 & 2304 & 192 & 576 & 1920 & 288 & 3072 \\
$h_Z^{2} J_{XX}^{2}$ & 128 & 128 & 192 & 192 & 256 & 384 & 512 & 128 & 384 & 640 & 192 & 512 \\
$J_{ZZ} J_{XZ}^{2} J_{XX}$ & -256 & -512 & -256 & -256 & -10240 & -19456 & -20480 & -1088 & -62208 & -121344 & -5760 & -53248 \\
$h_Z J_{XZ}^{2} J_{XX}$ & -128 & -256 & -160 & -160 & -2560 & -5120 & -5120 & -384 & -10368 & -20736 & -1440 & -12352 \\
$h_X J_{ZZ} J_{XZ} J_{XX}$ & -640 & -1024 & -864 & -864 & -6656 & -12544 & -13312 & -1408 & -24192 & -47104 & -4320 & -29248 \\
$h_Z h_X J_{XZ} J_{XX}$ & -320 & -640 & -416 & -416 & -1664 & -3328 & -3328 & -448 & -4032 & -8064 & -1056 & -7168 \\
$h_X^{2} J_{ZZ} J_{XX}$ & -256 & -384 & -384 & -384 & -1024 & -1920 & -2048 & -384 & -2304 & -4480 & -768 & -3968 \\
$h_Z h_X^{2} J_{XX}$ & -128 & -256 & -192 & -192 & -256 & -512 & -512 & -128 & -384 & -768 & -192 & -1024 \\
$J_{XZ}^{4}$ & 768 & 1536 & 1184 & 1184 & 16384 & 38912 & 57344 & 1440 & 89856 & 235008 & 7488 & 159040 \\
$h_X J_{XZ}^{3}$ & 1536 & 3072 & 3008 & 3008 & 14336 & 37888 & 53248 & 2432 & 50688 & 147456 & 7872 & 153344 \\
$J_{ZZ}^{2} J_{XZ}^{2}$ & 256 & 512 & 272 & 272 & 4096 & 8192 & 8192 & 256 & 20736 & 41472 & 1296 & 19488 \\
$h_Z J_{ZZ} J_{XZ}^{2}$ & 256 & 512 & 288 & 288 & 2048 & 4096 & 4096 & 256 & 6912 & 13824 & 864 & 8960 \\
$h_X^{2} J_{XZ}^{2}$ & 1024 & 2432 & 2624 & 2624 & 4352 & 13696 & 17920 & 1280 & 9984 & 34176 & 2784 & 57088 \\
$h_Z^{2} J_{XZ}^{2}$ & 64 & 128 & 80 & 80 & 256 & 512 & 512 & 64 & 576 & 1152 & 144 & 1056 \\
$h_X J_{ZZ}^{2} J_{XZ}$ & 256 & 512 & 288 & 288 & 2048 & 4096 & 4096 & 256 & 6912 & 13824 & 864 & 8960 \\
$h_Z h_X J_{ZZ} J_{XZ}$ & 256 & 512 & 320 & 320 & 1024 & 2048 & 2048 & 256 & 2304 & 4608 & 576 & 4224 \\
$h_X^{3} J_{XZ}$ & 256 & 896 & 960 & 960 & 512 & 2176 & 2560 & 256 & 768 & 3456 & 384 & 9728 \\
$h_Z^{2} h_X J_{XZ}$ & 64 & 128 & 96 & 96 & 128 & 256 & 256 & 64 & 192 & 384 & 96 & 512 \\
$h_X^{2} J_{ZZ}^{2}$ & 64 & 128 & 80 & 80 & 256 & 512 & 512 & 64 & 576 & 1152 & 144 & 1056 \\
$h_Z h_X^{2} J_{ZZ}$ & 64 & 128 & 96 & 96 & 128 & 256 & 256 & 64 & 192 & 384 & 96 & 512 \\
$h_X^{4}$ & 16 & 128 & 128 & 128 & 16 & 128 & 128 & 16 & 16 & 128 & 16 & 640 \\
$h_Z^{2} h_X^{2}$ & 16 & 32 & 32 & 32 & 16 & 32 & 32 & 16 & 16 & 32 & 16 & 64 \\
\bottomrule
\end{tabular}
\endgroup
\caption{Monomial coefficients of the fourth derivatives $F_1,\ldots,F_4$ for the isotropic TI IXZ models; the observables are ordered as in Appendix~\ref{sec:RI_TI_IXZ}.}
\label{tab:isotropic_fourth_order_coefficients}
\end{table}
These coefficients can be computed from four nested commutators of $Z^s$, or equivalently from the transition expansion \eqref{eq:fourth-order-expansion}.

\paragraph{Uniqueness of the four solutions.}
We now verify that these equations have only the four solutions allowed by the sign gauges.
The two second-order equations are linear in $S_{\mathfrak D}^2$ and $J_{XX}^2$, with coefficient matrix
  \begin{align*}
\begin{pmatrix}-4&-8\mathfrak{D}\\-8&-8(2\mathfrak{D}-1)\end{pmatrix},
\qquad \det=-32.
\end{align*} 
For each of the four sign choices for $S_{\mathfrak D}$ and $J_{XX}$, we substitute $h_X=S_{\mathfrak D}-2\mathfrak D J_{XZ}$ into the fourth-order equations.
We apply the polynomial certificate in Appendix~\ref{sec:certificate} to the resulting map $P_\theta$, with
\begin{align*}
\theta=(S_{\mathfrak D},J_{XX}),\qquad t=(h_Z,J_{ZZ},J_{XZ}).
\end{align*}
Here $\theta$ contains the chosen second-order parameters, $t$ is the true remaining parameter tuple, and $x$ denotes a candidate tuple.
The candidate satisfies the measured equations exactly when $E_j(x):=P_{\theta,j}(x)-P_{\theta,j}(t)=0$ for every $j$.
To form the certificate matrix, we use all monomials of degree at most $L$ as columns.
Each row lists the coefficients of $x^\alpha E_j(x)$ for $|\alpha|\leq L-\deg_xE_j$. 
We also include the row for the constant polynomial $1$.
The following table records exact computations at
\begin{align*}
(h_Z,h_X,J_{ZZ},J_{XZ},J_{XX})=(1,2,3,4,5),
\qquad t=(1,3,4),\quad \theta=(2+8\mathfrak{D},5).
\end{align*} 
\begin{center}
\begin{tabular}{c|c|c|c|c}
$\mathfrak{D}$  & $\deg_x(E_1,E_2,E_3,E_4)$ & $L$ & matrix size & selected determinant modulo $1000003$\\
\hline
1 & $(2,2,4,4)$ & 8 & $239\times165$ & $79624$\\
2 & $(2,2,2,4)$ & 6 & $116\times84$ & $124413$\\
3 & $(2,2,4,4)$ & 8 & $239\times165$ & $571887$
\end{tabular}
\end{center}
For reproducibility, we specify the square submatrices whose determinants appear in the table.
We order monomials first by total degree and then lexicographically by $(\alpha_1,\alpha_2,\alpha_3)$.
We order the rows for $x^\alpha E_j$ first by $j$ and then by this monomial order, placing the constant row last.
At the displayed parameter point, we perform elimination modulo $1000003$ from left to right, taking the first available nonzero pivot.
We select the original pivot rows, retain all columns, and restore the original row order.
The selected rows include the constant row, which remains last, as required by Lemma~\ref{lem:exactcertificate}.

The nonzero values in the table show that the selected determinants are nonzero polynomials in $(\theta,t)$.
Lemma~\ref{lem:exactcertificate} therefore gives a unique solution for the three remaining parameters generically, even among complex candidates.
Time reversal and $Z^{\otimes n}$ conjugation relate all four sign choices for $S_{\mathfrak D}$ and $J_{XX}$, giving exactly the four solutions in Appendix~\ref{sec:RI_TI_IXZ}.

\paragraph{Stability under Gaussian smoothing.}
The change of variables from the original coefficients to $(S_{\mathfrak D},J_{XX},t)$ is invertible and linear, and the substituted determinants are therefore nonzero.
We include $S_{\mathfrak D}$ and $J_{XX}$ in the conditioning event so that taking square roots of the second-order estimates is stable.
Lemmas~\ref{lem:gaussian_polynomial} and~\ref{lem:finite_stage_stability} then bound the errors in $\hat S_{\mathfrak D}$ and $\hat J_{XX}$ by a polynomial conditioning factor times the derivative error.
Because $S_{\mathfrak D}=h_X+2\mathfrak D J_{XZ}$ is correlated with $J_{XZ}$, we apply Corollary~\ref{cor:certificate_partial_reconstruction} in the original five independently smoothed coordinates.
This controls the errors in the remaining parameters, including the effects of estimating $S_{\mathfrak D}$, $J_{XX}$, and the fourth derivatives.
Appendix~\ref{sec:certificate_refinement} gives the classical cost for each sign choice.

\subsubsection{Learning TI NN IXZ Hamiltonians on lattices with \texorpdfstring{$Z$}{Z}-basis access}\label{sec:TI_Zbasis_proofs}
We now allow arbitrary computational-basis inputs and drop the isotropy assumption.
We prove Theorem~\ref{thm:TI_NN_z_to_zp} by first reconstructing the off-diagonal coefficients and then using fourth-order data to determine the diagonal coefficients and auxiliary relative signs except the unavoidable gauges.

\begin{proof}[Proof in one dimension]
As earlier, we write $E_z=\langle z|D|z\rangle$ and $f(r,z,z')=\tfrac{d^r}{dt^r}p_{z\to z'}(t)\bigg|_{t=0}$.
Fix an interior site $i$ and consider four inputs $z^{ab}$, with neighboring bits $z_{i-1}=a$ and $z_{i+1}=b$ and all other bits fixed.
The single-flip amplitude and second derivative are
\begin{align*}
\langle z^{ab}\oplus e_i|K|z^{ab}\rangle
&=h_X+(-1)^aJ_{ZX}+(-1)^bJ_{XZ},\\
f(2,z^{ab},z^{ab}\oplus e_i)
&=2\left(h_X+(-1)^aJ_{ZX}+(-1)^bJ_{XZ}\right)^2.
\end{align*}
Taking the four Fourier coefficients over $a,b\in\{0,1\}$ gives
\begin{align*}
&h_X^2+J_{ZX}^2+J_{XZ}^2=\tfrac18\sum_{a,b}f(2,z^{ab},z^{ab}\oplus e_i),\qquad
h_XJ_{ZX}=\tfrac1{16}\sum_{a,b}(-1)^a f(2,z^{ab},z^{ab}\oplus e_i),\\
&h_XJ_{XZ}=\tfrac1{16}\sum_{a,b}(-1)^b f(2,z^{ab},z^{ab}\oplus e_i),\qquad
J_{ZX}J_{XZ}=\tfrac1{16}\sum_{a,b}(-1)^{a+b}f(2,z^{ab},z^{ab}\oplus e_i).
\end{align*}
When $h_XJ_{ZX}J_{XZ}\ne0$, these three pairwise products determine $(h_X^2,J_{ZX}^2,J_{XZ}^2)$ and all relative signs.
Hence $(h_X,J_{ZX},J_{XZ})$ is determined up to the simultaneous sign flip
\begin{align*}
(h_X,J_{ZX},J_{XZ})\mapsto (-h_X,-J_{ZX},-J_{XZ}).
\end{align*}
Flipping an adjacent pair determines the remaining off-diagonal magnitude, since
\begin{align*}
\langle z\oplus e_i\oplus e_{i+1}|K|z\rangle=J_{XX},\quad\Rightarrow\quad f(2,z,z\oplus e_i\oplus e_{i+1})=2J_{XX}^2.
\end{align*} 
We choose one common sign for $(h_X,J_{ZX},J_{XZ})$ and one sign for $J_{XX}$.
The remaining unknowns are $h_Z$ and $J_{ZZ}$.

\paragraph{Recovering the diagonal coefficients.}
In Eq.~\eqref{eq:fourth-order-expansion}, the diagonal energy $E_z=\langle z|D|z\rangle=h_Z\sum_{j=1}^n(-1)^{z_j}
+J_{ZZ}\sum_{j=1}^{n-1}(-1)^{z_j+z_{j+1}}$ is linear in these two unknowns, while $K$, $K^2$, and $K^3$ are known.
Each fourth-order equation is therefore at most quadratic in $h_Z$ and $J_{ZZ}$.
We use four single-flip equations at an interior site $i$, setting $z_{i-2}=z_i=z_{i+2}=0$ and taking $z'=z\oplus e_i$.
In each equation, we subtract the value at $h_Z=J_{ZZ}=0$, which depends only on the recovered off-diagonal coefficients.
For $z_{i-1}=0,z_{i+1}=1$, this gives
\begin{equation*}
\resizebox{1.0\textwidth}{!}{$\begin{aligned}
f(4,z,z\oplus e_i)-f(4,z,z\oplus e_i)\big|_{h_Z=J_{ZZ}=0}=8h_Z(J_{XZ}-J_{ZX}-h_X)\left(6J_{XX}J_{XZ}-2J_{XX}J_{ZX}-J_{XZ}h_Z+J_{ZX}h_Z+h_Xh_Z\right).
\end{aligned}$}
\end{equation*} 
For $z_{i-1}=1,z_{i+1}=0$, the same subtraction gives
\begin{equation*}
\resizebox{1.0\textwidth}{!}{$\begin{aligned}
f(4,z,z\oplus e_i)-f(4,z,z\oplus e_i)\big|_{h_Z=J_{ZZ}=0}=8h_Z(J_{XZ}-J_{ZX}+h_X)\left(2J_{XX}J_{XZ}-6J_{XX}J_{ZX}-J_{XZ}h_Z+J_{ZX}h_Z-h_Xh_Z\right).
\end{aligned}$}
\end{equation*} 
The two equations involve only $h_Z$ and $h_Z^2$.
Dividing them by $8(J_{XZ}-J_{ZX}-h_X)$ and $8(J_{XZ}-J_{ZX}+h_X)$, respectively, gives the right-hand sides
\begin{align*}
2J_{XX}(3J_{XZ}-J_{ZX})h_Z-(J_{XZ}-J_{ZX}-h_X)h_Z^2,\qquad
2J_{XX}(J_{XZ}-3J_{ZX})h_Z-(J_{XZ}-J_{ZX}+h_X)h_Z^2.
\end{align*}
The divisors and the determinant of this linear system are nonzero provided that
\begin{align*}
J_{XZ}-J_{ZX}-h_X\neq 0,\quad
J_{XZ}-J_{ZX}+h_X\neq 0,\quad
J_{XX}(J_{XZ}-J_{ZX})(J_{XZ}+J_{ZX}+2h_X)\neq 0,
\end{align*}
in which case the two equations uniquely determine $h_Z$.
We next use the two remaining neighbor patterns to recover $J_{ZZ}$.
For $z_{i-1}=0,z_{i+1}=0$, the same subtraction gives
\begin{equation*}
\resizebox{1.0\textwidth}{!}{$\begin{aligned}
&f(4,z,z\oplus e_i)-f(4,z,z\oplus e_i)\big|_{h_Z=J_{ZZ}=0}=8(h_X+J_{ZX}+J_{XZ})\Big[-4(h_X+J_{ZX}+J_{XZ})J_{ZZ}^2+4\big(J_{XX}(J_{XZ}+J_{ZX}+2h_X)\\
&\qquad\qquad-(h_X+J_{ZX}+J_{XZ})h_Z\big)J_{ZZ}+h_Z\left(2J_{XX}J_{XZ}+2J_{XX}J_{ZX}+8J_{XX}h_X-J_{XZ}h_Z-J_{ZX}h_Z-h_Xh_Z\right)\Big].
\end{aligned}$}
\end{equation*} 
For $z_{i-1}=1,z_{i+1}=1$, the same subtraction gives
\begin{equation*}
\resizebox{1.0\textwidth}{!}{$\begin{aligned}
&f(4,z,z\oplus e_i)-f(4,z,z\oplus e_i)\big|_{h_Z=J_{ZZ}=0}=8(J_{XZ}+J_{ZX}-h_X)\Big[-4(J_{XZ}+J_{ZX}-h_X)J_{ZZ}^2+4\big(J_{XX}(J_{XZ}+J_{ZX}+2h_X)\\
&\qquad\qquad+(J_{XZ}+J_{ZX}-h_X)h_Z\big)J_{ZZ}+h_Z\left(6J_{XX}J_{XZ}+6J_{XX}J_{ZX}+8J_{XX}h_X-J_{XZ}h_Z-J_{ZX}h_Z+h_Xh_Z\right)\Big].
\end{aligned}$}
\end{equation*} 
We divide these equations by $8(h_X+J_{ZX}+J_{XZ})$ and $8(J_{XZ}+J_{ZX}-h_X)$, respectively, and subtract the terms depending only on the known $h_Z$.
The resulting linear system in $J_{ZZ}$ and $J_{ZZ}^2$ has determinant
\begin{align*}
32\left[J_{XX}h_X(J_{XZ}+J_{ZX}+2h_X)+h_Z\left((J_{XZ}+J_{ZX})^2-h_X^2\right)\right].
\end{align*}
Thus the divisors and determinant are nonzero provided that
\begin{align*}
h_X+J_{ZX}+J_{XZ}\neq 0,\quad
J_{XZ}+J_{ZX}-h_X\neq 0,\quad
J_{XX}h_X(J_{XZ}+J_{ZX}+2h_X)+h_Z\left((J_{XZ}+J_{ZX})^2-h_X^2\right)\neq 0,
\end{align*}
in which case $J_{ZZ}$ is uniquely determined.
Combining both reconstruction steps, we obtain uniqueness up to the unavoidable gauges in \Cref{thm:unobservable_gauges} whenever
\begin{align*}
&h_XJ_{ZX}J_{XZ}\ne0,\quad
J_{XZ}-J_{ZX}-h_X\neq 0,\quad
J_{XZ}-J_{ZX}+h_X\neq 0,\quad
J_{XX}(J_{XZ}-J_{ZX})(J_{XZ}+J_{ZX}+2h_X)\neq 0,\\
&h_X+J_{ZX}+J_{XZ}\neq 0,\quad
J_{XZ}+J_{ZX}-h_X\neq 0,\quad
J_{XX}h_X(J_{XZ}+J_{ZX}+2h_X)+h_Z\left((J_{XZ}+J_{ZX})^2-h_X^2\right)\neq 0.
\end{align*}
The two off-diagonal sign choices give exactly the following four parameter configurations:
\begin{align*}
&(h_X,h_Z,J_{XX},J_{XZ},J_{ZX},J_{ZZ}),\qquad
(-h_X,h_Z,J_{XX},-J_{XZ},-J_{ZX},J_{ZZ}),\\
&(-h_X,-h_Z,-J_{XX},-J_{XZ},-J_{ZX},-J_{ZZ}),\qquad
(h_X,-h_Z,-J_{XX},J_{XZ},J_{ZX},-J_{ZZ}).
\end{align*}

\paragraph{Stability under Gaussian smoothing.}
The products used to recover the off-diagonal coefficients, the two linear-system determinants, and the displayed divisors are nonzero polynomials in the six independent parameters.
Lemmas~\ref{lem:gaussian_polynomial} and~\ref{lem:finite_stage_stability} bound all of them away from zero by inverse-polynomial quantities and control the propagation of measurement error, including errors in the matrices and right-hand sides.
We choose derivative precision $\epsilon/\poly(R,\varsigma^{-1},\delta^{-1})$, below the stability threshold, so that the final parameter error is at most $\epsilon$.
Appendix~\ref{sec:defer_derivative_estimation} then gives the stated $\widetilde O(\epsilon^{-2})$ dependence on the target precision.
\end{proof}

\begin{proof}[Proof in two and three dimensions]
As in one dimension, the Fourier coefficients of the second-order single-flip data give the pairwise products of the single-$X$ coefficients and hence determine them up to a common sign.
Adjacent-pair flips give the magnitudes of the $XX$ coefficients, and Eq.~\eqref{eq:fourthordertransition} determines their relative signs across the lattice axes.
We choose representatives for these two common signs and recover the diagonal coefficients from fourth-order data.

\paragraph{The 2D case.}
Fix an interior site $r$ with input bit $z_r=0$.
Write $s_{\pm\alpha}=(-1)^{z_{r\pm\hat e_\alpha}}$ for its four neighbor signs, with $\alpha=x,y$.
The single-flip amplitude is $A=h_X+J^x_{ZX}s_{-x}+J^x_{XZ}s_{+x}+J^y_{ZX}s_{-y}+J^y_{XZ}s_{+y}$, and half the diagonal energy difference is $E=h_Z+J^x_{ZZ}(s_{-x}+s_{+x})+J^y_{ZZ}(s_{-y}+s_{+y})$.
After subtracting the terms depending only on $K$ as in Eq.~\eqref{eq:local_transition_bridge}, the fourth derivative is $-8A^2E^2$ plus a term linear in the diagonal coefficients.
This latter term has degree at most two in the neighbor signs: each contributing two-step transition flips one neighbor and the edge joining it to $r$, in either order, and distinct neighbors of $r$ are not adjacent.

We isolate the quadratic part by fixing all other input bits in the radius-$2$ region and taking five weighted averages over the four neighbor bits.
For each average, multiply the fourth derivative by one of the following sign products and average uniformly over the sixteen inputs:
\begin{align*}
(-1)^{z_{r+\hat e_x}+z_{r-\hat e_y}+z_{r+\hat e_y}},\quad 
(-1)^{z_{r-\hat e_x}+z_{r-\hat e_y}+z_{r+\hat e_y}},\quad 
(-1)^{z_{r-\hat e_x}+z_{r+\hat e_x}+z_{r+\hat e_y}},\\ 
(-1)^{z_{r-\hat e_x}+z_{r+\hat e_x}+z_{r-\hat e_y}},\quad 
(-1)^{z_{r-\hat e_x}+z_{r+\hat e_x}+z_{r-\hat e_y}+z_{r+\hat e_y}}.
\end{align*} 
Each sign product involves three or four neighbors, so these averages eliminate every term of degree at most two in the neighbor signs.
After the subtraction above and division by $-32$, the five averages give
\begin{align*}
&E_1=J^y_{ZX}J^y_{XZ} h_ZJ^x_{ZZ}+J^x_{XZ}(J^y_{ZX}+J^y_{XZ})h_ZJ^y_{ZZ}+h_X(J^y_{ZX}+J^y_{XZ})J^x_{ZZ}J^y_{ZZ}+h_XJ^x_{XZ}(J^y_{ZZ})^2,\\
&E_2=J^y_{ZX}J^y_{XZ} h_ZJ^x_{ZZ}+J^x_{ZX}(J^y_{ZX}+J^y_{XZ})h_ZJ^y_{ZZ}+h_X(J^y_{ZX}+J^y_{XZ})J^x_{ZZ}J^y_{ZZ}+h_XJ^x_{ZX}(J^y_{ZZ})^2,\\
&E_3=J^y_{XZ}(J^x_{ZX}+J^x_{XZ})h_ZJ^x_{ZZ}+J^x_{ZX}J^x_{XZ} h_ZJ^y_{ZZ}+h_X(J^x_{ZX}+J^x_{XZ})J^x_{ZZ}J^y_{ZZ}+h_XJ^y_{XZ}(J^x_{ZZ})^2,\\
&E_4=J^y_{ZX}(J^x_{ZX}+J^x_{XZ})h_ZJ^x_{ZZ}+J^x_{ZX}J^x_{XZ} h_ZJ^y_{ZZ}+h_X(J^x_{ZX}+J^x_{XZ})J^x_{ZZ}J^y_{ZZ}+h_XJ^y_{ZX}(J^x_{ZZ})^2,\\
&E_5=(J^x_{ZX}+J^x_{XZ})(J^y_{ZX}+J^y_{XZ})J^x_{ZZ}J^y_{ZZ}+J^y_{ZX}J^y_{XZ}(J^x_{ZZ})^2+J^x_{ZX}J^x_{XZ}(J^y_{ZZ})^2.
\end{align*}
To state the determinant of this system, write
\begin{align*} 
a=h_X,\quad b=J^x_{ZX},\quad c=J^x_{XZ},\quad d=J^y_{ZX},\quad e=J^y_{XZ}. 
\end{align*}
The five equations are linear in $(h_ZJ^x_{ZZ},h_ZJ^y_{ZZ},J^x_{ZZ}J^y_{ZZ},(J^x_{ZZ})^2,(J^y_{ZZ})^2)$, with coefficient determinant
\begin{align*}
3a^2bcde(b-c)(b+c)(d-e)(d+e).
\end{align*}
Thus these five products are uniquely determined when $a b c d e(b-c)(b+c)(d-e)(d+e)\neq0$.
For generic nonzero $J^x_{ZZ}$ and $J^y_{ZZ}$, their squares and product fix their relative sign, and the products with $h_Z$ then determine the field.
The only remaining ambiguity in the diagonal coefficients is
\begin{align*}
(h_Z,J^x_{ZZ},J^y_{ZZ})\mapsto-(h_Z,J^x_{ZZ},J^y_{ZZ}).
\end{align*}
We resolve this ambiguity relative to the chosen common $XX$ sign using the single-flip fourth derivative for the all-zero input.
For a neighbor $j$ of $r$, let $g_j$ and $c_j$ be the coefficients of $X_rZ_j$ and $Z_rX_j$, respectively, and let $e_j=J^\alpha_{ZZ}$ and $k_j=J^\alpha_{XX}$ when the edge lies on axis $\alpha$.
Then $A=h_X+\sum_jg_j$ and $E=h_Z+2\sum_\alpha J^\alpha_{ZZ}$ for this input.
The part of the fourth derivative linear in the diagonal coefficients is
\begin{align*}
\mathcal L=16A\sum_jk_j\bigl[2(A-c_j)(E-e_j)+2c_je_j-c_jE\bigr].
\end{align*}
This expression follows from the two orders of the neighbor flip and edge flip; excluding their shared edge gives half the neighbor energy difference $E-e_j$ and the amplitude $A-c_j$.
At $(h_X,J^x_{ZX},J^x_{XZ},J^y_{ZX},J^y_{XZ})=(1,2,3,4,5)$, $(J^x_{XX},J^y_{XX})=(1,2)$, and $(h_Z,J^x_{ZZ},J^y_{ZZ})=(1,2,3)$, we obtain $\mathcal L=239280\ne0$.
Thus $\mathcal L$ is a nonzero polynomial in the Hamiltonian coefficients.
Negating only the diagonal coefficients changes the sign of $\mathcal L$ and leaves the quadratic part unchanged, so this measurement distinguishes the two choices generically.

\paragraph{The 3D case.}
The second-order single-flip equations now depend on six neighbor signs.
Their Fourier coefficients determine $(h_X,\{J^\alpha_{ZX},J^\alpha_{XZ}\}_{\alpha=x,y,z})$ up to a common sign by the same pairwise-product reconstruction.
Adjacent-pair flips give $|J^\alpha_{XX}|$ for each axis.
To determine their relative signs, we apply Eq.~\eqref{eq:fourthordertransition} to each pair of distinct axes $\alpha,\beta=x,y,z$:
  \begin{align*}
\frac{1}{4}\sum_{z_{r+2\hat e_\alpha},z_{r+2\hat e_\beta}\in\{0,1\}}(-1)^{z_{r+2\hat e_\alpha}+z_{r+2\hat e_\beta}}f(4,z,z\oplus e_r) =48J^\alpha_{XX}J^\beta_{XX}J^\alpha_{XZ}J^\beta_{XZ}.
\end{align*} 
All input bits outside this average are fixed to zero.
Since the single-$X$ coefficients are already known up to a common sign, these measurements give the pairwise products of the three $XX$ coefficients and hence fix them up to one common sign.

We next recover $(h_Z,J^x_{ZZ},J^y_{ZZ},J^z_{ZZ})$.
Order the six neighbors as $(-x,+x,-y,+y,-z,+z)$ and write their input signs as $s_1,\ldots,s_6$, where each sign is $(-1)$ raised to the corresponding bit.
The single-flip amplitude and half the diagonal energy difference are
\begin{align*}
A&=h_X+J^x_{ZX}s_1+J^x_{XZ}s_2+J^y_{ZX}s_3+J^y_{XZ}s_4+J^z_{ZX}s_5+J^z_{XZ}s_6,\\
E&=h_Z+J^x_{ZZ}(s_1+s_2)+J^y_{ZZ}(s_3+s_4)+J^z_{ZZ}(s_5+s_6).
\end{align*}
We fix $z_r=0$ and all input bits other than the six neighbors.
As in two dimensions, subtracting the terms depending only on $K$ leaves $-8A^2E^2$ plus a term of degree at most two in the neighbor signs.
For the linear part in the diagonal coefficients, the contributing two-step transitions again flip one neighbor and its edge to $r$, in either order.
Their sum is $A$ times an affine function of the neighbor signs: distinct neighbors in the cubic lattice are not adjacent, so the remaining input bits that enter these transitions are fixed.
Weighted averages using products of three or four neighbor signs therefore eliminate the linear part.

We use the nine sign products $123,124,125,234,135,235,1234,134,156$, where, for example, $123$ denotes $s_1s_2s_3$.
For each product, we average uniformly over the six neighbor bits and divide by $-32$.
The resulting nine equations are linear in
\begin{align*}
(h_ZJ^x_{ZZ},h_ZJ^y_{ZZ},h_ZJ^z_{ZZ},J^x_{ZZ}J^y_{ZZ},J^x_{ZZ}J^z_{ZZ},J^y_{ZZ}J^z_{ZZ},(J^x_{ZZ})^2,(J^y_{ZZ})^2,(J^z_{ZZ})^2).
\end{align*}
At $(h_X,J^x_{ZX},J^x_{XZ},J^y_{ZX},J^y_{XZ},J^z_{ZX},J^z_{XZ})=(x_1,x_2,x_3,x_4,x_5,x_6,x_7)$ with $x_j=j$, the coefficient matrix is
\begin{align*}
\begin{pmatrix}
20&6&0&5&0&0&4&0&0\\
25&6&0&5&0&0&5&0&0\\
30&0&6&0&5&0&6&0&0\\
20&27&0&9&0&0&0&3&0\\
24&12&8&6&4&2&0&0&0\\
24&18&12&6&4&3&0&0&0\\
0&0&0&45&0&0&20&6&0\\
20&18&0&9&0&0&0&2&0\\
42&0&26&0&13&0&0&0&2
\end{pmatrix},
\qquad \det=777600.
\end{align*}
The nonzero determinant shows that this system is invertible for generic off-diagonal coefficients.
Its solution gives the squares and pairwise products of the diagonal couplings, as well as their products with $h_Z$.
For generic nonzero diagonal couplings, these determine $(h_Z,J^x_{ZZ},J^y_{ZZ},J^z_{ZZ})$ up to a common sign.

Finally, we determine the diagonal sign relative to the chosen common $XX$ sign from the single-flip fourth derivative for the all-zero input.
To verify that the part linear in the diagonal coefficients is a nonzero polynomial, use the single-$X$ coefficients above and take
\begin{align*}
(J^x_{XX},J^y_{XX},J^z_{XX})=(1,2,3),\qquad
(h_Z,J^x_{ZZ},J^y_{ZZ},J^z_{ZZ})=(1,2,3,4).
\end{align*}
The part of the fourth derivative linear in the diagonal coefficients is then $3527552$. 
It changes sign when only the diagonal tuple is negated and is unchanged when both that tuple and all $XX$ coefficients are negated. 
Thus this measurement fixes the relative sign generically, leaving the two independent sign choices in Theorem~\ref{thm:2D3D_TI_z_to_zp}.

\paragraph{Stability under Gaussian smoothing.}
The two displayed matrix determinants, the products used to recover coefficients from squares, and the polynomials that distinguish the sign choices are all nonzero and have fixed degree.
We apply Lemmas~\ref{lem:gaussian_polynomial} and~\ref{lem:finite_stage_stability} to this finite collection, accounting also for errors in the subtracted terms and matrix entries.
Derivative precision $\epsilon/\poly(R,\varsigma^{-1},\delta^{-1})$, below the stability threshold, then gives parameter error at most $\epsilon$.
The reconstruction uses a constant number of local transition probabilities, linear solves, and sign choices, so Appendix~\ref{sec:defer_derivative_estimation} gives the claimed $\widetilde O(\epsilon^{-2})$ dependence on precision.
\end{proof}

\subsubsection{Learning TI \texorpdfstring{$3$}{3}-local IXZ Hamiltonians on chains with \texorpdfstring{$Z$}{Z}-basis access}\label{sec:proof_TI_3_local}
\begin{proof}[Proof of \Cref{thm:TI_3_local}]
We first recover the off-diagonal coefficients from second-order data, then use fourth-order data to resolve the remaining choices and recover the diagonal coefficients.
For a flip set $\sigma$, write $z'=z\oplus 1_\sigma$.
The second-order coefficient $[t^2]p_{z\to z'}(t)=\left|\langle z|K|z'\rangle\right|^2$ gives the squared transition amplitude.

\paragraph{Recovering the off-diagonal coefficients.}
For a single-site flip $\sigma=\{i\}$, the contributing Pauli strings are $X_i$, $Z_{i-1}X_i$, $X_iZ_{i+1}$, $Z_{i-2}X_i$, $X_iZ_{i+2}$, $Z_{i-2}Z_{i-1}X_i$, $Z_{i-1}X_iZ_{i+1}$, and $X_iZ_{i+1}Z_{i+2}$.
Their sum gives the amplitude
\begin{align*}
\langle z|K|z'\rangle&=h_X+J_{ZIX}(-1)^{z_{i-2}}+J_{ZX}(-1)^{z_{i-1}}+J_{ZZX}(-1)^{z_{i-2}}(-1)^{z_{i-1}}+J_{ZXZ}(-1)^{z_{i-1}}(-1)^{z_{i+1}}\\
&\quad+J_{XZZ}(-1)^{z_{i+1}}(-1)^{z_{i+2}}+J_{XZ}(-1)^{z_{i+1}} +J_{XIZ}(-1)^{z_{i+2}}.
\end{align*}
We vary the four neighboring input signs $(-1)^{z_{i-2}},(-1)^{z_{i-1}},(-1)^{z_{i+1}},(-1)^{z_{i+2}}$.
Lemma~\ref{lem:boolean_amplitude} then determines the eight coefficients $(h_X,J_{ZIX},J_{ZX},J_{ZZX},J_{ZXZ},J_{XZZ},J_{XZ},J_{XIZ})$ generically up to a common sign.

For an adjacent-pair flip $\sigma=\{i,i+1\}$, the contributing strings are $X_iX_{i+1}$, $Z_{i-1}X_iX_{i+1}$, and $X_iX_{i+1}Z_{i+2}$, giving
\begin{align*}
\langle z|K|z'\rangle=J_{XX}+J_{ZXX}(-1)^{z_{i-1}}+J_{XXZ}(-1)^{z_{i+2}},\qquad
[t^2]p_{z\to z'}(t)=\langle z|K|z'\rangle^2.
\end{align*}
Varying the two input signs and applying Lemma~\ref{lem:boolean_amplitude} determines $(J_{XX},J_{ZXX},J_{XXZ})$ generically up to a common sign.
The character index sets are $\{\varnothing,\{-2\},\{-1\},\{1\},\{2\},\{-2,-1\},\{-1,1\},\{1,2\}\}$ for the single-flip amplitude and $\{\varnothing,\{-1\},\{2\}\}$ for the adjacent-pair amplitude.
Neither family is preserved by taking the symmetric difference of every index set with the same nonempty set.
This verifies the hypothesis of Lemma~\ref{lem:boolean_amplitude} in both applications.

For a distance-two pair flip $\sigma=\{i,i+2\}$, the contributing strings are $X_iX_{i+2}$ and $X_iZ_{i+1}X_{i+2}$, giving
\begin{align*}
\langle z|K|z'\rangle=J_{XIX}+J_{XZX}(-1)^{z_{i+1}},\qquad
[t^2]p_{z\to z'}(t)=\langle z|K|z'\rangle^2.
\end{align*}
The two Fourier coefficients of this squared amplitude are $J_{XIX}^2+J_{XZX}^2$ and $2J_{XIX}J_{XZX}$.
Equivalently, they determine $(J_{XIX}+J_{XZX})^2$ and $(J_{XIX}-J_{XZX})^2$, leaving a common sign and an interchange of the two coefficients undetermined.

Finally, a triple flip $\sigma=\{i,i+1,i+2\}$ receives only the direct contribution from $X_iX_{i+1}X_{i+2}$, so
\begin{align*}
[t^2]p_{z\to z'}(t)=J_{XXX}^2,
\end{align*}
which determines $J_{XXX}$ up to sign.

\paragraph{Resolving the remaining off-diagonal choices.}
We fix the common signs of the single-flip and adjacent-pair coefficients, as the two unavoidable gauges allow these choices independently.
There are then eight candidates, corresponding to the four choices for $(J_{XIX},J_{XZX})$ and the two signs of $J_{XXX}$.
To distinguish them, we use the all-zero input and the four-site flip sets $\{i,i+1,i+2,i+3\}$ and $\{i,i+1,i+2,i+4\}$.
No single term in $H$ flips either of these sets, so Eq.~\eqref{eq:fourth-order-expansion} reduces to
\begin{align*}
f(4,0^n,1_\sigma)=6\langle 1_\sigma|K^2|0^n\rangle^2.
\end{align*}
We verify that these two measurements distinguish the candidates at the parameter point
\begin{align*}
&(h_X,J_{XX},J_{XZ},J_{ZX},J_{XIX},J_{XIZ},J_{XXX})=(1,2,3,4,5,6,7),\\
&(J_{XXZ},J_{XZX},J_{XZZ},J_{ZIX},J_{ZXX},J_{ZXZ},J_{ZZX})=(8,9,10,11,12,13,14).
\end{align*}
In each row of the table, we interchange $(J_{XIX},J_{XZX})$ if indicated and then apply the indicated signs to this pair and to $J_{XXX}$, keeping all other coefficients fixed.
The last two columns give the corresponding fourth derivatives:
\begin{center}
\begin{tabular}{c|c|c|r|r}
interchange & pair sign & $J_{XXX}$ sign & $\{i,i+1,i+2,i+3\}$ & $\{i,i+1,i+2,i+4\}$\\
\hline
no & $+$ & $+$ & $6144864$ & $4968600$\\
no & $+$ & $-$ & $6741600$ & $95256$\\
no & $-$ & $+$ & $6144864$ & $95256$\\
no & $-$ & $-$ & $6741600$ & $4968600$\\
yes & $+$ & $+$ & $9166176$ & $4968600$\\
yes & $+$ & $-$ & $4193376$ & $95256$\\
yes & $-$ & $+$ & $9166176$ & $95256$\\
yes & $-$ & $-$ & $4193376$ & $4968600$
\end{tabular}
\end{center}
The eight pairs of values are distinct.
For every pair of candidates, at least one of the two measurement differences is therefore a nonzero polynomial in the Hamiltonian coefficients.
The two measurements distinguish all eight candidates generically, determining the off-diagonal part up to one common sign on the odd-$X$ terms and one on the even-$X$ terms.

\paragraph{Recovering the diagonal coefficients.}
The remaining parameters are $(h_Z,J_{ZZ},J_{ZIZ},J_{ZZZ})$.
For a flip set $\sigma$, let $\Delta_\sigma(z):=E_z-E_{z'}$, where $E_z:=\langle z|D|z\rangle$ and $z'=z\oplus1_\sigma$.
For every allowed flip set $\tau$, let $A_\tau(z)=\langle z\oplus\mathbf 1_\tau|K|z\rangle$.
Equation~\eqref{eq:fourth-order-expansion} gives
\begin{align*}
[t^4]p_{z\to z'}(t)&=\frac{1}{4!}f(4,z,z')=-\frac{1}{12}\langle z|K|z'\rangle^2\Delta_\sigma(z)^2
+\frac{1}{6}\langle z|K|z'\rangle\sum_\tau A_\tau(z)A_{\sigma\triangle \tau}(z\oplus\mathbf 1_\tau)
\left(\Delta_\tau(z)-\Delta_{\sigma\triangle\tau}(z\oplus\mathbf{1}_\tau)\right) \\
&\quad+\frac{1}{4} \langle z|K^2|z'\rangle^2-\frac{1}{3} \langle z|K|z'\rangle(K^3)_{z,z'}.
\end{align*} 
The sum runs over nonempty sets $\tau$ for which both $\tau$ and $\sigma\triangle\tau$ are $X$-supports of allowed $3$-\gl terms.
We subtract the last two terms, which depend only on the recovered off-diagonal coefficients.
The remaining expression is quadratic in the diagonal energy differences.
For the four flip sets $\{i\}$, $\{i,i+1\}$, $\{i,i+2\}$, and $\{i,i+1,i+2\}$, the all-zero-input energy differences determine the diagonal parameters through
\begin{align*}
\begin{cases}
\Delta_1=\Delta_{\{i\}}(0^n)=2h_Z+4J_{ZZ}+4J_{ZIZ}+6J_{ZZZ},\\
\Delta_2=\Delta_{\{i,i+1\}}(0^n)=4h_Z+4J_{ZZ}+8J_{ZIZ}+4J_{ZZZ},\\
\Delta_3=\Delta_{\{i,i+2\}}(0^n)=4h_Z+8J_{ZZ}+4J_{ZIZ}+8J_{ZZZ},\\
\Delta_4=\Delta_{\{i,i+1,i+2\}}(0^n)=6h_Z+4J_{ZZ}+8J_{ZIZ}+6J_{ZZZ}.
\end{cases} \quad\Rightarrow\quad
\begin{pmatrix}
h_Z\\
J_{ZZ}\\
J_{ZIZ}\\
J_{ZZZ}
\end{pmatrix}=\begin{pmatrix}
-\frac{3}{8}&-\frac{1}{4}&\frac{1}{8}&\frac{3}{8}\\
-\frac{1}{4}&\frac{1}{4}&\frac{1}{4}&-\frac{1}{4}\\
\frac{1}{8}&\frac{1}{4}&-\frac{1}{8}&-\frac{1}{8}\\
\frac{3}{8}&-\frac{1}{4}&-\frac{1}{8}&\frac{1}{8}
\end{pmatrix}\begin{pmatrix}
\Delta_1\\
\Delta_2\\
\Delta_3\\
\Delta_4
\end{pmatrix}.
\end{align*}
We therefore only need to recover $\Delta_1,\ldots,\Delta_4$.
For any input, each diagonal energy difference is linear in these four variables, so the fourth-order equations are linear in the following fourteen quantities:
\begin{align*}
(\Delta_1^2,\Delta_1\Delta_2,\Delta_1\Delta_3,\Delta_1\Delta_4,\Delta_2^2,\Delta_2\Delta_3,\Delta_2\Delta_4,\Delta_3^2,\Delta_3\Delta_4,\Delta_4^2,\Delta_1,\Delta_2,\Delta_3,\Delta_4).
\end{align*}
The coefficient matrix depends only on the known off-diagonal coefficients.
To show that it has full rank generically, we select fourteen equations as follows.
Call the four flip sets types $1,2,3,4$ in the order above, and set every input bit to zero except possibly $z_{i+3}$ and $z_{i+4}$.
We use all four types for each input pair $(z_{i+3},z_{i+4})=(0,0),(0,1),(1,0)$, followed by types $3,4$ for $(1,1)$.
We write these equations for $f(4,z,z')$, subtract the known off-diagonal terms, and order the columns by the fourteen quantities displayed above.
At the parameter point used in the sign table, the resulting $14\times14$ matrix has determinant
\begin{align*}
2^{57}3^2 7^{16}11^4 31^4\cdot68258543\neq0.
\end{align*}
Thus the system determines all fourteen quantities generically, including $\Delta_1,\ldots,\Delta_4$ themselves.
The displayed inverse matrix then gives $(h_Z,J_{ZZ},J_{ZIZ},J_{ZZZ})$.
The determinant is homogeneous of degree $32$ in the off-diagonal coefficients: the ten columns for quadratic terms each have degree two, and the four columns for linear terms each have degree three.
Changing the common odd-$X$ sign leaves the diagonal equations unchanged.
Changing the common even-$X$ sign negates their linear part, so it also negates the recovered diagonal coefficients.
These are the two signs on $\Xi_o$ and $\Xi_e$ in the theorem.

\paragraph{Stability under Gaussian smoothing.}
We apply Gaussian small-ball bounds to the finite collection of nonzero polynomials used in the reconstruction.
For the single-site step, there are finitely many candidate sign patterns.
Every pattern incompatible with the single-flip equations gives a nonzero Fourier residual that is linear in the true coefficients, and its absolute value bounds the distance from compatibility up to a fixed factor.
We also include the single-site amplitudes, the products used in the adjacent-pair reconstruction, and $J_{XIX}\pm J_{XZX}$.
For the fourth-order steps, we include the nonzero measurement-difference polynomials verified by the sign table and the determinant of the diagonal reconstruction matrix.
All these polynomials have fixed degree.
Their simultaneous inverse-polynomial lower bounds in $\varsigma$ and $\delta$, together with an upper bound on the coefficient radius, control the errors from taking square roots, choosing candidates, and solving the linear system.
We choose derivative precision $\epsilon$ divided by the resulting polynomial conditioning factor, so that the recovered parameters have error at most $\epsilon$.
There are only constantly many local input choices and candidates, so the reconstruction preserves the $\widetilde O(\epsilon^{-2})$ dependence on precision from derivative estimation.
\end{proof}

\subsection{Spatially varying chains}
\label{sec:proof_1D_local_Z_basis}

\begin{proof}[Proof of Theorem~\ref{thm:local_NN_z_to_zp} for chains]

Here and in the next subsection, $\eta_{\rm data}$ denotes a uniform additive error bound on all estimated second- and fourth-order local derivatives (before we take a Fourier transform or subtract the reconstructed $K$-only terms). On the conditioning event, the off-diagonal reconstruction in Appendix~\ref{sec:local_z_common} has coefficient error at most $\poly(n,R,\varsigma^{-1},\delta^{-1})\eta_{\rm data}$.

Now consider an open chain of length at least $4$, or a periodic chain as in Theorem~\ref{thm:1D_local_z_to_zp}. We first recover and globally align the off-diagonal coefficients as in Appendix~\ref{sec:local_z_common}; for an open chain, this includes recovering the off-diagonal single-flip coefficients at both endpoints using the endpoint tests. Throughout the open-chain analysis, coefficients on absent outside edges are set to zero.

\paragraph{Diagonal reconstruction.}
We now reconstruct the diagonal part of the Hamiltonian in this fixed gauge representative. The reconstruction is local: we consider overlapping four-site windows and use fourth-order transition data to determine the $Z$ fields and $ZZ$ couplings associated with each window. For an interior window with the input bits outside the window fixed to zero, the field at either end appears only through its sum with the adjacent outside $ZZ$ coupling. These coefficients therefore enter through two endpoint combinations, which neighboring windows subsequently disentangle.

Consider a four-site window $\{i,i+1,i+2,i+3\}$. For the already known off-diagonal coefficients that enter the fourth-order equations on this window, we write
\begin{align*}
a_r &= h_X^{i+r}\,, \qquad &&r = 0,1,2,3\,,\\
b_r &= J_{ZX}^{i+r-1}\,, \qquad &&r = 0,1,2,3\,,\\
c_r &= J_{XZ}^{i+r}\,, \qquad &&r = 0,1,2,3\,,\\
k_r &= J_{XX}^{i+r}\,, \qquad &&r = 0,1,2\,.
\end{align*}
Here $a_r$ is the transverse field at site $i+r$, while $b_r$ and $c_r$ are the two mixed couplings that can contribute to a single flip at that site. The coefficients $k_0,k_1,k_2$ are the $XX$ couplings on the three edges inside the window. In particular,
\begin{align*}
b_0 &= J_{ZX}^{i-1}\,, \qquad c_3 = J_{XZ}^{i+3}\,,
\end{align*}
belong to the two edges immediately outside the window.

For an interior window, the fourth-order equations depend on the nine diagonal coefficients
\begin{align*}
h_Z^i\,,\ h_Z^{i+1}\,,\ h_Z^{i+2}\,,\ h_Z^{i+3}\,,\ 
J_{ZZ}^{i-1}\,,\ J_{ZZ}^{i}\,,\ J_{ZZ}^{i+1}\,,\ J_{ZZ}^{i+2}\,,\ J_{ZZ}^{i+3}\,.
\end{align*}
However, the two coefficients at each end of the window appear only through the combinations
\begin{align*}
d_0 &= h_Z^i + J_{ZZ}^{i-1}\,, \qquad
d_3 = h_Z^{i+3} + J_{ZZ}^{i+3}\,.
\end{align*}
So we see that the data depend on only seven effective diagonal variables. We collect the five variables intrinsic to the window as
\begin{align*}
\zeta &= (x,y,\ell_0,\ell_1,\ell_2) = \bigl(h_Z^{i+1},h_Z^{i+2},J_{ZZ}^{i}, J_{ZZ}^{i+1},J_{ZZ}^{i+2}\bigr)\,.
\end{align*}
The remaining two effective variables are
the endpoint combinations $d_0$ and $d_3$.

The equations below are arranged so that $d_0$ and $d_3$ can be eliminated first, leaving a polynomial system for the five components of $\zeta$. Once $\zeta$ has been reconstructed, two of the original equations recover $d_0$ and $d_3$ by division. A single interior window does not by itself separate $h_Z^i$ from $J_{ZZ}^{i-1}$, or $h_Z^{i+3}$ from $J_{ZZ}^{i+3}$, but these endpoint fields appear as middle-site fields in neighboring windows. Overlapping windows then recover the individual coefficients. At an open boundary, the corresponding outside coupling is absent and is set to zero, and so the endpoint combination directly gives the boundary field.

For each middle site $r\in\{i+1,i+2\}$, we set all input bits to zero except $z_{r-1}$ and $z_{r+1}$, and write $L=(-1)^{z_{r-1}}$, $R=(-1)^{z_{r+1}}$, and $z'=z\oplus\mathbf{1}_{\{r\}}$.  We use $\emptyset,L,R,LR$ to label the four Fourier characters $1,L,R,LR$, respectively.  The four subtracted Fourier coefficients are
\begin{align*}
\begin{pmatrix}Y_{r,\emptyset}\\Y_{r,L}\\Y_{r,R}\\Y_{r,LR}\end{pmatrix}
=\frac{1}{4}\sum_{L,R\in\{\pm1\}}
\begin{pmatrix}1\\L\\R\\LR\end{pmatrix}
\left[\frac{f(4,z,z')}{24}
-\frac{1}{4}\bigl((K^2)_{z',z}\bigr)^2
+\frac{1}{3}K_{z',z}(K^3)_{z',z}\right].
\end{align*}
Here $f(4,z,z')$ denotes the fourth derivative at $t=0$, so $f(4,z,z')/24$ is the coefficient of $t^4$ in the Taylor expansion. By Eq.~\eqref{eq:local_transition_bridge}, these Fourier coefficients are obtained from $[t^4](p_H^B-p_K^B)$, with the $K$ contribution computed from the globally aligned off-diagonal estimates. The derivative data and recovered off-diagonal coefficients therefore determine $Y_{r,M}$. For $M\in\{\emptyset,L,R,LR\}$, the single-flip equations at $i+1$ are
\begin{equation*}
    3Y_{i+1,M} = \alpha_M d_0 + Q_{1,M}(\zeta)
\end{equation*}
where we have
\begin{equation*}
    (\alpha_\emptyset, \alpha_L, \alpha_R, \alpha_{LR}) = 4k_0(a_0 + b_0)\cdot (b_1, a_1, 0, c_1)
\end{equation*}
and
  \begin{align*}
Q_{1,\emptyset}(\zeta)&\coloneqq 2b_1(c_0k_0-b_2k_1)\,\ell_0-2a_1(c_0k_0+b_2k_1)\,x-2c_1(c_0k_0-b_2k_1)\,\ell_1 \\
&\qquad+4c_1k_1(a_2+c_2)(y+\ell_2)-(a_1^2+b_1^2+c_1^2)(x^2+\ell_0^2+\ell_1^2)-4a_1b_1\,x\ell_0-4a_1c_1\,x\ell_1-4b_1c_1\,\ell_0 \ell_1,\\
Q_{1,L}(\zeta)&\coloneqq 2a_1(c_0k_0-b_2k_1)\,\ell_0-2b_1(c_0k_0+b_2k_1)\,x\\
&\qquad-2a_1b_1(x^2+\ell_0^2+\ell_1^2)-2(a_1^2+b_1^2+c_1^2)\,x\ell_0-4b_1c_1\,x\ell_1-4a_1c_1\,\ell_0 \ell_1,\\
Q_{1,R}(\zeta)&\coloneqq-2c_1(c_0k_0+b_2k_1)\,x-2a_1(c_0k_0-b_2k_1)\,\ell_1\\
&\qquad+4a_1k_1(a_2+c_2)(y+\ell_2)-2a_1c_1(x^2+\ell_0^2+\ell_1^2)-4b_1c_1\,x\ell_0-2(a_1^2+b_1^2+c_1^2)\,x\ell_1-4a_1b_1\,\ell_0 \ell_1,\\
Q_{1,LR}(\zeta)&\coloneqq 2c_1(c_0k_0-b_2k_1)\,\ell_0-2b_1(c_0k_0-b_2k_1)\,\ell_1\\
&\qquad+4b_1k_1(a_2+c_2)(y+\ell_2)-2b_1c_1(x^2+\ell_0^2+\ell_1^2)-4a_1c_1\,x\ell_0-4a_1b_1\,x\ell_1-2(a_1^2+b_1^2+c_1^2)\,\ell_0 \ell_1\,.
\end{align*} 
For $M\in\{\emptyset,L,R,LR\}$, the single-flip equations at $i+2$ are
\begin{equation*}
    3Y_{i+2,M} = \beta_M d_3 + Q_{2,M}(\zeta)
\end{equation*}
with
\begin{equation*}
    (\beta_\emptyset, \beta_L, \beta_R, \beta_{LR}) = 4k_2(a_3 + c_3)\cdot (c_2, 0, a_2, b_2)
\end{equation*}
and
  \begin{align*}
Q_{2,\emptyset}(\zeta)&\coloneqq 4b_2k_1(a_1+b_1)(x+\ell_0)-2b_2(b_3k_2-c_1k_1)\,\ell_1-2a_2(b_3k_2+c_1k_1)\,y\\
&\qquad+2c_2(b_3k_2-c_1k_1)\,\ell_2-(a_2^2+b_2^2+c_2^2)(y^2+\ell_1^2+\ell_2^2)-4a_2b_2\,y\ell_1-4a_2c_2\,y\ell_2-4b_2c_2\,\ell_1\ell_2,\\
Q_{2,L}(\zeta)&\coloneqq 4a_2k_1(a_1+b_1)(x+\ell_0)-2a_2(b_3k_2-c_1k_1)\,\ell_1-2b_2(b_3k_2+c_1k_1)\,y\\
&\qquad-2a_2b_2(y^2+\ell_1^2+\ell_2^2)-2(a_2^2+b_2^2+c_2^2)\,y\ell_1-4b_2c_2\,y\ell_2-4a_2c_2\,\ell_1\ell_2,\\
Q_{2,R}(\zeta)&\coloneqq-2c_2(b_3k_2+c_1k_1)\,y+2a_2(b_3k_2-c_1k_1)\,\ell_2\\
&\qquad-2a_2c_2(y^2+\ell_1^2+\ell_2^2)-4b_2c_2\,y\ell_1-2(a_2^2+b_2^2+c_2^2)\,y\ell_2-4a_2b_2\,\ell_1\ell_2,\\
Q_{2,LR}(\zeta)&\coloneqq 4c_2k_1(a_1+b_1)(x+\ell_0)-2c_2(b_3k_2-c_1k_1)\,\ell_1+2b_2(b_3k_2-c_1k_1)\,\ell_2\\
&\qquad-2b_2c_2(y^2+\ell_1^2+\ell_2^2)-4a_2c_2\,y\ell_1-4a_2b_2\,y\ell_2-2(a_2^2+b_2^2+c_2^2)\,\ell_1\ell_2\,.
\end{align*} 

We first use the constant-character equations to solve for the two endpoint combinations $d_0$ and $d_3$ in terms of $\zeta$. This requires $\alpha_\emptyset\beta_\emptyset\ne0$, which holds on the conditioning event introduced below. Explicitly,
  \begin{equation*}
    d_0 = \frac{3Y_{i+1,\emptyset} - Q_{1,\emptyset}(\zeta)}{\alpha_\emptyset}\,, \qquad d_3 = \frac{3Y_{i+2,\emptyset} - Q_{2,\emptyset}(\zeta)}{\beta_\emptyset}\,.
\end{equation*} 

Substituting these expressions for $d_0$ and $d_3$ into the remaining single-flip equations leaves equations involving only the five components of $\zeta$. For $M\in\{L,R,LR\}$, it is convenient to write these equations as
\begin{align}
\begin{split}
G_{1,M}(\zeta) &\coloneqq \alpha_\emptyset Q_{1,M}(\zeta)-\alpha_M Q_{1,\emptyset}(\zeta)=3\alpha_\emptyset Y_{i+1,M} - 3\alpha_M Y_{i+1,\emptyset} \label{eq:Gs}\\
G_{2,M}(\zeta) &\coloneqq \beta_\emptyset Q_{2,M}(\zeta)-\beta_M Q_{2,\emptyset}(\zeta)=3\beta_\emptyset Y_{i+2,M} - 3\beta_M Y_{i+2,\emptyset}\,, \qquad M\in\{L,R,LR\}\,. 
\end{split}
\end{align} 

We next use transitions that flip the two middle sites of the window. Set $z_{i+1} = z_{i+2}=0$, vary the two outer bits $z_i$ and $z_{i+3}$ over their four possible values, and set all remaining input bits to zero. For each such input, let $z'=z\oplus\mathbf 1_{\{i+1,i+2\}}$. As above, subtract the $K$-only contribution from the fourth-order transition coefficient, and take its four Fourier coefficients with respect to the signs $L=(-1)^{z_i}$ and $R=(-1)^{z_{i+3}}$. We denote these coefficients by $Y_{\{i+1,i+2\},M}$ for $M\in\{\emptyset,L,R,LR\}$, where $M=\emptyset$ denotes the constant Fourier character. The resulting equations are
  \begin{align}
&3Y_{\{i+1,i+2\},\emptyset}=-k_1\Bigl(k_1\bigl((x+y)^2+\ell_0^2+\ell_2^2\bigr)+2b_1b_2\,\ell_0+2(a_1b_2-a_2c_1)(x-y)+4(b_2c_1-a_1a_2)\,\ell_1+2c_1c_2\,\ell_2\Bigr),\nonumber\\
&3Y_{\{i+1,i+2\},L}=-2k_1\Bigl(k_1\,\ell_0(x+y)+(a_1b_2-a_2c_1)\,\ell_0+b_1b_2(x-y)-2a_2b_1\,\ell_1\Bigr),\nonumber\\
&3Y_{\{i+1,i+2\},R}=-2k_1\Bigl(k_1\,\ell_2(x+y)-c_1c_2(x-y)-2a_1c_2\,\ell_1-(a_1b_2-a_2c_1)\,\ell_2\Bigr),\nonumber\\
&3Y_{\{i+1,i+2\},LR}=-2k_1\Bigl(k_1\,\ell_0 \ell_2-c_1c_2\,\ell_0-2b_1c_2\,\ell_1-b_1b_2\,\ell_2\Bigr)\,. \label{eq:Y12s}
\end{align} 

The six single-flip equations in \eqref{eq:Gs} and the four pair-flip equations in \eqref{eq:Y12s} therefore give ten polynomial constraints on the five unknown diagonal coefficients collected in $\zeta$. All coefficients and right-hand sides of these equations are determined by the measured fourth-order data and the off-diagonal parameters reconstructed in the previous stage. It remains to show that, for generic Hamiltonian coefficients, these ten constraints determine $\zeta$ uniquely. We establish this using the polynomial-system certificate of Lemma~\ref{lem:exactcertificate}.

To apply Lemma~\ref{lem:exactcertificate}, we linearize these ten polynomial constraints by treating the relevant monomials in $\zeta=(x,y,\ell_0,\ell_1,\ell_2)$ as independent coordinates. Let $\mathcal S$ be the set of all monomials of total degree at most $3$ in these five variables, so $|\mathcal S|=\binom{5+3}{3}=56$.

Let $P_1,\ldots,P_{10}$ denote the ten constraint polynomials, ordered as $G_{1,L},G_{1,R},G_{1,LR},G_{2,L},G_{2,R},G_{2,LR}$, followed by the four pair-flip constraints in the order in which they appear in \eqref{eq:Y12s}. If $\zeta^*$ denotes the true value of the five unknown coefficients, define $E_j(\zeta)=P_j(\zeta)-P_j(\zeta^*)$. Each $E_j$ vanishes at the true solution.

Following Appendix~\ref{sec:certificate}, for each $j\in[9]$ we use the six polynomials $E_j,xE_j,yE_j,\ell_0E_j,\ell_1E_j,\ell_2E_j$. Together with $E_{10}$, these give $55$ polynomial equations. Appending the normalization equation for the constant monomial gives $56$ linear equations in the $56$ monomials in $\mathcal S$. Their coefficient matrix is the $56\times56$ matrix $\mathcal M$ appearing in Lemma~\ref{lem:exactcertificate}.

It remains to show that $\det\mathcal M$ is not the zero polynomial in the underlying Hamiltonian coefficients. For this purpose, we evaluate $\mathcal M$ at the following explicit choice of the already reconstructed off-diagonal coefficients and of the true diagonal tuple $\zeta^*$:
\begin{align*}
\begin{split}
(a_0,a_1,a_2,a_3)=(1,2,3,4),\quad (b_0,b_1,b_2,b_3)=(5,6,7,8),\\
(c_0,c_1,c_2,c_3)=(9,10,11,12),\quad (k_0,k_1,k_2)=(2,3,5),\\
\zeta^* = (1,2,3,4,5)\,.
\end{split}
\end{align*}
With columns ordered by total degree then lexicographically in the exponent tuple $(x,y,\ell_0,\ell_1,\ell_2)$, the resulting $56\times56$ coefficient matrix has determinant $765893$ modulo $1000003$ at this specialization. This nonzero residue proves that the determinant is a nonzero polynomial in the original Hamiltonian coefficients. Whenever the determinant is nonzero at the true coefficients, the certificate uniquely identifies $\zeta$, including among complex candidates.

The same certificate also applies to windows that meet an open boundary. To see this, modify the witness above only by setting $b_0=c_3=0$, corresponding to the absence of the two outside edges. At this specialization, the determinant of $\mathcal M$ is $181025$ modulo $1000003$, and is therefore nonzero. The same specialization is valid when only the left outside edge is absent, only the right outside edge is absent, or both are absent, so the certificate determinant is a nonzero polynomial for each type of boundary window as well. Notice also that the middle-site transition equations do not involve either outside $XX$ coupling. The missing outside diagonal couplings are set to zero, so for the first window of an open chain $d_0=h_Z^1$, while for the last window $d_3=h_Z^n$. In particular, when the chain has length four, both outside edges are absent and the same argument applies without modification.

Once $\zeta$ has been reconstructed, the constant-character equations determine the two endpoint combinations $d_0$ and $d_3$. The corresponding divisors are $\alpha_\emptyset=4k_0(a_0+b_0)b_1$ and $\beta_\emptyset=4k_2(a_3+c_3)c_2$. Both are nonzero degree-three polynomials in the Hamiltonian coefficients, so we include their nonvanishing in the conditioning event. Suppose that the reconstructed off-diagonal coefficients, the subtracted fourth-order data, and the reconstructed value of $\zeta$ each have error at most $\eta$. The numerators and divisors in the formulas for $d_0$ and $d_3$ then have error at most $\poly(R)\eta$. If these errors are at most $\rho/2$, where $\rho$ is a lower bound on the magnitudes of both divisors, standard stability of division gives an error at most $\poly(R,\rho^{-1})\eta$ in $d_0$ and $d_3$, using that their true values are bounded by $2R$. We have therefore reconstructed all seven effective diagonal variables associated with the window.

We use the same globally aligned representative of the off-diagonal coefficients from Appendix~\ref{sec:local_z_common} in every window, in particular when computing the $K$-only subtractions. On the conditioning event, the polynomial system for each window has a unique diagonal solution in this fixed representative, so there is no additional local time-reversal ambiguity that must be aligned between different windows.

For an open chain, apply the reconstruction to every window $i=1,\ldots,n-3$. Every nonendpoint field appears as one of the two middle-site fields in at least one window, and every $ZZ$ coupling appears as an internal coupling in at least one window. The two endpoint fields are obtained from $d_0$ in the first window and $d_3$ in the last window, since the corresponding outside couplings are absent. For a periodic chain, every field appears as a middle-site field in some window. When a coefficient is reconstructed in more than one window, we retain any one of the local estimates. In particular, numerical estimates are not propagated from one window to the next, so the reconstruction error does not accumulate along the chain.

It remains to control all of these local reconstructions simultaneously. For each of the $O(n)$ windows, include in the conditioning event the off-diagonal conditioning polynomials from Appendix~\ref{sec:local_z_common}, the appropriate interior or boundary certificate determinant, and the two divisors $\alpha_\emptyset$ and $\beta_\emptyset$. Each of these quantities is a nonzero polynomial in a fixed-size set of the original independent Hamiltonian coefficients. Applying Lemma~\ref{lem:gaussian_polynomial} with failure probability of order $\delta/n$ to each such polynomial, and then taking a union bound over all windows and their finitely many local gauge choices, gives simultaneous lower bounds on all required conditioning quantities. On the intersection of this event with the coefficient-radius event, we may take $\rho^{-1}\le\poly(n,R,\varsigma^{-1},\delta^{-1})$.

Corollary~\ref{cor:certificate_partial_reconstruction}, Lemma~\ref{lem:finite_stage_stability}, and the stability estimate for the final divisions then imply that the error in every reconstructed coefficient is at most $\poly(n,R,\varsigma^{-1},\delta^{-1})\eta_{\rm data}$. We therefore choose $\eta_{\rm data}$ to be $\epsilon$ divided by this polynomial, and small enough to satisfy the stability threshold at every stage. For $0<\epsilon<1$, these additional threshold requirements change the precision only by another conditioning polynomial. Finally, applying the coarse search and refinement procedure of Appendix~\ref{sec:certificate_refinement} to each constant-size polynomial system, together with the derivative-estimation bounds of Appendix~\ref{sec:defer_derivative_estimation}, gives the stated soft $\epsilon^{-2}$ bounds on the number of measurements, total evolution time, and classical post-processing time.

\end{proof}

\subsection{Spatially varying square lattices}\label{sec:proof_2D_local_Z_basis}

\begin{proof}[Proof of Theorem~\ref{thm:local_NN_z_to_zp} for square lattices]

Consider an open rectangle with both side lengths at least $3$, or the periodic square lattice of Theorem~\ref{thm:2D_local_z_to_zp}. As in the one-dimensional case, we first recover and globally align the off-diagonal coefficients using the nonadjacent-pair tests of Appendix~\ref{sec:local_z_common}. It remains to reconstruct the diagonal $Z$ fields and $ZZ$ couplings.

The reconstruction is local. Fix a vertex $r$ with $q$ actual neighbors, where $q\in\{2,3,4\}$ on an open rectangle and $q=4$ on a periodic lattice. We label these neighbors by $r_\bullet$, with $\bullet$ ordered according to $(\rightarrow,\leftarrow,\uparrow,\downarrow)$ and absent directions omitted. Since distinct neighbors of a vertex on the square lattice are not adjacent to one another, the single-flip equations around $r$ have a particularly simple local structure.

For the already reconstructed off-diagonal coefficients, write $a=h_X^r$, let $b_\bullet$ be the coefficient of $X_rZ_{r_\bullet}$, $c_\bullet$ the coefficient of $Z_rX_{r_\bullet}$, and $k_\bullet$ the coefficient of $X_rX_{r_\bullet}$. We also define
\begin{align*}
A_\bullet=h_X^{r_\bullet}+
\sum_{v\sim r_\bullet:\,v\ne r}
\text{coefficient of }X_{r_\bullet}Z_v\,.
\end{align*}
All of these quantities are known from the off-diagonal reconstruction. The quantity $A_\bullet$ collects the transverse field at the neighboring site $r_\bullet$ together with the mixed couplings from $r_\bullet$ to its other neighbors, whose input $Z$ signs will be fixed to $+1$ below.

For a degree-four interior vertex, as many as $21$ diagonal coefficients can enter the single-flip data: the field at $r$, the four $ZZ$ couplings incident to $r$, the four fields at its neighbors, and the $ZZ$ couplings joining those neighbors to sites outside this local star. However, the coefficients on the neighbor side enter only through one sum for each neighbor. The data therefore depend on only nine effective variables: the center field, four incident $ZZ$ couplings, and four neighbor-side sums. More generally, a degree-$q$ center has $2q+1$ effective variables.

Relative to the center $r$, write $d=h_Z^r$, let $e_\bullet$ denote the coefficient of $Z_rZ_{r_\bullet}$, and define
\begin{align*}
S_\bullet=h_Z^{r_\bullet}+
\sum_{v\sim r_\bullet:\,v\ne r}
\text{coefficient of }Z_{r_\bullet}Z_v\,.
\end{align*}
We collect the center field and incident couplings as $\zeta=(d,(e_\bullet)_{\bullet\in N(r)})$, and the neighbor-side sums as $S=(S_\bullet)_{\bullet\in N(r)}$. The goal is first to determine $\zeta$ while eliminating the auxiliary variables $S_\bullet$, and then to recover the latter if needed. We apply this reconstruction first at interior and straight-boundary vertices. Once their fields and incident couplings are known, each corner has only its own field left to determine.

\paragraph{Diagonal reconstruction.}
Consider the transition that flips only the center site, $z'=z\oplus\mathbf 1_{\{r\}}$. We fix $z_r=0$, vary the $q$ input bits on the actual neighbors of $r$, and set every other input bit in the local region to zero. Thus the $Z$ signs outside the varied neighbor set are all $+1$, which is why the diagonal coefficients on the far side of $r_\bullet$ enter only through $S_\bullet$.

Using Eq.~\eqref{eq:fourth-order-expansion}, we subtract the part of the fourth-order transition coefficient that depends only on the already known off-diagonal Hamiltonian $K$. The resulting quantity is
\begin{align*}
[t^4]p_{z\to z'}(t)-\frac{1}{4}\left((K^2)_{z',z}\right)^2+\frac{1}{3}\langle z'|K|z\rangle (K^3)_{z',z}\,.
\end{align*}
Here $[t^4]p_{z\to z'}(t)=f(4,z,z')/24$. The derivative data together with the reconstructed off-diagonal coefficients determine this subtracted quantity for every assignment of the neighboring input bits.

We extract the dependence of this subtracted fourth-order transition coefficient on the neighboring input bits by taking its Boolean Fourier coefficients with respect to those bits. For a subset $M$ of the actual neighbors of $r$, let $Y_{r,M}$ denote the average over all $2^q$ assignments of the neighboring bits, weighted by the character $(-1)^{\sum_{v\in M}z_v}$ and normalized by $2^{-q}$. The local implementation is justified by Eq.~\eqref{eq:local_transition_bridge}. We write $Y_{r,r_\bullet}$ for a singleton mask $M=\{r_\bullet\}$ and $Y_{r,r_\bullet,r_\circ}$ for a pair mask $M=\{r_\bullet,r_\circ\}$.

For the constant Fourier character $M=\emptyset$, we obtain
\begin{align*}
3Y_{r,\emptyset}&=-\Bigl(a^2+\sum_\bullet b_\bullet^2\Bigr)\Bigl(d^2+\sum_\bullet e_\bullet^2\Bigr)+2\Bigl(a^2d^2+\sum_\bullet b_\bullet^2e_\bullet^2\Bigr)-2\Bigl(ad+\sum_\bullet b_\bullet e_\bullet\Bigr)^2\\
&\quad-2\Bigl(\sum_\bullet k_\bullet c_\bullet\Bigr)\Bigl(ad+\sum_\bullet b_\bullet e_\bullet\Bigr)+4\sum_\bullet k_\bullet b_\bullet\bigl(c_\bullet e_\bullet+A_\bullet S_\bullet\bigr)\,.
\end{align*}
For a singleton mask $M=\{r_\bullet\}$, we obtain
\begin{align*}
3Y_{r,r_\bullet}&=-2\Bigl(a^2+\sum_{\bullet'}b_{\bullet'}^2\Bigr)de_\bullet-2ab_\bullet\Bigl(d^2+\sum_{\bullet'}e_{\bullet'}^2\Bigr)\\
&\quad-2\bigl(ae_\bullet+b_\bullet d\bigr)\Bigl(2\sum_{\bullet'\neq\bullet}b_{\bullet'}e_{\bullet'}+\sum_{\bullet'}k_{\bullet'}c_{\bullet'}\Bigr)+4ak_\bullet\bigl(c_\bullet e_\bullet+A_\bullet S_\bullet\bigr)\,.
\end{align*}
Finally for a pair mask $M=\{r_\bullet,r_\circ\}$ with $\bullet\neq\circ$, we obtain
\begin{align*}
3Y_{r,r_\bullet,r_\circ}&=-2\Bigl(a^2+\sum_{\bullet'}b_{\bullet'}^2\Bigr)e_\bullet e_\circ-2b_\bullet b_\circ\Bigl(d^2+\sum_{\bullet'}e_{\bullet'}^2\Bigr)\\
&\quad-2\bigl(b_\bullet e_\circ+b_\circ e_\bullet\bigr)\Bigl(2ad+2\sum_{\bullet'\neq\bullet,\circ}b_{\bullet'}e_{\bullet'}+\sum_{\bullet'}k_{\bullet'}c_{\bullet'}\Bigr)\\
&\quad+4k_\bullet b_\circ\bigl(c_\bullet e_\bullet+A_\bullet S_\bullet\bigr)+4k_\circ b_\bullet\bigl(c_\circ e_\circ+A_\circ S_\circ\bigr)\,.
\end{align*}

All sums in these equations run over the actual neighbors of $r$. The measured quantities $Y_{r,M}$ and all coefficients $a,b_\bullet,c_\bullet,k_\bullet,A_\bullet$ are known. The remaining unknowns are the center field $d$, the incident couplings $e_\bullet$, and the neighbor-side sums $S_\bullet$. We next combine these Fourier equations so that the $S_\bullet$ cancel, leaving a polynomial system only in the center field and its incident $ZZ$ couplings.

First define
\begin{equation*}
\begin{aligned}
    G_\emptyset &= 3aY_{r,\emptyset}-3\sum_\bullet b_\bullet Y_{r,r_\bullet}\\
    &=\Bigl(\sum_\bullet b_\bullet^2-a^2\Bigr)\Bigl(a\bigl(d^2+\sum_\bullet e_\bullet^2\bigr)+2d\sum_\bullet\bigl(b_\bullet e_\bullet+k_\bullet c_\bullet\bigr)\Bigr)+2\sum_{\bullet\neq\circ}b_\bullet b_\circ e_\bullet\bigl(ae_\circ+2b_\circ d\bigr).
\end{aligned}
\end{equation*}
Here the final sum is over ordered pairs of distinct neighbors $\bullet\ne\circ$. For each unordered pair $\{\bullet,\circ\}$ of distinct neighbors, we similarly form
\begin{equation*}
\begin{aligned}
    G_{\bullet,\circ} &= 3aY_{r,r_\bullet,r_\circ}-3b_\bullet Y_{r,r_\circ}-3b_\circ Y_{r,r_\bullet}\\
    &=2ab_\bullet b_\circ\Bigl(d^2+\sum_{\bullet'}e_{\bullet'}^2\Bigr)+2a\Bigl(b_\bullet^2+b_\circ^2-a^2-\sum_{\bullet'\neq\bullet,\circ}b_{\bullet'}^2\Bigr)e_\bullet e_\circ\\
    &\quad+2\Bigl(\sum_{\bullet'}b_{\bullet'}^2-a^2\Bigr)d\bigl(b_\bullet e_\circ+b_\circ e_\bullet\bigr)+4b_\bullet b_\circ d\Bigl(b_\bullet e_\bullet+b_\circ e_\circ+2\sum_{\bullet'\neq\bullet,\circ}b_{\bullet'}e_{\bullet'}+\sum_{\bullet'}k_{\bullet'}c_{\bullet'}\Bigr).
\end{aligned}
\end{equation*}
By construction, these combinations no longer depend on the auxiliary variables $S_\bullet$. They therefore give $1+\binom q2$ quadratic equations involving only the $q+1$ unknowns in $\zeta=(d,(e_\bullet)_{\bullet\in N(r)})$. For an interior vertex, where $q=4$, this gives seven equations for the five unknowns $\zeta=(d,e_\rightarrow,e_\leftarrow,e_\uparrow,e_\downarrow)$. It remains to show that these equations determine the five unknowns uniquely for generic Hamiltonian coefficients.

To apply Lemma~\ref{lem:exactcertificate}, let $P_1,\ldots,P_7$ denote the seven constraint polynomials, ordered as $G_\emptyset$ followed by the six $G_{\bullet,\circ}$ for unordered pairs of neighbors. If $\zeta^*$ denotes the true value of the five unknown coefficients, define $E_j(\zeta)=P_j(\zeta)-P_j(\zeta^*)$. We linearize the resulting polynomial system by treating monomials in $\zeta$ as independent coordinates. Let $\mathcal S$ consist of all $\binom{5+4}{4}=126$ monomials of total degree at most $4$ in the five variables. For each $j\in[7]$, multiply $E_j$ by every monomial $\zeta^\alpha$ of degree $|\alpha|\le2$. There are $21=\binom{5+2}{2}$ such multipliers, giving $147$ polynomial rows in total. Appending the normalization row for the constant monomial gives a $148\times126$ coefficient matrix.

To show that this matrix has full column rank for generic coefficients, it suffices to exhibit one specialization with a nonzero $126\times126$ minor. We use
\begin{equation*}
    a = 1\,, \qquad b = (1,2,3,4)\,, \qquad \sum_\bullet k_\bullet c_\bullet = 5\,,
\end{equation*}
together with
\begin{equation*}
    \zeta^* = (1,2,3,4,5)\,.
\end{equation*}
Order the seven constraints as $G_\emptyset$ followed by the six unordered pairs in lexicographic order of $(\rightarrow,\leftarrow,\uparrow,\downarrow)$. Also order columns by increasing total degree, breaking ties by increasing lexicographic order of the exponent tuple $(d,e_\rightarrow,e_\leftarrow,e_\uparrow,e_\downarrow)$. Further order the rows first by constraint and then by multiplier in the same monomial order, with the constant row last. Perform Gaussian elimination modulo $1000003$, scanning columns from left to right. At each pivot position, swap into that position the first remaining row with a nonzero entry in the current column, and eliminate below it. Track the original row labels through all swaps. The $126$ original pivot rows, restored to increasing original-label order, give a determinant $604465$ modulo $1000003$. This is sufficient to specify a nonzero certificate minor. The specialization $\sum_\bullet k_\bullet c_\bullet=5$ is realizable by the original independent Hamiltonian coefficients, so after substituting back into those coefficients the determinant is still a nonzero polynomial. Lemma~\ref{lem:exactcertificate} therefore implies that, whenever this determinant is nonzero at the true Hamiltonian, the seven equations uniquely determine the center field and its four incident $ZZ$ couplings.

Once $\zeta$ is known, the original singleton-mask equations recover the auxiliary sums $S_\bullet$. Indeed, the coefficient of $S_\bullet$ in the corresponding singleton equation is $4ak_\bullet A_\bullet$, so this recovery is possible whenever $4ak_\bullet A_\bullet\ne0$. We include these divisors in the conditioning event below.

\paragraph{Straight sides: degree three.}
At a straight boundary site, the center has three neighbors, so the combinations above give the four equations $G_\emptyset,G_{12},G_{13},G_{23}$ in the four variables $\zeta=(d,e_1,e_2,e_3)$. For the algebraic certificate used here, we supplement these equations with one additional Fourier coefficient involving all three neighboring signs. Define
\begin{equation}\label{eq:side_triple_mask}
H(d,e)=\sum_{s\in\{\pm1\}^3}s_1s_2s_3
\left(a+\sum_{j=1}^3b_js_j\right)^2
\left(d+\sum_{j=1}^3e_js_j\right)^2
=-24Y_{r,N(r)}.
\end{equation}
To see why this quantity is available from the measured data, write $\mathcal A(s)=a+\sum_{j=1}^3b_js_j$ and $\mathcal E(s)=d+\sum_{j=1}^3e_js_j$. After subtracting the $K$-only contribution, the fourth derivative is $-8\mathcal A(s)^2\mathcal E(s)^2$ plus a term that is linear in the diagonal parameters and has degree at most two in the three neighbor signs. The latter degree bound follows because each contributing two-step path uses one neighbor and its incident edge, while distinct neighbors of $r$ are not adjacent. Taking the Fourier coefficient of degree three therefore removes the entire term linear in the diagonal parameters, leaving Eq.~\eqref{eq:side_triple_mask}. The factor $-24$ comes from converting the fourth derivative to its $t^4$ coefficient and averaging over the eight assignments of the three neighboring bits.

The five polynomials $G_\emptyset,G_{12},G_{13},G_{23},H$ are quadratic in $\zeta=(d,e_1,e_2,e_3)$. Let $P_j$ denote these five constraint polynomials and define $E_j(\zeta)=P_j(\zeta)-P_j(\zeta^*)$. To certify uniqueness, evaluate the resulting system at
\begin{align*}
a=3\,,\qquad b=(2,4,7)\,,\qquad \sum_{j=1}^3k_jc_j=5\,,\qquad \zeta^*=(1,2,3,4)\,.
\end{align*}
Use all $126=\binom{4+5}{5}$ monomials of degree at most $5$ as columns and the $175=5\binom{4+3}{3}$ rows $\zeta^\alpha E_j$ with $|\alpha|\le3$, followed by the constant row. Order the constraints as $G_\emptyset,G_{12},G_{13},G_{23},H$, and use the same monomial ordering and row-swap selection rule, with exponent tuple $(d,e_1,e_2,e_3)$. The resulting $176\times126$ matrix has a selected full-size minor equal to $37124$ modulo $1000003$. Hence the determinant of this submatrix is a nonzero polynomial in the independent Hamiltonian coefficients at a degree-three boundary center. Whenever it is nonzero at the true coefficients, Lemma~\ref{lem:exactcertificate} uniquely determines the center field and its three incident $ZZ$ couplings.

\paragraph{Corners: degree two.}
We treat the corners only after reconstructing every interior and straight-boundary center. Every $ZZ$ edge incident to a corner also has a noncorner endpoint, and each neighbor of a corner lies on a straight side and has degree three. Consequently, by the time we reach a corner, its two incident couplings $e_1,e_2$ are already known. The neighbor-side sums $S_1,S_2$ are also known from the reconstructed neighbor fields and their remaining incident couplings. Thus the only diagonal coefficient still unknown at the corner is its field $d$.

Write $B=b_1^2+b_2^2$, $E=b_1e_1+b_2e_2$, $C=k_1c_1+k_2c_2$, and denote the polynomial right-hand sides of $3Y_{r,\emptyset}$ and $3Y_{r,r_1}$ by $P_0(d)$ and $P_1(d)$. These are quadratic polynomials in the single remaining unknown $d$, with quadratic coefficients $-(a^2+B)$ and $-2ab_1$, respectively. Taking the following linear combination cancels the quadratic term:
\begin{equation*}
\begin{aligned}
\mathcal R&=3(a^2+B)Y_{r,r_1}-6ab_1Y_{r,\emptyset},\\
(a^2+B)P_1(d)-2ab_1P_0(d)&=L_1d+\mathcal R_0,\\
L_1&=-2(a^2+B)^2e_1-2(a^2+B)b_1(2b_2e_2+C)\\
   &\qquad+8a^2b_1E+4a^2b_1C.
\end{aligned}
\end{equation*}
Here $\mathcal R_0$ is the known value of $(a^2+B)P_1(d)-2ab_1P_0(d)$ at $d=0$. Thus the corner field is determined linearly provided $L_1\ne0$. To verify that this condition is generic, take the realizable assignment $a=2$, $b=(1,0)$, $e=(1,0)$, and $C=0$, for which $L_1=-18$. Hence $L_1$ is a nonzero polynomial in the original Hamiltonian coefficients, and generically $d=(\mathcal R-\mathcal R_0)/L_1$.

We use the same globally aligned off-diagonal representative throughout the reconstruction. On a periodic lattice, every center is degree four and is handled by the interior certificate. On an open rectangle, we first reconstruct all interior and straight-boundary centers and retain one estimate for each shared $ZZ$ coupling; we then recover the four corner fields using the linear equation above. The corner step adds only one further substitution stage, independent of the lattice size. In this way every $Z$ field and $ZZ$ coupling is recovered, with no additional diagonal gauge that must be aligned between different local reconstructions.

It remains to control the effect of noisy derivative data and previously reconstructed off-diagonal coefficients. If the off-diagonal estimates have error at most $\eta_{\rm off} \le 1$, the quartic $K$-only subtractions have error at most $C_0R^3\eta_{\rm off}$ for a fixed constant $C_0$, while the Fourier combinations and coefficients of the local polynomial systems have error at most $\poly(R)(\eta_{\rm off}+\eta_{\rm data})$. The two certificate determinants above are fixed nonzero polynomials after substitution into the original independent Hamiltonian coefficients, even though $\sum_\bullet k_\bullet c_\bullet$ itself is not an independent Gaussian coordinate. We also include in the conditioning event the auxiliary divisors $4ak_\bullet A_\bullet$, the corner coefficients $L_1$, and the off-diagonal conditioning polynomials from Appendix~\ref{sec:local_z_common}.

Corollary~\ref{cor:certificate_partial_reconstruction} controls the degree-four and degree-three polynomial solves, including perturbations in both the reconstructed off-diagonal parameters and the subtracted right-hand sides. Lemmas~\ref{lem:gaussian_polynomial} and~\ref{lem:finite_stage_stability} control the conditioning and error propagation in the auxiliary and corner divisions. Taking a union bound over all centers and their finitely many local gauge choices gives, simultaneously for every retained field and coupling estimate, an error of at most $\poly(n,R,\varsigma^{-1},\delta^{-1})\eta_{\rm data}$. Choosing $\eta_{\rm data}$ to be $\epsilon$ divided by this polynomial, and also below the required conditioning thresholds, gives final error at most $\epsilon$. There are only $O(n)$ constant-size local systems, so the coarse search and refinement procedure preserves the soft $\epsilon^{-2}$ cost of derivative estimation.

\end{proof}

\section{Certification of identifiability of polynomial systems}
\label{sec:certificate}

In several versions of the $Z$-basis reconstruction procedure we can solve part of a polynomial system in the Hamiltonian coefficients \textit{by hand}, but are left with a residual polynomial system in the remaining unknowns.  Various coefficients of the residual polynomial system can depend on solved coefficients.  

Because the solved coefficients and the measured right-hand sides of the residual system are known only approximately, we need to know that (i) the true values of the remaining unknowns are its only solution, (ii) every approximate solution of the noisy system lies close to these true values, and (iii) an approximate solution can be computed without the classical cost growing faster than the $\epsilon^{-2}$ measurement cost.  A full-rank Jacobian at the true solution gives neither (i) nor (ii), since it rules out other solutions nearby but says nothing about distant ones.  In this Appendix, we see how the three desiderata can be satisfied.  Below, the certificate of \Cref{lem:exactcertificate} establishes (i), \Cref{thm:quantitative} and \Cref{cor:certificate_partial_reconstruction} establish (ii), and \Cref{prop:coarse_refinement} establishes (iii).

\subsection{Exact certification}

Let $\theta \in \mathbb R^{\mathsf U}$ collect the quantities solved in earlier stages, and let $t \in \mathbb R^{\mathsf V}$ collect the true values of the remaining unknowns. In this Appendix, $t$ denotes a parameter tuple as opposed to an evolution time. In the isotropic models of \Cref{sec:equation_system_TIRI_allzero}, for example, $\theta = (S_{\mathfrak D},J_{XX})$ and $t = (h_Z,J_{ZZ},J_{XZ})$.

In the noiseless problem, the fourth-order data give us the values $y = P_\theta(t)$ of $\mathsf W$ polynomials evaluated at the true unknowns $t$, where each entry $P_{\theta,j}(x)$ of $P_\theta : \mathbb C^{\mathsf V} \to \mathbb C^{\mathsf W}$ is a polynomial in $\theta$ and $x$ with rational coefficients. To reconstruct $t$, we solve $P_\theta(x) = y$ for a candidate $x$, which we allow to be complex so that the uniqueness statements below hold over $\mathbb C^{\mathsf V}$. We write $\deg(P)$ for the maximum total degree of the entries of $P$ in $(\theta,x)$, and $\deg_x$ for the degree in $x$ alone. It is useful to define the residuals
\begin{equation*}
E_j(\theta,t,x) \triangleq P_{\theta,j}(x) - P_{\theta,j}(t)\,,\qquad 1 \le j \le \mathsf W\,,
\end{equation*}
which vanish at $x = t$.

To show that $x = t$ is the only solution of $P_\theta(x) = y$, we are effectively turning this polynomial system into a linear one. Each residual $E_j$ is a linear combination of monomials $x^\alpha$ with coefficients that are polynomials in $\theta$ and $t$. If we treat each monomial as an independent unknown, the equation $E_j = 0$ becomes a linear equation in the vector of monomials. There are typically more monomials than equations, so we generate further equations by multiplication. If $E_j(x) = 0$, then $x^\alpha E_j(x) = 0$ for every monomial $x^\alpha$, and so each such product is again linear in the monomials even though it involves monomials of higher degree.

To write these products as a matrix equation, fix a set $\mathcal S$ of monomials in $x$ that contains $1,x_1,\ldots,x_{\mathsf V}$, let $\mathsf S = |\mathcal S|$, and let $v(x) \in \mathbb C^{\mathsf S}$ be the vector of these monomials evaluated at $x$. Then choose $\mathsf S - 1$ products $x^{\alpha_r}E_{j_r}$, so that together with the normalization there are as many equations as unknowns.\footnote{In applications we will sometimes first form a rectangular matrix containing additional product rows and the normalization row. We then select a nonzero full-size minor $\mathcal M$ at a fixed witness, retain all columns, and fix the selected original row indices once and for all.   Every product row annihilates $v(t)$, so a nonzero full-size minor necessarily contains the normalization row, which we put last.} Each product must involve only monomials in $\mathcal S$, so that it is a linear function of $v(x)$. In practice we take $\mathcal S$ to be all monomials of degree at most some $L$ and initially allow every multiplier with $|\alpha_r| \le L - \deg_x E_{j_r}$. The indices $j_r \in [\mathsf W]$ may repeat (see e.g.~the worked example below), and we write $\mathsf D = \max_r |\alpha_r|$ for the largest multiplier degree. Let $\mathcal M(\theta,t) \in \mathbb Q[\theta,t]^{\mathsf S \times \mathsf S}$ be the matrix whose $r$th row, for $r < \mathsf S$, contains the coefficients of $x^{\alpha_r}E_{j_r}$, so that the product of this row with $v(x)$ is $x^{\alpha_r}E_{j_r}(\theta,t,x)$. The last row has a $1$ in the column of the constant monomial and zeros elsewhere, and its product with $v(x)$ is $1$. Then
\begin{equation*}
\mathcal M(\theta,t)\,v(x) = \begin{pmatrix} x^{\alpha_1}E_{j_1}(\theta,t,x) \\ \vdots \\ x^{\alpha_{\mathsf S - 1}}E_{j_{\mathsf S - 1}}(\theta,t,x) \\ 1 \end{pmatrix}
\end{equation*}
for every $x$. If $x$ solves $P_\theta(x) = y$, every residual vanishes at $x$, and the right-hand side reduces to the standard basis vector $e_{\mathsf S} = (0\,0\,\cdots\,0\,1)$. The residuals vanish at $x = t$ by construction, so the same holds there. Hence $v(x)$ and $v(t)$ both solve the linear system
\begin{equation*}
\mathcal M(\theta,t)\,w = e_{\mathsf S}
\end{equation*}
for $w \in \mathbb C^{\mathsf S}$. This system has a unique solution exactly when
\begin{equation*}
\Delta(\theta,t) \triangleq \det\mathcal M(\theta,t)
\end{equation*}
is nonzero. In that case $v(x) = v(t)$, and comparing the entries for $x_1,\ldots,x_{\mathsf V}$ gives $x = t$. We call $\Delta$ a \emph{certificate} when it is not the zero polynomial.

Let us work through an example to make things concrete.  Suppose there is a single remaining unknown, so that $\mathsf V = 1$ and both $t$ and the candidate $x$ are numbers.  Further suppose there are no $\theta$'s, and take $P(x) = (x^2,x^3)$, so that the data are $y = (t^2,t^3)$. Now take $\mathcal S = \{1,x,x^2,x^3\}$ so that $v(x) = (1,x,x^2,x^3)^\top$. The residuals $E_1 = x^2 - t^2$ and $E_2 = x^3 - t^3$, together with the normalization, give three linear equations in the four entries of $v(x)$. The product $xE_1 = x^3 - t^2 x$ supplies a fourth, and it is the only row that involves the entry $x$. With rows ordered as $E_1$, $E_2$, $xE_1$, and taking into account the normalization, we have
\begin{equation*}
\mathcal M(t) = \begin{pmatrix} -t^2 & 0 & 1 & 0 \\ -t^3 & 0 & 0 & 1 \\ 0 & -t^2 & 0 & 1 \\ 1 & 0 & 0 & 0 \end{pmatrix}\,,
\end{equation*}
which has determinant $t^2$. For $t \ne 0$ the linear system therefore determines $v(x)$ uniquely, and with it $x = t$. We can make this explicit by combining rows so as to isolate the entry $x$. Subtracting $xE_1$ from $E_2$ cancels the $x^3$ terms, leaving only the entries $1$ and $x$, and gives the polynomial identity
\begin{equation*}
t^2(x - t) = E_2 - xE_1\,,
\end{equation*}
which is valid for every $x$. If $x$ solves $P(x) = y$, then $E_1 = E_2 = 0$, so $t^2(x - t) = 0$ and $x = t$. If $x$ solves $P(x) = y$ only approximately, so that $E_1$ and $E_2$ are small but nonzero, the identity gives $|x - t| \le (|E_2| + |x|\,|E_1|)/t^2$. The factor $t^2$ multiplying $x - t$ is $\det\mathcal M(t)$, and when it is small, a small residual permits a large error. \Cref{lem:exactcertificate} below shows that every linearized system yields an identity of this form.

\begin{lemma}\label{lem:exactcertificate}
There is a matrix $C(\theta,t,x) \in \mathbb Q[\theta,t,x]^{\mathsf V \times \mathsf W}$ such that
\begin{equation}\label{eq:certificate_identity}
\Delta(\theta,t)\,(x - t) = C(\theta,t,x)\,\bigl(P_\theta(x) - P_\theta(t)\bigr)
\end{equation}
as polynomials in $(\theta,t,x)$. The entries of $C$ have degree at most $\mathsf D$ in $x$ and total degree at most $\mathsf D + (\mathsf S - 1)\deg(P)$. In particular, if $\Delta(\theta,t) \ne 0$, then $x = t$ is the only solution of $P_\theta(x) = P_\theta(t)$ in $\mathbb C^{\mathsf V}$.
\end{lemma}

\begin{proof}
Let $v(x) \in \mathbb C^{\mathsf S}$ list the monomials of $\mathcal S$ evaluated at $x$, in the column order of $\mathcal M$. Let $E = (E_1,\ldots,E_{\mathsf W})^\top$, and let $B(x)$ be the $\mathsf S \times \mathsf W$ matrix whose $r$th row, for $r < \mathsf S$, has $x^{\alpha_r}$ in column $j_r$ and zeros elsewhere, and whose last row vanishes. For $r < \mathsf S$, the $r$th entry of $\mathcal M v(x)$ is $x^{\alpha_r}E_{j_r}$, and the last entry is $1$, so
\begin{equation*}
\mathcal M(\theta,t)\,v(x) = B(x)\,E(\theta,t,x) + e_{\mathsf S}\,.
\end{equation*}
Every residual vanishes at $x = t$, so $\mathcal M v(t) = e_{\mathsf S}$. We subtract the two equations and multiply by the adjugate, which satisfies $\operatorname{adj}(\mathcal M)\,\mathcal M = \Delta\, I$. This gives
\begin{equation*}
\Delta(\theta,t)\bigl(v(x) - v(t)\bigr) = \operatorname{adj}(\mathcal M)\,B(x)\,E(\theta,t,x)\,.
\end{equation*}
Let $\Pi \in \{0,1\}^{\mathsf V \times \mathsf S}$ select the entries of $v$ corresponding to $x_1,\ldots,x_{\mathsf V}$, so that $\Pi v(x) = x$. Applying $\Pi$ proves \eqref{eq:certificate_identity} with $C = \Pi\operatorname{adj}(\mathcal M)B$.

For the degree bounds, the coefficient of a nonconstant monomial in $E_j$ is a coefficient of $P_{\theta,j}$, of degree at most $\deg(P)$ in $\theta$. The constant coefficient $P_{\theta,j}(0) - P_{\theta,j}(t)$ has degree at most $\deg(P)$ in $(\theta,t)$. Multiplication by $x^{\alpha_r}$ shifts monomials without changing coefficients, so every entry of $\mathcal M$ has degree at most $\deg(P)$ in $(\theta,t)$. Each entry of $\operatorname{adj}(\mathcal M)$ is a signed minor of size $\mathsf S - 1$, of degree at most $(\mathsf S - 1)\deg(P)$, and the entries of $B(x)$ are monomials of degree at most $\mathsf D$. The bounds on $C$ follow. Finally, if $\Delta(\theta,t) \ne 0$ and $P_\theta(x) = P_\theta(t)$, the right-hand side of \eqref{eq:certificate_identity} vanishes, so $x = t$.
\end{proof}

\Cref{lem:exactcertificate} is useful only when $\Delta$ is not the zero polynomial. In our applications $\mathcal M$ has up to $165$ columns with polynomial entries, so computing $\Delta$ symbolically is impractical. We instead evaluate $\mathcal M(\theta,t)$ at a single rational point $(\theta,t)$, which we call a witness. If $\Delta$ is nonzero at the witness, then $\Delta$ is not the zero polynomial. We compute the determinant at the witness modulo the prime $p = 1000003$. If $p$ divides none of the denominators of the entries, a nonzero residue modulo $p$ implies a nonzero rational determinant, and hence $\Delta \not\equiv 0$.

When $\Delta(\theta,t) \ne 0$, dividing \eqref{eq:certificate_identity} by it gives $x - t = C(\theta,t,x)\,\bigl(P_\theta(x) - y\bigr)/\Delta(\theta,t)$. The error of a candidate $x$ is therefore controlled by how well it satisfies the equations, through the factor $\|C(\theta,t,x)\|/|\Delta(\theta,t)|$. This holds for every candidate, regardless of how close it is to $t$. To have a useful error bound, $|\Delta(\theta,t)|$ must not be too small, and $C(\theta,t,x)$ must stay bounded over the candidates $x$ we consider. The next subsection shows that $C(\theta,t,x)$ is bounded when $\theta$, $t$, and $x$ lie in a bounded box, and that $|\Delta(\theta,t)|$ is bounded below with high probability when the Hamiltonian coefficients are smoothed as in \Cref{def:smoothed_hamiltonians}.

\subsection{Quantitative bounds}

In practice we know $\theta$ and $y$ only approximately, so we must distinguish their true values from our estimates. From now on we write $\theta^*$ and $t^*$ for the true values of the solved quantities and of the remaining unknowns, and $y^* = P_{\theta^*}(t^*)$ for the noiseless data. So $\theta$ and $t$ now denote the variables of \Cref{lem:exactcertificate}. We make the estimates $\hat\theta$ of $\theta^*$ and $\hat y$ of $y^*$, and we accept a candidate $\hat x$ when the computable residual $P_{\hat\theta}(\hat x) - \hat y$ is small. We now show that every accepted candidate $\hat x$ lies close to $t^*$. The proof uses the certificate $\Delta$ and the matrix $C$ evaluated at the unknown true values $(\theta^*,t^*)$. The reconstruction algorithm cannot compute $\Delta(\theta^*,t^*)$ or $C(\theta^*,t^*,\hat x)$, and it does not need to. To find and accept a candidate, the algorithm only evaluates $P_{\hat\theta}$ and, for the refinement in \Cref{sec:certificate_refinement}, the Jacobian of $P_{\hat\theta}$.

Our argument has a deterministic part and a probabilistic part. \Cref{thm:quantitative} assumes that $|\Delta(\theta^*,t^*)|$ is bounded below and that $\theta^*$, $t^*$, and the candidates lie in a known box. Under these assumptions, the identity \eqref{eq:certificate_identity} converts the computable residual of any candidate into a bound on its distance from $t^*$. \Cref{cor:certificate_partial_reconstruction} then shows that $|\Delta(\theta^*,t^*)|$ is bounded below with high probability for Hamiltonians smoothed as in \Cref{def:smoothed_hamiltonians}. Proving this lower bound takes a step beyond \Cref{lem:gaussian_polynomial}. That lemma applies to a polynomial whose variables are independent Gaussians, but the entries of $\theta^*$ and $t^*$ need not be independent, because they can be combinations of the same smoothed Hamiltonian coefficients. For example in the isotropic models, one entry of $\theta^*$ is $S_{\mathfrak D} = h_X + 2\mathfrak D J_{XZ}$ and one entry of $t^*$ is $J_{XZ}$, so the two share the Gaussian perturbation of $J_{XZ}$. Instead of applying \Cref{lem:gaussian_polynomial} to $\Delta$ as a polynomial in the dependent entries of $\theta^*$ and $t^*$, the corollary writes $\Delta(\theta^*,t^*)$ as a polynomial in the original smoothed coefficients, which \textit{are} independent, and applies the lemma to that polynomial.

\begin{theorem}\label{thm:quantitative}
Let $\bar R \ge 1$ and $\Delta_0 > 0$. Suppose the true values satisfy $\|\theta^*\|_\infty \le \bar R$, $\|t^*\|_\infty \le \bar R$, and $|\Delta(\theta^*,t^*)| \ge \Delta_0$, and let $y^* = P_{\theta^*}(t^*)$. Let $\hat\theta$ and $\hat y$ be estimates with
\begin{equation*}
\|\hat\theta - \theta^*\|_\infty \le \eta_\theta \le 1\,,\qquad \|\hat y - y^*\|_\infty \le \eta_y\,.
\end{equation*}
There is a constant $K \ge 1$, depending only on $P$ and the chosen products, such that every $\hat x \in \mathbb C^{\mathsf V}$ with $\|\hat x\|_\infty \le 2\bar R$ satisfies
\begin{equation}\label{eq:certificate_error_bound}
\|\hat x - t^*\|_\infty \le \frac{K\bar R^{\mathsf D + \mathsf S\deg(P)}}{\Delta_0}\Bigl(\|P_{\hat\theta}(\hat x) - \hat y\|_\infty + \eta_y + \eta_\theta\Bigr)\,.
\end{equation}
Moreover, for any $\eta_{\rm res}>0$ satisfying $\eta_{\rm res} \ge 2(\eta_y + K\bar R^{\deg(P)}\eta_\theta)$, a grid search over $[-\bar R,\bar R]^{\mathsf V}$ that evaluates $P_{\hat\theta}$ at no more than $\bigl(2 + 2K\bar R^{\deg(P) + 1}/\eta_{\rm res}\bigr)^{\mathsf V}$ points finds a real $\hat x$ with $\|\hat x\|_\infty \le \bar R$ and $\|P_{\hat\theta}(\hat x) - \hat y\|_\infty \le \eta_{\rm res}$.
\end{theorem}

\begin{proof}
Every point at which we evaluate $P$ or $C$ below has all coordinates of magnitude at most $2\bar R$. For $\theta^*$, $t^*$, and $\hat x$ this bound is part of the hypotheses. For $\hat\theta$ it follows from $\|\hat\theta\|_\infty \le \|\theta^*\|_\infty + \eta_\theta \le \bar R + 1 \le 2\bar R$. We call the set of such points `the box'. In the box, every monomial of degree at most $d$ has magnitude at most $(2\bar R)^d$, since $2\bar R \ge 1$. A fixed polynomial of degree $d$ is therefore bounded on the box by a constant times $\bar R^d$, where the constant depends only on the polynomial, and the same holds for its first derivatives. The box is convex, so any two points in it are joined by a segment that stays in the box. Bounding the derivative along that segment gives
\begin{equation}\label{eq:certificate_lipschitz}
\|P_\theta(x) - P_{\theta'}(x)\|_\infty \le K\bar R^{\deg(P)}\|\theta - \theta'\|_\infty\,,\qquad
\|P_\theta(x) - P_\theta(x')\|_\infty \le K\bar R^{\deg(P)}\|x - x'\|_\infty\,,
\end{equation}
and the degree bound of \Cref{lem:exactcertificate} gives $\|C(\theta^*,t^*,\hat x)\|_{\infty\to\infty} \le K\bar R^{\mathsf D + (\mathsf S - 1)\deg(P)}$, where $\|\cdot\|_{\infty\to\infty}$ is the maximum absolute row sum. We enlarge $K$ finitely many times below, each time by a factor that depends only on $P$ and the chosen products.

Decompose the residual of the true system as
\begin{equation*}
P_{\theta^*}(\hat x) - P_{\theta^*}(t^*) = \bigl(P_{\theta^*}(\hat x) - P_{\hat\theta}(\hat x)\bigr) + \bigl(P_{\hat\theta}(\hat x) - \hat y\bigr) + \bigl(\hat y - y^*\bigr)\,.
\end{equation*}
By \eqref{eq:certificate_lipschitz} and the hypotheses, the three terms on the right are bounded in norm by $K\bar R^{\deg(P)}\eta_\theta$, $\|P_{\hat\theta}(\hat x) - \hat y\|_\infty$, and $\eta_y$, respectively. Thus we have
\begin{equation*}
\|P_{\theta^*}(\hat x) - P_{\theta^*}(t^*)\|_\infty \le K\bar R^{\deg(P)}\eta_\theta + \|P_{\hat\theta}(\hat x) - \hat y\|_\infty + \eta_y\,.
\end{equation*}
Evaluating \eqref{eq:certificate_identity} at $(\theta^*,t^*,\hat x)$ and dividing by $\Delta(\theta^*,t^*)$ gives
\begin{equation*}
\hat x - t^* = \frac{C(\theta^*,t^*,\hat x)\bigl(P_{\theta^*}(\hat x) - P_{\theta^*}(t^*)\bigr)}{\Delta(\theta^*,t^*)}\,.
\end{equation*}
The bound on $C$ and $|\Delta(\theta^*,t^*)| \ge \Delta_0$ then give \eqref{eq:certificate_error_bound}.

For the grid search, $t^*$ is real and lies in $[-\bar R,\bar R]^{\mathsf V}$. Take a grid in this cube with $\ell^\infty$ covering radius $\gamma = \eta_{\rm res}/(2K\bar R^{\deg(P)})$, which needs at most $2 + \bar R/\gamma$ points per coordinate. Some grid point $x'$ satisfies $\|x' - t^*\|_\infty \le \gamma$. By \eqref{eq:certificate_lipschitz},
\begin{equation*}
\|P_{\hat\theta}(x') - \hat y\|_\infty \le K\bar R^{\deg(P)}\gamma + K\bar R^{\deg(P)}\eta_\theta + \eta_y \le \eta_{\rm res}\,,
\end{equation*}
where the three terms come from moving $x'$ to $t^*$, $\hat\theta$ to $\theta^*$, and $y^*$ to $\hat y$. The search evaluates $\|P_{\hat\theta}(x) - \hat y\|_\infty$ at every grid point $x$ and returns any grid point at which this residual is at most $\eta_{\rm res}$. The grid point $x'$ above shows that one exists.
\end{proof}

\Cref{thm:quantitative} takes the box radius $\bar R$ and the lower bound $\Delta_0$ as given. For smoothed Hamiltonians, both follow from properties of the original coefficients. A bound on their size bounds $\theta^*$ and $t^*$, and \Cref{lem:gaussian_polynomial}, applied to $\Delta$ written as a polynomial in these coefficients, bounds $|\Delta(\theta^*,t^*)|$ from below.

\begin{corollary}[Certification after partial reconstruction]\label{cor:certificate_partial_reconstruction}
Let $\params^* = \overline\params + \varsigma Z$, with $Z \sim \mathcal N(0,I)$, be the vector of smoothed Hamiltonian coefficients of Definition~\ref{def:smoothed_hamiltonians}. Suppose the solved quantities and the remaining unknowns are fixed polynomial functions of these coefficients, $\theta^* = \Theta(\params^*)$ and $t^* = T(\params^*)$, and that the certificate written as a function of the coefficients,
\begin{equation*}
D(\params) \triangleq \Delta\bigl(\Theta(\params),T(\params)\bigr)\,,
\end{equation*}
is not the zero polynomial. Let $s = \deg D$. Choose $R \ge 1$ with $\Pr[\|\params^*\|_\infty > R] \le \delta/2$, for example the effective radius \eqref{eq:effective_radius}, and let $\bar R(R) \ge 1$ be a polynomial in $R$ that bounds $\|\Theta(\params)\|_\infty$ and $\|T(\params)\|_\infty$ whenever $\|\params\|_\infty \le R$. Then there is an event of probability at least $1 - \delta$, depending only on $Z$, on which the hypotheses of \Cref{thm:quantitative} on $\theta^*$ and $t^*$ hold with
\begin{equation*}
\bar R = \bar R(R)\,,\qquad \Delta_0 = c_D\Bigl(\frac{\delta\varsigma}{2}\Bigr)^s\,,
\end{equation*}
where $c_D > 0$ is the constant of \Cref{lem:gaussian_polynomial} for $D$. Consequently all conclusions of \Cref{thm:quantitative} hold on this event. In particular, for all estimates $\hat\theta$ and $\hat y$ satisfying the error bounds of \Cref{thm:quantitative} and every $\hat x$ with $\|\hat x\|_\infty \le 2\bar R$,
\begin{equation*}
\|\hat x - t^*\|_\infty \le \frac{Q(R)}{(\delta\varsigma)^s}\Bigl(\|P_{\hat\theta}(\hat x) - \hat y\|_\infty + \eta_y + \eta_\theta\Bigr)\,,
\end{equation*}
where $Q$ is a polynomial that depends only on $P$, the chosen products, $\Theta$, and $T$. If $s = 0$, then $D$ is a nonzero constant, $\Delta_0 = |D|$, and the factor $(\delta\varsigma)^{-s}$ is absent.
\end{corollary}

\begin{proof}
We construct the event in the statement as the intersection of two events, each determined by $Z$. Event~1 is the radius event $\|\params^*\|_\infty \le R$, which has probability at least $1 - \delta/2$ by the choice of $R$. When Event~1 holds, $\|\theta^*\|_\infty \le \bar R(R)$ and $\|t^*\|_\infty \le \bar R(R)$. Event~2 is $|D(\params^*)| \ge \Delta_0$. Since $D$ is a nonzero polynomial in the independently smoothed coordinates $\params^*$, \Cref{lem:gaussian_polynomial} with $\delta' = \delta/2$ shows that Event~2 has probability at least $1 - \delta/2$. If $s = 0$, Event~2 holds with probability one. When both events hold, which happens with probability at least $1 - \delta$, we have $|\Delta(\theta^*,t^*)| = |D(\params^*)| \ge \Delta_0$, so the hypotheses of \Cref{thm:quantitative} on $\theta^*$ and $t^*$ hold. The error bound \eqref{eq:certificate_error_bound} of \Cref{thm:quantitative} then gives the displayed bound with $Q(R) = 2^s K\bar R(R)^{\mathsf D + \mathsf S\deg(P)}/c_D$. Neither event depends on the estimates, so the bound holds for every pair of estimates obeying the error bounds, including estimates computed adaptively in earlier noisy stages.
\end{proof}

We note that two special cases clarify the hypothesis $D \not\equiv 0$. If each entry of $\theta^*$ and $t^*$ is itself one of the smoothed coefficients, then $\Theta$ and $T$ are coordinate projections, and $D$ is simply $\Delta$ regarded as a function of $\params$. More generally, if $(\Theta,T)$ is an invertible linear change of variables, as in the isotropic models, then $D \not\equiv 0$ is equivalent to $\Delta \not\equiv 0$. Without such structure, $D \not\equiv 0$ can fail even when $\Delta \not\equiv 0$, because $\Delta$ might vanish at every point $(\Theta(\params),T(\params))$ that the Hamiltonian family actually produces. In the other applications we therefore choose the witness in the original coefficient space, evaluating $\mathcal M$ at $(\Theta(\params),T(\params))$ for a rational $\params$.

The $Z$-basis data determine the Hamiltonian only up to the sign flips of \Cref{thm:unobservable_gauges}, so the earlier stages fix one representative, for example by making each field $h_X^v$ positive. This choice requires one more check. The true values that the reconstruction targets are then $\Theta(g\params^*)$ and $T(g\params^*)$, where $g$ flips the signs of some coefficients. The relevant certificate is then $D \circ g$. Each sign flip is an invertible linear map, so $D \circ g \not\equiv 0$ whenever $D \not\equiv 0$. Which $g$ the reconstruction selects depends on $\params^*$, but each local patch admits only finitely many choices, so a union bound over them controls the conditioning of every representative (see \Cref{sec:constant_size_stability}).

\subsection{Coarse search and local refinement}\label{sec:certificate_refinement}

\Cref{thm:quantitative} finds a candidate $\hat x$ for the remaining unknowns $t^*$ by searching a grid for a point whose computable residual $\|P_{\hat\theta}(\hat x) - \hat y\|_\infty$ is at most $\eta_{\rm res}$. The cost of this search grows quickly with the precision. In particular, reaching parameter error $\epsilon$ through the error bound \eqref{eq:certificate_error_bound} requires $\eta_{\rm res}$ proportional to $\epsilon$, which requires a grid with spacing proportional to $\epsilon$, and hence $O(\epsilon^{-\mathsf V})$ grid points. The systems in our applications have up to five unknowns, so this classical cost can reach order $\epsilon^{-5}$, which exceeds the $\widetilde O(\epsilon^{-2})$ measurement cost. We therefore replace the fine grid by a coarse grid search followed by Newton's method. The coarse grid, whose spacing does not depend on $\epsilon$, locates a point within a fixed distance of $t^*$. Newton's method then refines this point to precision $\epsilon$. The certificate identity \eqref{eq:certificate_identity} justifies both steps. Its global error bound guarantees that the best coarse grid point lies close to $t^*$. Differentiating the identity at $t^*$ shows that the Jacobian of $P_{\theta^*}$ is well conditioned there, which makes Newton's method contract near $t^*$.

Throughout this subsection we assume the hypotheses of \Cref{thm:quantitative}, and we let $K$ be the constant of that theorem. Our analysis of Newton's method uses the smallest singular value of the Jacobian of $P_\theta$, which is defined with respect to Euclidean norms, so we use Euclidean norms instead of the maximum norms of \Cref{thm:quantitative}. Converting between the two norms costs factors that depend only on the dimensions $\mathsf U$, $\mathsf V$, and $\mathsf W$, and we absorb these into $K$. Define
\begin{equation*}
\kappa \triangleq \max\Bigl\{1,\frac{K\bar R^{\mathsf D + \mathsf S\deg(P)}}{\Delta_0}\Bigr\}\,,\qquad L \triangleq K\bar R^{\deg(P)}\,.
\end{equation*}
The constant $\kappa$ converts residuals into parameter errors, and $L$ bounds the derivatives of $P$. Applying \Cref{thm:quantitative} with exact data, $\hat\theta = \theta^*$ and $\hat y = y^*$, gives
\begin{equation}\label{eq:refinement_global}
\|x - t^*\|_2 \le \kappa\,\|P_{\theta^*}(x) - P_{\theta^*}(t^*)\|_2 \qquad \text{for } \|x\|_\infty \le 2\bar R\,.
\end{equation}

We write $DP_\theta(x)$ for the $\mathsf W \times \mathsf V$ Jacobian matrix of $P_\theta$ with respect to $x$. Differentiating \eqref{eq:certificate_identity} with respect to $x$ gives one term in which the derivative falls on $P_\theta(x) - P_\theta(t)$ and one in which it falls on $C$. At $x = t$ the second term vanishes, because $P_\theta(t) - P_\theta(t) = 0$. Evaluating at $(\theta^*,t^*)$ therefore gives $\Delta(\theta^*,t^*)\,I = C(\theta^*,t^*,t^*)\,DP_{\theta^*}(t^*)$. Thus $C(\theta^*,t^*,t^*)/\Delta(\theta^*,t^*)$ is a left inverse of $DP_{\theta^*}(t^*)$, and the bounds on $|\Delta|$ and $C$ from \Cref{thm:quantitative} give
\begin{equation}\label{eq:refinement_jacobian}
\sigma_{\min}\bigl(DP_{\theta^*}(t^*)\bigr) \ge \frac{|\Delta(\theta^*,t^*)|}{\|C(\theta^*,t^*,t^*)\|_2} \ge \frac{1}{\kappa}\,.
\end{equation}

Finally, in the box specified in \Cref{thm:quantitative}, i.e.~where every coordinate of $\theta$ and $x$ has magnitude at most $2\bar R$, the first and second derivatives of $P$ are bounded by a constant times $\bar R^{\deg(P)}$. After enlarging $K$, the constant $L$ therefore bounds the norm of $DP_\theta(x)$, the Lipschitz constant of $DP_\theta(x)$ in $x$, and the Lipschitz constants of $P_\theta(x)$ and $DP_\theta(x)$ in $\theta$.

The algorithm cannot evaluate the true residual map $F^*(x) \triangleq P_{\theta^*}(x) - y^*$, since $\theta^*$ and $y^*$ are unknown. It evaluates the computable residual map $F(x) \triangleq P_{\hat\theta}(x) - \hat y$ instead. For $x$ in the box, the difference $F(x) - F^*(x) = \bigl(P_{\hat\theta}(x) - P_{\theta^*}(x)\bigr) - (\hat y - y^*)$ has Euclidean norm at most
\begin{equation*}
\eta \triangleq \sqrt{\mathsf W}\,\eta_y + L\sqrt{\mathsf U}\,\eta_\theta\,,
\end{equation*}
by the Lipschitz bound in $\theta$ and the conversion from maximum norms to Euclidean norms. We will show that Newton's method contracts on the ball of radius
\begin{equation*}
r \triangleq \min\{1/2,\, 1/(16\kappa L)\}
\end{equation*}
around $t^*$. The bound $r \le 1/2$ keeps this ball inside the box, since $\|t^*\|_\infty \le \bar R$ and $\bar R \ge 1$. The bound $r \le 1/(16\kappa L)$ makes the quadratic term in the Newton error at most one sixteenth of the current error.

The following proposition carries out the coarse search and the Newton refinement. Since our systems have more equations than unknowns (i.e.~$\mathsf W > \mathsf V$), we use the Gauss--Newton form of Newton's method, which replaces the inverse of the Jacobian by its pseudoinverse.

\begin{proposition}\label{prop:coarse_refinement}
Assume the hypotheses of \Cref{thm:quantitative}, let $0 < \epsilon < 1$, and suppose that
\begin{equation}\label{eq:refinement_precision}
\eta \le \min\Bigl\{\frac{r}{16\kappa},\, \frac{\epsilon}{8\kappa}\Bigr\}\,.
\end{equation}
Let $x_0$ minimize $\|F\|_2$ over a grid in $[-\bar R,\bar R]^{\mathsf V}$ with Euclidean covering radius $r/(8\kappa L)$, and define the Gauss--Newton iterates
\begin{equation*}
x_{j+1} = x_j - DF(x_j)^+F(x_j)\,,
\end{equation*}
where $^+$ denotes the pseudoinverse. Then $\|x_j - t^*\|_2 \le \epsilon$ for every $j \ge \log_{16}(1/\epsilon)$. The grid has at most $\bigl(2 + 8\sqrt{\mathsf V}\,\kappa L\bar R/r\bigr)^{\mathsf V}$ points, independently of $\epsilon$.
\end{proposition}

\begin{proof}
Write $e_j = \|x_j - t^*\|_2$, and recall that the ball of radius $r$ around $t^*$ lies in the box.

Some grid point $x'$ satisfies $\|x' - t^*\|_2 \le r/(8\kappa L)$. Since $L$ bounds the Jacobian and $F^*(t^*) = 0$, we have $\|F^*(x')\|_2 \le r/(8\kappa)$, and therefore $\|F(x')\|_2 \le r/(8\kappa) + \eta$. The minimizer $x_0$ satisfies the same bound, so $\|F^*(x_0)\|_2 \le r/(8\kappa) + 2\eta$. By \eqref{eq:refinement_global} and \eqref{eq:refinement_precision},
\begin{equation*}
e_0 \le \kappa\Bigl(\frac{r}{8\kappa} + 2\eta\Bigr) \le \frac{r}{4}\,.
\end{equation*}
The grid count follows from using at most $2 + \sqrt{\mathsf V}\,\bar R\cdot 8\kappa L/r$ points per coordinate.

Now suppose $e_j \le r$, and let $A = DF(x_j) = DP_{\hat\theta}(x_j)$. It differs from $DP_{\theta^*}(t^*)$ by at most $L\sqrt{\mathsf U}\,\eta_\theta + Lr$. The first term is at most $\eta \le 1/(32\kappa)$ and the second is at most $1/(16\kappa)$. By \eqref{eq:refinement_jacobian} and Weyl's inequality, $\sigma_{\min}(A) \ge 1/(2\kappa)$, so $A$ has full column rank, $A^+A = I$, and $\|A^+\|_2 \le 2\kappa$. Taylor expansion around $x_j$ gives
\begin{equation*}
F(x_j) = F(t^*) + A(x_j - t^*) - \rho_j\,,\qquad \|\rho_j\|_2 \le \tfrac12 L e_j^2\,,
\end{equation*}
and $\|F(t^*)\|_2 = \|F(t^*) - F^*(t^*)\|_2 \le \eta$. Substituting into the Newton step gives $x_{j+1} - t^* = A^+\bigl(\rho_j - F(t^*)\bigr)$, so
\begin{equation*}
e_{j+1} \le \kappa L e_j^2 + 2\kappa\eta \le \frac{e_j}{16} + 2\kappa\eta\,,
\end{equation*}
where the second inequality uses $\kappa L e_j \le \kappa L r \le 1/16$. Since $2\kappa\eta \le r/8$, this gives $e_{j+1} \le r/16 + r/8 < r$, and induction from $e_0 \le r/4$ keeps every iterate in the ball. Iterating the recursion,
\begin{equation*}
e_j \le 16^{-j}e_0 + \tfrac{32}{15}\kappa\eta \le \tfrac18\,16^{-j} + \tfrac{4}{15}\epsilon\,,
\end{equation*}
which is at most $\epsilon$ once $16^{-j} \le \epsilon$.
\end{proof}

The number of coarse grid points is independent of $\epsilon$, so the precision enters the classical cost of solving the system only through the $\log_{16}(1/\epsilon)$ Newton steps. Each step solves a least-squares problem of fixed size whose matrix $A$ has condition number at most $2\kappa L$, since $\sigma_{\min}(A) \ge 1/(2\kappa)$ and $\|A\|_2 \le L$. We now check that these steps can be carried out in finite precision. Set $\nu = \min\{r/(16\kappa),\,\epsilon/(8\kappa)\}$. Then $\eta \le \nu$ by \eqref{eq:refinement_precision}, and $\nu > 0$ even when the data are exact, so $\kappa\nu$ is a positive rounding tolerance. If each Newton step is computed with absolute error at most $\kappa\nu$, the recursion $e_{j+1} \le e_j/16 + 2\kappa\eta$ from the proof of \Cref{prop:coarse_refinement} becomes $e_{j+1} \le e_j/16 + 3\kappa\nu$. Since $3\kappa\nu \le 3r/16$, the iterates remain in the ball of radius $r$. Iterating gives $e_j \le 16^{-j}e_0 + 16\kappa\nu/5$, and $16\kappa\nu/5 \le 2\epsilon/5$, so the same number of steps reaches precision $\epsilon$. Standard stability bounds for least squares show that $O(\log(\kappa L\bar R/\nu))$ bits of precision per step achieve the tolerance $\kappa\nu$. The coarse search and refinement therefore use a number of arithmetic operations polynomial in $\kappa$, $L$, and $\bar R$, times $1 + \log(1/\epsilon)$. The rest of the classical processing consists of reading and averaging the experimental samples, which takes time linear in their number, and the interpolation of \Cref{thm:main_complexity4derivative_estimation}, which takes $O(d^2)$ operations. Using this solver in the applications therefore preserves the $\widetilde O(\epsilon^{-2})$ dependence of the sample count, the total evolution time, and the classical processing time.

In our applications, we use the single coefficient-radius event $\|\params^*\|_\infty\le R$ together with the individual determinant budgets $\beta$ specified below. With $\bar R=\bar R(R)$ and $\Delta_0=c_D(\beta\varsigma)^s$, both $\kappa$ and $L$ are bounded by fixed polynomials in $R$ and $(\beta\varsigma)^{-1}$. It therefore suffices to estimate the data to errors $\eta_\theta,\eta_y\le \epsilon/\poly(R,(\beta\varsigma)^{-1})$ in order to satisfy \eqref{eq:refinement_precision}.

These polynomial bounds can be made explicit without knowing $t^*$ or expanding the certificate determinant symbolically. An upper bound for $K$ follows from the coefficient sums of $P$ and elementary cofactor bounds for $\mathcal M$. In our applications, $\Theta$ and $T$ have rational coefficients. Let $q \ge 1$ be an integer clearing the coefficient denominators of $\mathcal M(\Theta(\params),T(\params))$, let $m \ge 1$ be the number of coefficient variables appearing in this matrix, and choose an integer $h\ge\max\{1,\deg D\}$. We can then take
\begin{align}
\Delta_0=q^{-\mathsf S}
\left(\frac{\min\{1,\beta\varsigma\}}{2h^2\sqrt m}\right)^h.
\end{align}
Here $q^{\mathsf S}D$ is a nonzero polynomial with integer coefficients, and the proof of Lemma~\ref{lem:gaussian_polynomial} gives this lower bound with failure probability at most $\beta$. These bounds determine $\kappa$, $L$, $r$, and the coarse grid without us needing to know $t^*$. Their dependence on $R$, $\beta^{-1}$, and $\varsigma^{-1}$ is polynomial.

One minor complication is that the isotropic reconstruction has four possible sign choices for $S_{\mathfrak D}$ and $J_{XX}$, each with its own residual system and certificate. We impose the determinant bound for all four choices, so a union bound covers whichever choice is used by the reconstruction.

The same bounds must hold simultaneously across all of the local systems in a spatially varying Hamiltonian. The reconstruction solves one constant-size local system for each window or vertex. Some of these systems are linear, while the others have the polynomial form $P_\theta(x) = y$ treated above. All of them use the radius $R$ from \eqref{eq:effective_radius}, where $N$ is the number of independently smoothed Hamiltonian coefficients, and share the total failure probability $\delta$. With probability at least $1-\delta/4$, every coefficient of $H$ has magnitude at most $R$. On this event, every coefficient of $H/R$ has magnitude at most one. Since our lattices have fixed degree, $H/R$ has a known local energy bound $g=O(1)$, so we apply \Cref{thm:main_complexity4derivative_estimation} to $H/R$.

Applying the derivative-estimation theorem to $H/R$ rescales the smoothing strength to $\varsigma/R$ and the target coefficient error to $\epsilon/R$. Evolution under $H/R$ for time $u$ is equivalent to evolution under $H$ for time $u/R$, and an $r$th time derivative for $H$ is $R^r$ times the corresponding derivative for $H/R$. Since every local reconstruction system has fixed size and degree, this normalization introduces only fixed powers of $R$ into the residual, derivative, and determinant bounds. For base coefficients of magnitude $O(1)$ and $0 < \varsigma \le 1$, we have $R = O(\sqrt{\log(N/\delta)})$, so these additional factors are polylogarithmic in $N/\delta$.

It remains to distribute the total failure probability among the radius, conditioning, and sampling events. Reserve $\delta/4$ for the common coefficient-radius event. Let $B \ge 1$ be a deterministic upper bound on the total number of conditioning tests, including all local systems and their finitely many gauge choices, and assign each test failure probability $\beta = \delta/(4B)$. Applying Lemma~\ref{lem:gaussian_polynomial} separately to these tests and taking a union bound gives all of the required conditioning bounds simultaneously except with probability at most $\delta/4$.

The remaining budget $\delta/2$ is assigned to derivative estimation. At stage $j \ge 0$, let $A_j \ge 1$ be a deterministic upper bound on the number of such calls, and assign each call failure probability at most $\frac{\delta}{2 A_j}\frac{6}{\pi^2(j+1)^2}$. The total failure probability at stage $j$ is then at most $\frac{\delta}{2}\frac{6}{\pi^2(j+1)^2}$, and summing over all stages gives at most $\delta/2$, since $\sum_{j \ge 0} \frac{6}{\pi^2(j+1)^2} =1$. At later stages, the requested precision may depend on the outcomes of earlier stages. In this case, we use fresh samples and condition on all previous outcomes. Once these outcomes are fixed, the requested precision is fixed as well, so the derivative-estimation guarantee applies with the failure probability assigned above. No independence between different local systems is required. Together with the $\delta/4$ budget for the coefficient-radius event and the $\delta/4$ budget for the conditioning events, this gives total failure probability at most $\delta$. The same allocation applies to the TI protocols; in the isotropic case, the four possible sign choices contribute only a constant number of additional conditioning tests, so $B$ remains constant.

\end{document}